\documentclass[preprint]{article}
\usepackage{graphicx}
\usepackage{amsmath}
\usepackage{arydshln}
\usepackage{amssymb}
\usepackage{rotating}
\usepackage{hyperref}
 \usepackage{amsthm}
 \usepackage{subcaption}

 \fontfamily{pbk}\selectfont

\usepackage{placeins}

\usepackage{float}
\floatstyle{plaintop}
\restylefloat{table}

\usepackage{nicematrix}

\usepackage{tikz}
\usepackage{eso-pic} 

\usepackage[table]{xcolor} 
\definecolor{lightgray}{HTML}{EFEFEF}

\usetikzlibrary{shapes.geometric, arrows.meta, positioning, calc}

\newcommand{\toprowspace}{\rule{0pt}{5ex}}
\newcommand{\midrowspace}{\rule{0pt}{3ex}}
\usepackage{tikz}           
\usetikzlibrary{tikzmark}

 \usepackage{scalerel}

 \usepackage{etoolbox}

\usepackage{booktabs} 

\newcommand{\tablefont}{\footnotesize} 

\AtBeginEnvironment{tabular}{\tablefont}

\AtBeginEnvironment{tabular*}{\tablefont}

\usepackage[dvipsnames]{xcolor}

 \usepackage{xcolor}
 \usepackage{rotating}
 \usepackage{comment}
\newtheorem{thm}{Theorem}[section]
\newtheorem{cor}{Corollary}[section]
\newtheorem{prop}{Proposition}[section]
\newtheorem{conj}{Conjecture}[section]
\newtheorem{remark}{Remark}[section]
\newcommand*{\img}[1]{%
    \raisebox{-.3\baselineskip}{%
        \includegraphics[
        height=\baselineskip,
        width=\baselineskip,
        keepaspectratio,
        ]{#1}%
    }%
}
\usepackage{wrapfig}
\graphicspath{ {./images/} }

\newcommand{\be}{\begin{equation}}
\newcommand{\ee}{\end{equation}}
\newcommand{\ba}{\begin{eqnarray}}
\newcommand{\ea}{\end{eqnarray}}

\newcommand{\al}{\alpha}

\begin{document}
\setlength\dashlinedash{.6pt}
\setlength\dashlinedash{2pt}

\hoffset=-.4truein\voffset=-0.5truein
\setlength{\textheight}{8.5 in}
\begin{titlepage}
\begin{center}
\vskip 40 mm
{\large \bf The HOMFLY--PT polynomial and HZ factorisation}

\vskip 10mm
{\bf {Andreani Petrou and Shinobu Hikami }}
\vskip 5mm

Okinawa Institute of Science and Technology Graduate University,\\
 1919-1 Tancha, Okinawa 904-0495, Japan.

\vskip 5mm
{\bf Abstract}
\vskip 3mm
\end{center}
The Harer--Zagier (HZ) transform maps the HOMFLY--PT polynomial into a rational function. For some special  knots and links, the latter admits a simple factorised form, 
which is referred to as \emph{HZ factorisation}. This property is preserved under 
full twists and  the Jucys--Murphy twists, which are hence used to generate infinite HZ-factorisable families of hyperbolic knots. For such families, the HOMFLY--PT polynomial can be fully encoded in two sets of integers, corresponding to the numerator and denominator exponents,  
which turn out to be related to the double-grading in
Khovanov homology. 
Moreover, a relation between the HOMFLY--PT and Kauffman polynomials, which was only known to hold for torus knots, is now proven for several of these hyperbolic families. Such a  relation has a peculiar implication in topological string theory, namely, it is equivalent to the vanishing of the two-crosscap BPS invariants. 
It is conjectured that the HOMFLY--PT-Kauffman relation provides a criterion for HZ   factorisability.

 \end{titlepage}

\section{Introduction}
  The well known relation between knot polynomial invariants and the 3-dimensional Chern--Simons (CS) gauge theory \cite{Witten,Jones}  has led to an increasingly fruitful interchange between pure mathematics and theoretical physics. 
 The HOMFLY--PT polynomial $H(\mathcal{K};a,z)$ \cite{HOMFLY,PT} of a knot or link $\mathcal{K}$, is a Laurent polynomial in two variables that can be defined via the skein relation
\be\label{skein}
    a^{-1}H(\img{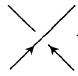})-aH(\img{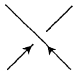})=zH(\img{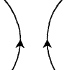}),
\ee
along with normalisation condition $H(\bigcirc)=1$. 
It corresponds to the $SU(N)$-invariant $\bar{H}(\mathcal{K};q^N,q)$ of CS theory, which  can be derived   by taking averages of Wilson loops around the knot $\mathcal{K}$. Such a quantity is related to $H(\mathcal{K};a,z)$ 
 as defined in (\ref{skein}) by 
\be\label{SU(N)invariant}
\bar{H}(\mathcal{K};q^N,q)=\frac{a-a^{-1}}{z}H(\mathcal{K};a,z)\bigg{\vert}_{\scaleto{a=q^N,\;z=q-q^{-1}}{10pt}}.
\ee
The $SU(N)$-invariant is referred to as the \emph{unnormalised HOMFLY--PT polynomial} since $\bar{H}(\bigcirc)=\frac{q^N-q^{-N}}{q-q^{-1}}$.  Here $q$ is a quantum group  parameter \cite{Jones} and can be determined by the level $k\in\mathbb{Z}$ of CS theory and the rank $N$ of the gauge group as $q=e^{-\pi i/(k+N)}$. 
When $N=2$ (i.e. $a=q^2$) and $N=0$ ($a=1$),
the HOMFLY--PT reduces  to the Jones and Alexander polynomials, respectively. The $SO(N+1)$-invariant, on the other hand, corresponds to the Kauffman polynomial.

    The Harer--Zagier (HZ) transform   was first introduced in \cite{HZ}  where it was used as a generating function for the computation of  Euler characteristics of the moduli space of Riemann surfaces, via  random matrix theory \cite{ItzyksonZuber1,Goulden}. This matrix model approach
  is useful, not only for the study of Euler characteristics, but also for computing
  the intersection numbers of the orbifold, which are the Gromov-Witten invariants for a point \cite{Penner,Kontsevich,BrezinHikami,BrezinHikami2}. 
  The  HZ transform in relation to Gaussian correlators (Gaussian means) has been considered in \cite{Morozov2}.
 
 The HZ transform has recently been applied  to knot polynomials \cite{Morozov, Petrou}.
In this context, it amounts to a discrete version of the Laplace transform applied to the parameter $N$ in the unnormalised HOMFLY--PT polynomial, 
which can be expressed as
\be\label{KhovanovC}
Z(\mathcal{K};\lambda,q) = \sum_{N=0}^\infty   \bar{H}(\mathcal{K}; a=q^N,q)\lambda^N.
\ee
At each monomial  $a^\beta$ ($\beta\in\mathbb{Z}$) in $\bar{H}$, the sum over $N$ can be evaluated via geometric series  (assuming $|\lambda q^\beta|<1$)  as
 \be\label{HZ2}
 \sum_{N=0}^\infty a^\beta \lambda^N  = \sum_{N=0}^\infty (q^{\beta}\lambda)^N = \frac{1}{1- \lambda q^\beta}.
 \ee 
 Thus the evaluation of HZ amounts to the substitution $a^\beta\mapsto(1- \lambda q^\beta)^{-1}$ in $\bar H(a,q)$ and, hence, it transforms the HOMFLY--PT polynomial into a rational function, involving the ratio of polynomials in the parameters $q$ and $\lambda$. 
The HZ transform is said to be \emph{factorisable}, if both the numerator and denominator can be expressed as the product of monomials in $\lambda$ of the form $(1\pm\lambda q^\beta)$.

  As an explicit example, consider the right handed  trefoil knot $3_1$, which has HOMFLY--PT polynomial 
 \be
 H(3_1;a,z)= - a^4 + a^2 ( z^2+2) \big\vert_{z= q-q^{-1}}= -a^4 + a^2 (q^2 + q^{-2}).
 \ee
 Its unnormalised version is obtained by multiplying with an overall factor $(a-a^{-1})/(q-q^{-1})$ and becomes $\bar {H}(3_1) =  (q - q^3)^{-1}(a(1+q^4) - a^3 (1 + q^2 + q^4) +  a^5 q^2  )$. The HZ transform is obtained by applying (\ref{KhovanovC}) 
and has the factorised form
 \be\label{3_1}
 Z(3_1;\lambda,q) = \frac{\lambda (1- q^9 \lambda)}{(1- q \lambda)(1- q^3 \lambda)(1- q^5 \lambda)}.
 \ee
  In fact,   $3_1=T(2,3)$ belongs to the family of torus knots $T(m,n)$, whose all members are known to have factorised HZ transform \cite{Morozov}. Beyond the torus family,  a hyperbolic family of pretzel knots that also enjoys this property was recently found in \cite{Petrou}.

  In the present article, we extend the results of \cite{Petrou} by finding further infinite families of hyperbolic knots which admit a factorisable HZ transform. As explained in Sec.~\ref{sec:factknots}, such knots are related to each other by full twists $F_m$ or  Jucys--Murphy twists $E_m$, examples of which are shown in Fig.~\ref{fig:FnEn}. The factorisability property seems to be preserved under these operations applied an indefinite amount of times and, hence, they are used to generate the infinite families. Furthermore, in Sec.~\ref{sec:factlinks} we consider the HZ factorisability in the case of links with two components.   
  Such knots and links are special since their HZ function  can be fully determined by two  sets of integers corresponding to the numerator and denominator exponents of $q$; and a parameter $m$, which gives a lower bound to the braid index, according to the Morton--Franks--Williams inequality. 
  We describe how the HZ transform may give an explanation for the non-sharpness of the inequality for the exceptional cases.
 
In Sec.~\ref{sec:InverseHZ},  Theorems~\ref{thmInverseHZ} and ~\ref{thmInverseAlex} suggest a way to recover the HOMFLY--PT and Alexander polynomials of a knot or link, given its HZ function. This is achieved by applying the inverse HZ transform, which is computed via  contour integrals. As an application, we provide the compact formula (\ref{FactorisedHOMFLY}) for the HOMFLY--PT polynomial of the infinite families of knots that admit HZ factorisability, which is somewhat reminiscent of the Rosso-Jones formula for torus knots \cite{RossoJones}. In (\ref{eq:factInverseAlex}) we also give a closed expression for the Alexander polynomial of the factorised cases.

   Furthermore, in Sec.~\ref{sec:HandKF} 
   we conjecture that HZ factorisation for knots occurs if and only if their HOMFLY--PT and Kauffman polynomials are related by \be
   H(\mathcal{K})=\widehat{KF}(\mathcal{K})):=\widetilde{KF}_{\rm even}(\mathcal{K};a,z)-\frac{z}{a-a^{-1}}\widetilde{KF}_{\rm odd}(\mathcal{K};a,z)\ee  (for notation see (\ref{tildeKF})), 
  hence providing a necessary and sufficient condition  for factorisability. 
  This can also be reduced to a relation (\ref{AKF}) between the Alexander and Kauffman polynomials. In fact, such relations were long known to exist for torus knots \cite{Labastida}, but verifying  their validity beyond the torus case has since been an open problem. The present article remedies this, by finding that these relations  
  hold for an exhaustive list  with up-to-12-crossings hyperbolic knots that admit  HZ factorisability and, in Theorem~\ref{thm:} -which is among the central results of this article-, we prove it for the HZ-factorisable families $P(2,-3,\pm(2j+1))$, $\mathcal{K}_{j,2}$, $\mathcal{K}_{j,3}$ and $5_2^{3k}$ (defined in Sec.~\ref{sec:factorisedHZ}). We propose a similar HOMFLY--PT/Kauffman relation (\ref{linkcriterion})  in the case of links, but this is no longer in 1-1 correspondence with HZ factorisability.  In Sec.~\ref{sub:BPS}, we further suggest that the relation $H(\mathcal{K})=\widehat{KF}(\mathcal{K})$   is equivalent to the vanishing of the two-crosscap BPS invariants $\hat{N}^{c=2}_{g,Q}$ of topological string theory, which are related to link invariants via the gauge/string duality \cite{MorozovCheck}. The BPS invariants for several HZ-factorisable two-component links are included in the Appendix.

 In Sec.~\ref{sec:Kh}, we suggest a homological interpretation of the exponents in the factorised HZ transform. This is  based on the observation that  for a given knot, the HZ exponents coincide with the  $j$-grading appearing in the rows of its  Khovanov homology (Kh) table \cite{Bar-Natan} and their appearance in the numerator or denominator is associated with the sign of the corresponding Euler characteristic \cite{Khovanov,Khovanov2,Bar-Natan1,Bar-Natan2}. 
 In fact, Prop.~\ref{propJonesCoeffpm2} shows that in cases of HZ factorisability, the Jones polynomial (the graded Euler characteristic of Khovanov homology)  encapsulates all the information about the HZ exponents.

 A summary and discussion can be found in Sec.~\ref{Sec:summary}. In a follow up paper \cite{PetrouII}, we explain  HZ-factorisability through a more rigorous approach involving represenation theory and, furthermore, we investigate the structure of the HZ function and its exponents for  more general, non HZ-factorisable knots. 

\vskip 3mm
\section{Knots and links with factorised HZ transform}\label{sec:factorisedHZ}

The HOMFLY--PT polynomial and, consequently, also its HZ transform, behave differently in the case of links with odd and even  number of components, 
due to the different parity of the powers of the HOMFLY--PT variables $a$ and $z$. 
Hence,  the factorisability properties of their HZ function 
are considered separately. 
\subsection{Knots\label{sec:factknots}}
We consider the unnormalised HOMFLY--PT polynomial $\bar{H}(\mathcal{K};q^N,q)$ of a knot  as defined in (\ref{SU(N)invariant}).
Its HZ transform (\ref{KhovanovC})  is said to be factorisable when it can be written in the form
\be\label{eq:factorisableZ}
Z(\mathcal{K};\lambda,q)=\frac{\lambda\prod_{i=0}^{m-2}(1-\lambda q^{\alpha_i})}{\prod_{i=0}^{m}(1-\lambda q^{\beta_i})},
\ee
where  
$m,\alpha_i,\beta_i\in\mathbb{Z}$ are constants that depend on $\mathcal{K}$ and satisfy $\sum_{i=0}^{m-2}\alpha_i=\sum_{i=0}^{m}\beta_i$. In fact, the denominator exponents $\{\beta_i=e+2i\}_{i=0}^m$ are fully determined by $e\in\mathbb{Z}$, which is the lowest power of $a$ in $\bar{H}$ 
and the parameter $m$. As we shall explain below, the latter mostly coincides with the braid index $\iota$.

In the previous article \cite{Petrou} we verified that the factorisability property (\ref{eq:factorisableZ}) holds for the $m$-stranded torus knots\footnote{Up to 13 crossings, torus knots are listed in \cite{KnotInfo} as $3_1=T(2,3)$, $5_1=T(2,5)$, $7_1=T(2,7)$, $8_{19}=T(3,4)$, $9_1=T(2,9)$, $10_{124}=T(3,5)$, $11a_{367}=T(2,11)$ and $13a_{4878}=T(2,13)$.} $T(m,n)$ 
and for the pretzel family\footnote{The overline notation $\overline n$ is used interchangeably  with the more standard $-n$ to indicate negative tangles.} 
$P(\overline{2},3,\overline{2j+1})$, 
with exponents\\
 {\bf $\bullet$ }$T(m,n)$; $\alpha_i=(m+1)n-m+2+2i$; $e=(m-1)n-m$ \\
  {\bf $\bullet$ }$P(\overline{2},3,\overline{2j+1})$; $m=3$;  $\alpha_0=13-2j$, $\alpha_1=3-6j$; $e=1-2j$.
  
An exhaustive list for hyperbolic knots with up to 13 crossings satisfying (\ref{eq:factorisableZ}) together with their HZ parameters is the following\footnote{We used the website KnotInfo \cite{KnotInfo} for the data of the HOMFLY--PT polynomial of knots with up to 13 crossings and its HZ transform was computed using the Mathematica software. Note that, although we often use Rolfsen notation to describe a knot $\mathcal{K}$, we sometimes refer to the mirror image $\mathcal{K^*}$ of the one  originally listed in the Rolfsen table. 
The HOMFLY--PT polynomials and HZ transforms of mirror knots are related by $H(\mathcal{K};a,z)=H(\mathcal{K^*};a^{-1},-z)$ and $Z(\mathcal{K};\lambda,q)=Z(\mathcal{K^*};\lambda,q^{-1})$ respectively, and hence conclusions about factorisation are not affected by this choice. For clarity, sometimes a superscript $\mathcal{K}^{\pm}$ will be added to the Rolfsen notation  indicating whether we refer to the positive or negative one (as indicated by the sign of their writhe), respectively, which are the mirror image of each other. Note that their positivity or negativity is also reflected in the sign of the exponents in the HZ transform.\label{ft:mirror-knot}}.\\
 {\bf $\bullet$ }$5_2=P(\overline{2},3,\overline{1})$; $m=3$; $\alpha_0=13$, $\alpha_1=3$; $e=1$\\
 {\bf $\bullet$ }$8_{20}=P(\overline{2},3,\overline{3})$;  $m=3$; $\alpha_0=11$, $\alpha_1=-3$; $e=-1$\\
 {\bf $\bullet$ }$10_{125}=P(\overline{2},3,\overline{5})$; $m=3$;  $\alpha_0=9$, $\alpha_1=-9$; $e=-3$\\
 {\bf $\bullet$ }$10_{128}$; $m=4$; $\alpha_0=7$, $\alpha_1=17$, $\alpha_2=21$; $e=5$\\
{\bf $\bullet$ }$10_{132}$; $m=2(\neq\iota=4)$; $\alpha_0=15$; $e=3$ \\
{\bf $\bullet$ }$10_{139}=:\mathcal{K}_{-1,2}$; $m=3$; $\alpha_0=15$, $\alpha_1=25$, $e=7$\\
{\bf $\bullet$ }$10_{161}=:\mathcal{K}_{0,2}$; $m=3$; $\alpha_0=9$, $\alpha_1=23$; $e=5$.\\
{\bf $\bullet$ }$12n_{235}=P(\overline{2},3,\overline{7})$; $m=3$; $\alpha_0=-15$, $\alpha_1=7$; $e=-5$ .\\
{\bf $\bullet$ }$12n_{242}=P(\overline{2},3,7)=\mathcal{K}_{-2,2}$; $m=3$; $\alpha_0=21$, $\alpha_1=27$; $e = 9$. 
\\{\bf $\bullet$ }$12n_{318}$; $m=4$; $\alpha_0=-5$, $\alpha_1=5$, $\alpha_2=15$; $e=-1$.
\\{\bf $\bullet$ }$12n_{749}=:\mathcal{K}_{1,2}$;    $m=2(\neq \iota=3)$; $\alpha_0=21$; $e=5$.

\vskip1mm

Interestingly, these knots 
can be related to one another by two braid operations, the  full twist $F_m=(\sigma_{m-1}\sigma_{m-2}\cdots\sigma_{1})^m$ and the Jucys--Murphy twist \cite{Jucys, Murphy} 
$E_m:=\sigma_{1}\sigma_{2}\cdots\sigma_{m-1}\sigma_{m-1}$ $\cdots\sigma_2\sigma_{1}$  or, equivalently (rotated by $180^\circ$), 
$\tilde{E}_m:=\sigma_{m-1}\cdots\sigma_1\sigma_1\cdots\sigma_{m-1}$, examples of which for $m=5$ are shown in Fig.~\ref{fig:FnEn}.   The relations among the above knots  under these operations are summarised in 
  Table~\ref{fig:factTable}. 

  \begin{figure}[h!]
\begin{subfigure}[h]{0.32\linewidth}
\centering
\includegraphics[width=0.59\textwidth]{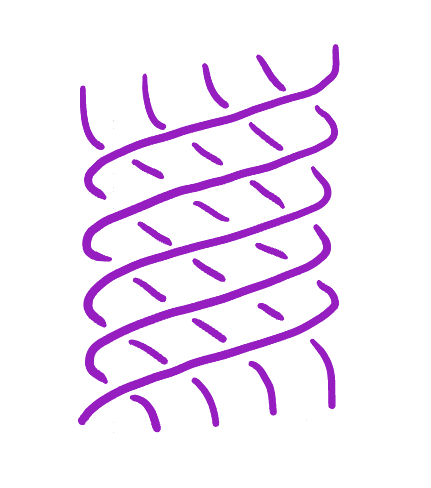}
\vspace{-1mm}
\caption{Full twist $F_5$}
\end{subfigure}
\begin{subfigure}[h]{0.325\linewidth}
\centering
\includegraphics[width=0.54\textwidth]{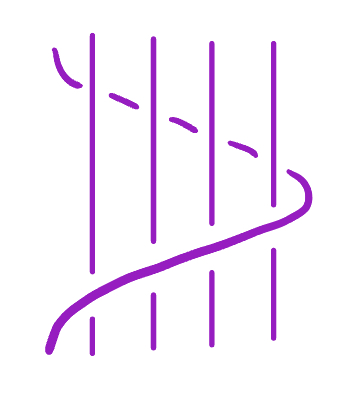}
\vspace{1mm}
\caption{Jucys--Murphy braid $E_5$ }
\end{subfigure}
\begin{subfigure}[h]{0.32\linewidth}
\centering
\includegraphics[width=0.54\textwidth,angle=180]{TwistE5.jpeg} 

\caption{Jucys--Murphy braid $\tilde E_5$ }
\end{subfigure}
\caption{5-strand examples of the braid operations that preserve the factorisability of the HZ transform. The Jucys-Murphy twist is depicted in two equivalent braid representations in (b) 
and (c), 
which are related to each other by $180^\circ$ rotation.}\label{fig:FnEn}
\end{figure}

\begingroup
  \begin{table}[h!]
  \caption{Relations between HZ-factorisable knots  with   $m=3$ (top) and $m=4$ (bottom). They can be obtained from each other in the horizontal direction by introducing a \textcolor{Purple}{Jucys--Murphy twist $E_m$}   or \textcolor{Magenta}{$\tilde{E}_m$}, or in the diagonal direction by a \textcolor{teal}{full twist $F_m$}, as indicated by the    arrows.}
      \centering
    
  \label{fig:factTable}
  
     $(m=3)\;\;$ 
     {\footnotesize
     \begin{NiceTabular}{|c|ccc|}
      \hline
    \toprowspace
     \toprowspace
         j  & \tikzmarknode{a}{$\sigma_1\otimes T(2,-2j-1)$}\tikzmarknode{g}{} %
         &$\;\;\;\;\;\;$\tikzmarknode{b}{$\mathcal{K}_{j,1}=P(\overline2,3,\overline{2j+1})$}\tikzmarknode{f}{}& \tikzmarknode{c}{$\mathcal{K}_{j,2}$} \\
         &$\sigma_1\sigma_2^{-2j-1}$  &&\\
          \rowcolor{gray!15}&&\tikzmarknode{d}{$\mathcal{K}_{j-1,1}$}& \tikzmarknode{e}{$\mathcal{K}_{j-1,2}$}\\[4pt]
          \hline
          \hline
           \midrowspace
        3& $\sigma_1\sigma_2^{-7}$& $12n_{235}^-$&
        $8_{20}\otimes 10_{125}$\\
            \cellcolor{gray!15}2& \cellcolor{gray!15}$\sigma_1\sigma_2^{-5}$& \cellcolor{gray!15}$10_{125}$&\cellcolor{gray!15}$8_{20}\otimes 8_{20}$\\
        1& $\sigma_1\sigma_2^{-3}$& $8_{20}^+$&$12n_{749}^+$\\
         \rowcolor{gray!15}0& $\sigma_1\sigma_2^{-1}$& $5_{2}^+$&$10_{161}^+$\\
       -1& $\sigma_1\sigma_2$& $5_{1}=T(2,5)$&$10_{139}^+$\\
        \rowcolor{gray!15}-2& $\sigma_1\sigma_2^3$& $8_{19}=T(3,4)$&$12n_{242}^+$\\
      -3& $\sigma_1\sigma_2^5$& $10_{124}=T(3,5)$&$T(3,7)$\\
        \hline
          
      \end{NiceTabular}
      \begin{tikzpicture}[overlay, remember picture,
                    shorten >=4pt, shorten <=4pt]
                    ]
\draw [violet,  thick, ->] (a) --node[midway, above, violet] {$E_3$} (b);
\draw [violet,  thick, ->] (b) --node[midway, above, violet] {$E_3$} (c);
\end{tikzpicture}
\begin{tikzpicture}[overlay, remember picture,
                    shorten >=4pt, shorten <=4pt]
                    ]
  \draw [teal, thick, ->] (g) -- node[midway, below, teal] {$F_3$} (d);
\draw [teal,  thick, ->] (f) --node[midway, below, teal] {$F_3$} (e);
\end{tikzpicture}}
  \end{table}
\endgroup

\begingroup
  \begin{table}[h!]
      \centering
      ($m=4$)
      \vskip1mm
  \vspace{6pt}
    \begin{tabular}{c}
\begin{minipage}{0.54\linewidth}
\centering
     \setlength{\tabcolsep}{15pt}
\renewcommand{\arraystretch}{1}
   \begin{NiceTabular}{|rc|}
       \hline
      $\sigma_3\otimes 10_{139}^-$\tikzmarknode{a}{}&\\
\scalebox{0.8}{ $\sigma_3\sigma_1^{-1}\sigma_2^{-2}\sigma_1(\sigma_2^{-1}\sigma_1^{-1})^4$}     &  \\
\rowcolor{gray!15}\tikzmarknode{b}{$\sigma_3\otimes 5_{2}^-$}\tikzmarknode{f}{}&\tikzmarknode{d}{$12n_{318}^+$}\\
\rowcolor{gray!15}\scalebox{0.8}{$\sigma_3\sigma_1^{-1}\sigma_2^{-2}\sigma_1(\sigma_2^{-1}\sigma_1^{-1})\;$}     &  \\
\rowcolor{white!15}\tikzmarknode{c}{$\sigma_3\otimes 3_{1}^+$}&\tikzmarknode{e}{$10_{128}^+$}\\
\rowcolor{white!15}\scalebox{0.8}{$\sigma_3\sigma_1^{-1}\sigma_2^{-2}\sigma_1(\sigma_2\sigma_1)$}     &  \\
\hline
      \end{NiceTabular}
          \begin{tikzpicture}[overlay, remember picture,
                    shorten >=4pt, shorten <=4pt]
                    ]
\draw [magenta,  thick, ->] (b) --node[midway, above, magenta] {$\tilde{E}_4$} (d);
\draw [magenta,  thick, ->] (c) -- (e);
\end{tikzpicture}
\begin{tikzpicture}[overlay, remember picture,
                    shorten >=4pt, shorten <=4pt]
                    ]
\draw [teal,  thick, ->] (a) --node[midway, above, teal] {$F_4$} (d);
\draw [teal,  thick, ->] (f) -- (e);
\end{tikzpicture}
      \end{minipage}
\hfill

\hspace{-6mm}
\begin{minipage}{0.45\linewidth}
\centering
 \setlength{\tabcolsep}{20pt}
\renewcommand{\arraystretch}{1}
      \begin{NiceTabular}{|rc|}
      \hline
     $\sigma_3\otimes T(3,-5)$ \tikzmarknode{a}{}&\\
\scalebox{0.8}{$\sigma_3(\sigma_2^{-1}\sigma_1^{-1})^5$}     &  \\

\rowcolor{gray!15}\tikzmarknode{b}{$\sigma_3\otimes T(3,-2)$} \tikzmarknode{f}{}&\tikzmarknode{d}{$10_{132}^+$}\\
\rowcolor{gray!15}\scalebox{0.8}{$\sigma_3(\sigma_2^{-1}\sigma_1^{-1})^2$}     &  \\
\rowcolor{white!15}\tikzmarknode{c}{$\sigma_3\otimes T(3,1)$}&\tikzmarknode{e}{$7_{1}^+$}\\
\rowcolor{white!15}\scalebox{0.8}{$\sigma_3\sigma_2\sigma_1$}     &  \\
\hline
      \end{NiceTabular}
             \begin{tikzpicture}[overlay, remember picture,
                    shorten >=4pt, shorten <=4pt]
                    ]
\draw [magenta,  thick, ->] (b) --node[midway, above, magenta] {$\tilde{E}_4$} (d);
\draw [magenta,  thick, ->] (c) -- (e);
\end{tikzpicture}
\begin{tikzpicture}[overlay, remember picture,
                    shorten >=4pt, shorten <=4pt]
                    ]
\draw [teal,  thick, ->] (a) --node[midway, above, teal] {$F_4$} (d);
\draw [teal,  thick, ->] (f) -- (e);
\end{tikzpicture}
      \end{minipage}
      \end{tabular}
  \end{table}
\endgroup
\FloatBarrier

The factorisability of the HZ transform seems to   be preserved after further applying these braid operations an arbitrary number of times. Hence infinite families of knots that are HZ-factorisable can be constructed 
as follows. Given an $m$-strand braid representative of a knot $\mathcal{K}$ that satisfies (\ref{eq:factorisableZ}), concatenation  with $F_m^k$ or $E_m^k$  for any $k\geq 0$, which will be denoted by the symbol $\otimes$,   
results into the infinite families $\mathcal{K}\otimes F_m^k=:\mathcal{K}^{mk}$ or $\mathcal{K}\otimes E_m$, respectively. 
Note that in the case of full twists, which are equivalent to  Dehn twists, the resulting family does not depend on the choice of the initial braid representative for $\mathcal{K}$.  
This is not the case, however, for the Jucys-Murphy twist, for which the resulting knots and their factorisability will  depend on the choice of initial braid representative and, hence, the latter should be specified to obtain an unambiguous definition for $\mathcal{K}\otimes E_m$ (or $\mathcal{K}\otimes \tilde E_m$). 
Several explicit examples  shall be considered below.

\vskip2mm
{\bf $\bullet$ } 
The family $5_2\otimes F_3^k=:5_2^{3k}$  is 
generated by $k$ full twists on $5_2^+$ (where the subscript denotes positive writhe, c.f. footnote~\ref{ft:mirror-knot}), which has braid index $3$ and is depicted in   Fig.~\ref{fig:5_2_3k}. This family has been previously considered in \cite{Barkowski}.
\begin{figure}[!h]
   \centering
   \includegraphics[scale=0.29]{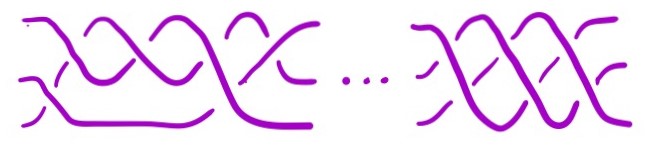}
\caption{The family $5_2^{3k}$ is the closure of the braid obtained by a braid representative for $5_2^+$ (left part) concatenated with $k$ full twists $F_3$ (one of such twists is depicted in the right part).\vspace{2mm}}
   \label{fig:5_2_3k}
\end{figure}
For instance, at $k=1$ it includes the 10-crossing knot $10_{139}^+=5_2^+\otimes F_3$. 
In turn, as suggested in  Table.~\ref{fig:factTable} ($m=3$), the knot $5_2^+$ itself can be obtained by a full twist on 
a 3-stranded version of the left handed trefoil $T(2,-3)$ 
with braid  $\sigma_1\sigma_2^{-3}$, shown in Fig.~\ref{fig:s1s23}. 
\begin{figure}[h]
    \centering
    \includegraphics[width=0.10\linewidth,angle=90]{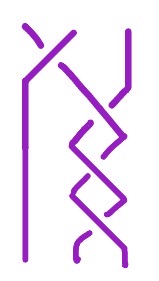}
    \caption{Braid $\sigma_1\sigma_2^3$. When a full twist $F_3$ is attached this yields a braid diagram for the knot $5_2$. }
    \label{fig:s1s23}
\end{figure} 
Hence, 
$5_2^+=\sigma_1\sigma_2^{-3}\otimes F_3$ and $5_2^{3k}=\sigma_1\sigma_2^{-3}\otimes F_3^{k+1}$.
 The HOMFLY--PT polynomial for this family  can be computed recursively, for $k\geq1$, by
\ba\label{523krec}
    H(5_{2}^{3k})&=&a^{2}(1+z^{2})\left(H_{T(3,3k+1)}+a^{4}z^{2}\sum_{i=0}^{3k-2}a^{2i}H_{T(3,3k-1-i)}
    +a^{6k}(1-v^{2})\right)\nonumber\\
    &&+a^{6}z^{2}H_{T(3,3k-1)}
\ea
and its HZ transform can be written in a symmetric way in terms of the number of crossings $n=6k+4$ as
\be\label{5_2^{3k}}
Z(5_2^{3k};\lambda,q)=\frac{\lambda\left(1-\lambda q^{2n-5}\right)\left(1-\lambda q^{2n+5}\right)}{\left(1-\lambda q^{n-3}\right)\left(1-\lambda q^{n-1}\right)\left(1-\lambda q^{n+1}\right)\left(1-\lambda q^{n+3}\right)}.
\ee
In terms of $k$ the HZ parameters become $\alpha_i(5_2^{3k})=12k+3+10i$, $e(5_2^{3k})=6k+1$  and $m=3$. 
\begin{remark}\label{rmk:additive}
\textup{Note that  for each subsequent member of the family the value of  $e$ increases by $6$, which is equal to the writhe (same as the number of crossings) added 
by a full twists $F_3$. 
Moreover, the values of the numerator exponents $\alpha_i$ also change in an additive way with each extra full twist. }   
\end{remark}

\vspace{2mm}
{\bf $\bullet$ }The whole of the pretzel family  $P(\overline{2},3,\overline{2j+1})=\sigma_1\sigma_2^{-(2j+3)}\otimes F_3$ can be further generalised into $P(\overline{2},3,\overline{2j+1})^{3k}=\sigma_1\sigma_2^{-(2j+3)}\otimes F_3^{k+1}$,   
with HZ transform\footnote{In general, it is often difficult to find  simple recursive or explicit formulas for the HOMFLY--PT polynomial  for  families generated by full twists. However, for the first few members of the family (with no more than 45 crossings), it can be computed using the Mathematica software, with the aid of the package "Knot Theory" \cite{Bar-Natan}. Although the HOMFLY--PT polynomial becomes increasingly  complicated, the   exponents in the HZ transform, which remains factorisable, show a clear pattern that follows the observation in Remark~\ref{rmk:additive}. In the sequel, the HZ formulas for such families are determined in this way, by extrapolating on the pattern exhibited by their first few members. In a follow up paper \cite{PetrouII} we will provide a more rigorous derivation of these formulas.} 
\ba\label{P(-2,3,-2j-1),3k}
Z(P(\overline{2},3,\overline{2j+1})^{3k};\lambda,q)=\frac{\lambda\left(1-\lambda q^{12k-2j+13}\right)\left(1-\lambda q^{12k-6j+3}\right)}{\prod_{i=0}^3\left(1-\lambda q^{6k-2j+1+2i}\right)}.
\ea
At $j=0$ it reduces to the family $5_2^{3k}$ mentioned above, while at $j=1$ it contains the subfamily  $8_{20}^{3k}$, 
which includes the knot $10_{161}^+=8_{20}\otimes F_3=\sigma_1\sigma_2^{-5}\otimes F_3^2$ at $k=1$. Further such relations can be traced in Table~\ref{fig:factTable} ($m=3$).

\vspace{2mm}
{\bf $\bullet$ } A different family can be generated by $k$ full twists $F_3$ on the mirror of the pretzel family $P(2,\overline{3},2j+1)$. It is denoted as $P(2,\overline{3},2j+1)^{3k}=\sigma_1^{2j+3}\sigma_2^{-1}\otimes F_3^{k+1}$ and,  at $k=1,2$, these are simply torus knots $T(2,2j+3)=P(2,\overline{3},2j+1)\otimes F_3$ and $T(2,2j+7)=P(2,\overline{3},2j+1)\otimes F_3^2$. Their 
 HZ transform reads
\be\label{3k}
Z(P(2,\bar{3},2j+1)^{3k};\lambda,q)=\frac{\lambda\left(1-\lambda q^{12k+2j-13}\right)\left(1-\lambda q^{12k+6j-3}\right)}{\prod_{i=0}^3\left(1-\lambda q^{6k+2j-7+2i}\right)}.
\ee
Note that this expression is very similar to Eq.~(\ref{P(-2,3,-2j-1),3k}) with the exponents of $q$ related to $k$ (amount of full twists) being the same, while the remaining ones  switch sign (c.f. footnote \ref{ft:mirror-knot}). 
Similarly, by applying (positive) full twists on the mirror of any of the knots satisfying (\ref{eq:factorisableZ}), 
 different families can be generated, 
 but their HZ transforms will be related in a similar way as in this case.

\vspace{2mm}
{\bf $\bullet$ } The family $P(\overline{2},3,2j+1)=\sigma_1^{2j-1}\sigma_2\otimes F_3$ 
contains the torus knots $5_1=T(2,5)$ at $j=0$, $8_{19}=T(3,4)$ at $j=1$, $10_{124}=T(3,5)$ at $j=2$ and the hyperbolic knot $12n_{242}$
 at $j=3$. The latter 
can equivalently be obtained by a full twist on a 4-stranded trefoil as $\sigma_1^3\sigma_2^{-1}\sigma_3^{-1}\otimes F_4$, or alternatively as $T(3,4)\otimes E_3:=(\sigma_2\sigma_1)^4\otimes E_3$ (c.f.  Table~\ref{fig:factTable}).  
The HOMFLY--PT polynomial of this family can be obtained  recursively via the formula
\be\label{P232j+1} H_{P(\overline{2},3,2j+1)}=v^{4}(2+z^{2}-v^{2})H_{T(2,2j+1)}+z^{2}v^{2}H_{T(2,2j+3)}+zv^{3}H_{T(2,2j+2)}
\ee
and its HZ transform, in terms of the crossing number $n=6+2j$, reads
\be\label{eq:Pmin23plus2j}
    Z(P(\overline{2},3,2j+1))=\frac{\lambda \left(1-\lambda q^{n+9}\right) \left(1-\lambda q^{3n-9}\right)}{\left(1-\lambda q^{n-3}\right) \left(1-\lambda q^{n-1}\right) \left(1-\lambda q^{n+1}\right) \left(1-\lambda q^{n+3}\right)}.
\ee
By increasing the number of full twists, this can be  generalised further to $P(\overline{2},3,2j+1)^{3k}=\sigma_1^{2j-1}\sigma_2\otimes F_3^{k+1}$  with 
HZ transform
\be\label{3kII}
Z(P(\bar{2},3,2j+1)^{3k};\lambda,q)=\frac{\lambda\left(1-\lambda q^{12k+2j+15}\right)\left(1-\lambda q^{12k+6j+9}\right)}{\prod_{i=0}^3\left(1-\lambda q^{6k+2j+3+2i}\right)}.
\ee

\vspace{2mm}
{\bf $\bullet$ }  The knot $10_{128}^+$, which has braid index $\iota=4=m$, can be obtained by a full twist on a 4-stranded version of $5_2^-=\sigma_3\sigma_1^{-1}\sigma_2^{-2}\sigma_1\sigma_2^{-1}\sigma_1^{-1}$, as  depicted in Fig.~\ref{fig:52F4}. 
\begin{figure}[h]
    \centering
\includegraphics[width=0.135\linewidth,angle=90.5]{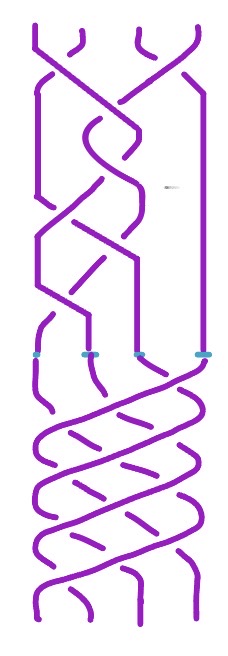}
    \caption{4-strand braid representative of $5_2$ (left) concatenated with a full twist $F_4$ (right).}
    \label{fig:52F4}
\end{figure} 
 It generates the family $10_{128}^{4k}=10_{128}\otimes F_4^{k}$ with HZ transform
\be\label{1284k}   
Z(10_{128}^{4k};\lambda,q)=\frac{\lambda\left(1-\lambda q^{20k+7}\right)\left(1-\lambda q^{20k+17}\right)\left(1-\lambda q^{20k+21}\right)}{\prod_{i=0}^4\left(1-\lambda q^{12k+5+2i}\right)}.
\ee

\vspace{2mm}
{\bf $\bullet$ } The knot $10_{132}^+$, which has braid index $\iota=4$ but  has the same HOMFLY--PT polynomial as $T(2,5)$, can be obtained by a full twist on a 4-stranded version of $T(3,-5)$ in the form $\sigma_1(\sigma_3^{-1}\sigma_2^{-1})^5\otimes F_4$. 
It generates the family $10_{132}^{4k}=\sigma_1(\sigma_3\sigma_2)^{-5}\otimes (F_4)^{k+1}$ with 
HZ transform 
\be\label{1324k}
Z(10_{132}^{4k};\lambda,q)=\frac{\lambda\left(1-\lambda q^{-1+20k}\right)\left(1-\lambda q^{1+20k}\right)\left(1-\lambda q^{15+20k}\right)}{\prod_{i=0}^4\left(1-\lambda q^{12k-1+2i}\right)}.
\ee

 \vskip2mm
Next, we turn to families obtained by concatenation with the Jucys--Murphy's braid $E_m$. As we mentioned earlier, this will be braid word depended and hence for the resulting knot $\mathcal{K}\otimes E_m^k$ to be well defined, the former 
shall be specified in each case.

\vspace{2mm}
{\bf $\bullet$ } The family $\mathcal{K}_{j,k}:=\sigma_1\sigma_2^{-2j-1}\otimes E_3^k$ 
has HZ transform
\be\label{containT211}
    Z(\mathcal{K}_{j,k})=\frac{\lambda \left(1-\lambda q^{6(k-j)-3}\right) \left(1-\lambda q^{10k-2j+3}\right)}{\prod_{i=0}^3\left(1-\lambda q^{4k-2j-3+2i}\right)}.
\ee
Note that at $k=1$ this is just the pretzel family $P(\bar2,3,\overline{2j+1})$. 
At $k=2$, it includes $10_{161}$ at $j=0$, $12n_{749}$ at $j=1$, $10_{139}$ at $j=-1$, $12n_{242}$ at $j=-2$,  
while the knot corresponding to $j=2$ can also be obtained as $8_{20}\otimes8_{20}:=(\sigma_2^{-3}\sigma_1\sigma_2^3\sigma_1)(\sigma_2^3\sigma_1\sigma_2^{-3}\sigma_1)$
 and  to $j=3$ as $8_{20}\otimes10_{125}:=(\sigma_2^{-3}\sigma_1\sigma_2^3\sigma_1)(\sigma_2^3\sigma_1\sigma_2^{-5}\sigma_1)$ (c.f.  Table~\ref{fig:factTable}). The family $\mathcal{K}_{j+1,2}$ can alternatively  be obtained via a full twist $F_4$ on a 4-stranded $T(2,-(2j+1))$ in the form $\sigma_1^{-(2j+1)}\sigma_2\sigma_3^{-1}\otimes F_4$. 
Its HOMFLY--PT polynomial has the recursive formula
\ba\label{eq:recursivejk}
H(\mathcal{K}_{j,2})&=&a^{-2}H(\mathcal{K}_{j-1,2})+z^{2}a^{-2j-2}\sum_{i=1}^{j}a^{2i}H(\mathcal{K}_{i-1,2})\nonumber\\
&&-a^{-2j+2}(1-a^{2}+z^{2})H(5_1^+)-z^2a^{-2j}H(P(\bar{2},3,3)).
\ea
At $k=3$, the recursive formula becomes
\ba\label{eq:recursivej3}
&&H(\mathcal{K}_{j,3})=a^{-2}H(\mathcal{K}_{j-1,3})+z^{2}a^{-2j-2}\sum_{i=1}^{j}a^{2i}H(\mathcal{K}_{i-1,3})\\
&&-a^{-2j+4}(1-a^{2}+z^{2})H(7_1^+)-z^2a^{-2j+2}H(P(\bar{2},3,5))-z^2a^{-2j}H(P(\bar{2},3,7)\nonumber,
\ea
but it gets much more complicated for $k\geq 4$.
For general $k$, at $j=0$ the family can also be expressed as $5_{2}\otimes E_3^{k-1}:=(\sigma_2^3\sigma_1\sigma_2^{-1}\sigma_1)(\sigma_2^2\sigma_1^2)^{k-1}$, 
within which  
$5_2\otimes E_3^2=8_{20}\otimes10_{139}:=(\sigma_2^{-3}\sigma_1\sigma_2^3\sigma_1)(\sigma_2^4\sigma_1\sigma_2^3\sigma_1^2)$ (at $k=3$) 
  has the same HOMFLY--PT polynomial as $T(2,11)$. 
Moreover, at $j=1$ this corresponds to $8_{20}\otimes E_3^{k-1}:=(\sigma_1\sigma_2^3\sigma_1\sigma_2^{-3})(E_3)^{k-1}$, 
within which  
$8_{20}\otimes E_3^2=8_{20}\otimes10_{161}:=(\sigma_2^{-3}\sigma_1\sigma_2^3\sigma_1)(\sigma_2^3\sigma_1\sigma_2^{-1}\sigma_1\sigma_2^2\sigma_1^2)$ at $k=3$  (c.f.  Table~\ref{fig:factTable}).

\vspace{2mm}
 {\bf $\bullet$ } The family $12n_{318}\otimes(\tilde{E}_4)^{k}$, which is generated by $12n_{318}=5_2^-\otimes\tilde{E}_4$ $:=(\sigma_3\sigma_1^{-1}\sigma_2^{-2}\sigma_1\sigma_2^{-1}\sigma_1^{-1})\tilde{E}_4$  
 has HZ transform
\be
    Z(12n_{318}\otimes\tilde{E}_4^{k})=\frac{\lambda  \left(1-\lambda q^{-5+8k}\right)\left(1-\lambda q^{5+8k}\right) \left(1-\lambda q^{15+14k}\right)}{\prod_{i=0}^4 \left(1-\lambda q^{6k-1+2i}\right) }.
\ee

\vspace{2mm}
 {\bf $\bullet$ } The family $T(3,s)\otimes E_3^k:=(\sigma_1\sigma_2)^{s}(E_3)^k$ when $s\mod3=1$ 
has HZ transform
\be\label{t3sE3k}
    Z(T(3,s)\otimes E_3^k)= \frac{\lambda \left(1-\lambda q^{6k+4s-1}\right) \left(1-\lambda q^{10k+4s+1}\right)}{\prod_{i=0}^3\left(1-\lambda q^{4k+2s-3+2i}\right)}.
\ee
When $s\mod 3=2$, i.e. $s=3j-1$ ($j\geq 1$), the HZ transform for $ T(3,3j-1)\otimes E_3^k$ can still be obtained from the above formula since
$T(3,3j-1)\otimes E_3^k=T(3,3j+1)\otimes E_3^{k-1}$. 
The knots $T(3,2)\otimes E_3=T(3,4)$ and $T(3,4)\otimes E_3=12n_{242}$ also belong to the pretzel family $P(\overline{2},3,2j+1)$  (c.f.~(\ref{eq:Pmin23plus2j}) and   Table~\ref{fig:factTable}).
The knot $T(3,7)\otimes E_3$ corresponds to an $18$-crossing knot, which can equivalently be expressed in terms of  
torus knots with additional strands of the form $(\sigma_2\sigma_1)^4\sigma_3^{-1}$ or $\sigma_1^7\sigma_2\sigma_3^{-1}$ 
times a full twist $F_4$.

\vspace{2mm}
 {\bf $\bullet$ } The family $T(4,s)\otimes E_4^k$ when $s\mod 4=3$ has HZ transform
\be\label{T4se4k}
    Z(T(4,s)\otimes E_4^k)=\frac{\lambda \left(1-\lambda q^{8k+5s}\right) \left(1-\lambda q^{8k+5s+2}\right)\left(1-\lambda q^{14k+5s-2}\right)}{\prod_{i=0}^4\left(1-\lambda q^{6k+3s-4+2i}\right)}.
\ee
When $s\mod 4=1$, i.e. $s=4j+1$ ($j\geq 1$), we get $T(4,4j+1)\otimes E_4^k=T(4,4j-1)\otimes E_4^{k+1}$ since $T(4,4j-1)\otimes E_4=T(4,4j+1)$.

\vspace{2mm}
 {\bf $\bullet$ } The family $T(m,m+1)\otimes E_m^k:=(\sigma_1\cdots\sigma_{m-1})^{m+1}(E_m)^k$ which has arbitrary braid index $m$  
 has HZ transform\footnote{We were only able to confirm the validity of (\ref{T56e5}) for $m\leq 5$.} 
 \ba\label{T56e5}
    &&\hspace{-14mm}Z(T(m,m+1)\otimes E_m^k;\lambda,q)=\nonumber\\
    &&\frac{\lambda\left(1-\lambda q^{2(2m-1)k+m(m+3)-1}\right) \prod_{i=0}^{m-3}\left(1-\lambda q^{2mk+m(m+1)+3+2i}\right) }{\prod_{i=0}^{m}\left(1-\lambda q^{2(m-1)k+m(m-1)-1+2i}\right) }.
\ea

\vspace{2mm}
{\bf $\bullet$ } Some members in the family of \textbf{ twisted torus knots $T(m,n,p,pk)$}, which are obtained from a torus knot $T(m,n)$ by $k$ full twists on its last $p$ strands, also have a factorised HZ transform. Depending on the relative values of $(m,n,p,k)$ these contain many different kinds of knots (see e.g. \cite{LeetwistedTorus}). For instance, they are torus knots when $p=m$ (or $p=n$), cables of torus knots when $p$ is a multiple of $n$ (or $m$) and hyperbolic when $p(\neq m)$ is not a multiple of $n$ (to guarantee hyperbolicity, it is also required that  the number of full twists is $k\geq 3$).
An example of twisted torus knots which are hyperbolic  is the family $T(3,3j+1,2,2k)$, 
 for which
\be
Z(T(3,3j+1,2,2k);\lambda,q)=\frac{\lambda \left(1-\lambda q^{2k+12j+5}\right) \left(1-\lambda q^{6k+12j+3}\right)}{\prod_{i=0}^3\left(1-\lambda q^{2k+6j-1+2i}\right) }.
\ee
Note that $T(3,3j+2,2,2k)=T(3,3j+1,2,2(k+1))$, while the particular cases  $T(3,1,2,2)=3_1$ ($j=0,k=1$)  and  $T(3,4,2,4)=12n_{242}$ ($j=1,k=2$).
Another example of twisted torus knots with factorised HZ transform is the family $T(4,2j+1,3,3k)$, members of which are depicted in Fig.~\ref{fig:T(4,2j+1,3,3k)}
\begin{figure}[h!]
    \centering
   \begin{turn}{91}
    \includegraphics[width=0.15\textwidth]{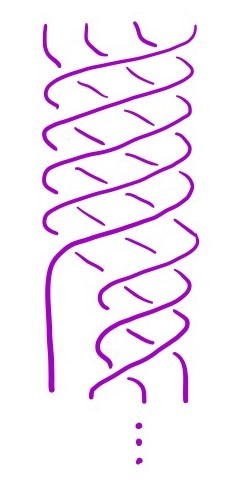}
    \end{turn}
    \caption{The twisted torus knot $T(4,5,3,3)$ is the closure of this braid. $T(4,5,3,3k)$ are obtained by inserting more 3-strand full twists at the place marked with 3 dots.}
    \label{fig:T(4,2j+1,3,3k)}
\end{figure}
and whose HZ transform for odd $j$ is
\be\label{eq:T(m,n,p,k)}
Z(T(4,2j+1,3,3k))=\frac{\lambda \left(1-\lambda q^{6k+10j+3}\right) \left(1-\lambda q^{12k+10j+5}\right)\left(1-\lambda q^{12k+10j+7}\right)}
    {\prod_{i=0}^4\left(1-\lambda q^{6(k+j)-1+2i}\right)}.
\ee 
For even $j$, the numerator exponents are changed to $\alpha_0=6k+10j+7$, $\alpha_1=12k+10j+3$ and $\alpha_2=12k+10j+5$. At $j=1$ (the factors with exponent $\alpha_0=\beta_4=6k+13$ cancel), these are just the torus knots $T(3,4+3j)$. They can be hyperbolic only when $\gcd(2j+1,3)=1$. Note that the twisted torus knots $T(5,7,2,2)$ and $T(5,7,3,3)$ have HZ transform that only partially factorises (admits quasi-factorisation), meaning that the numerator can be expressed as an overall factor $(1-\lambda q^\alpha)$ times a more general ($\lambda,q$)-polynomial.

\vskip 3mm
\subsection{Links\label{sec:factlinks}}

The unnormalised HOMFLY--PT polynomial for a link $\mathcal{L}=(\mathcal{K}_1,\mathcal{K}_2,...,\mathcal{K}_{2l})$ consisting of $2l$ components will be defined as
\be\label{SU(N)links}
\bar{H}(\mathcal{L})=\frac{a^{-1}-a}{q-q^{-1}}H(\mathcal{L}).
\ee
In particular, in this case the normalisation  has been multiplied by an overall minus sign, in order to remain consistent with our conventions\footnote{This ambiguity in the overall sign is an artifact of a discrepancy between conventions used in the literature. The fact that it can be ignored for links with odd number of components (including knots) is  because of the even parity in the powers of $a$ and $z$ that appear in $H$. } in (\ref{skein}).

\vskip2mm
 {\bf $\bullet$ } The   2-stranded torus links $T(2,2j)$ with parallel relative orientation\footnote{Throughout this paper, when we refer to torus links $T(m,n)$ we mean the links obtained as the closure of the braid $(\sigma_{m-1}\cdots\sigma_1)^n$, which naturally induces the parallel orientation on all components.} on the 2 components have a factorisable HZ transform. A simple example 
 at $j=1$ is shown in Fig.~\ref{fig:hopfLink}. This corresponds to the Hopf link $T(2,2)$, which is tabulated as $L2a1\{1\}$ in \cite{KnotInfo}, in which the notation   $\{\cdot\}$  denotes the choice of orientation. For $j>1$, the torus links with parallel orientation are tabulated as $L4a1\{1\}=T(2,4),$ $L6a3\{0\}^+=T(2,6),$ $L8a14\{0\}=T(2,8)$ and $L10a118\{0\}=T(2,10)$. 
\begin{figure}[h!]
    \centering
    \includegraphics[width=0.15\textwidth]{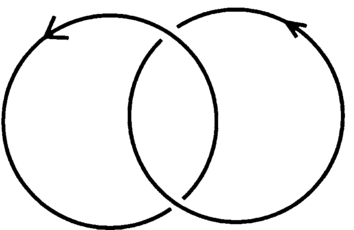}
    \caption{Hopf link $T(2,2)$ with parallel relative orientation on its components.}
    \label{fig:hopfLink}
\end{figure}
For these torus links the HZ transform reads
 \begin{equation}\label{T(2,2j)}
     Z(T(2,2j);\lambda,q)=\frac{\lambda\left(1+\lambda q^{6j}\right)}{\left(1-\lambda q^{2j-2}\right)\left(1-\lambda q^{2j}\right)\left(1-\lambda q^{2j+2}\right)}.
 \end{equation}
 The difference with the HZ function corresponding to 2-stranded torus knots $Z(T(2,2j+1))$  is the plus instead of minus sign that appears in the numerator. Note that in  the case of the Hopf link $T(2,2)$, the opposite orientation, i.e. $L2a1\{0\}$ results in the same HOMFLY--PT polynomial and hence its HZ transform still factorises, but this is not true for $j>1$. 
 The HZ transform is also  not factorisable for torus links with higher number of strands $m\geq3$, no matter the choice of relative orientation. For example,  the HZ transform of the 3-stranded torus link $T(3,3)=L6n1\{0,1\}$ 
 becomes 
 \be\label{L6n1}
 Z(L6n1\{0,1\};\lambda,q)= \frac{\lambda(1+ 2 q^{11}\lambda + 2 q^{13}\lambda + q^{24}\lambda^2)}{(1-q^3 \lambda)(1-q^5\lambda)(1-q^7\lambda)(1-q^9\lambda)}.\ee

Below, we give an exhaustive list  of non-torus links  with up to 11 crossings  (using the notation\footnote{As for knots (c.f. footnote~\ref{ft:mirror-knot}), the links referred to here might, in fact, correspond to the mirror image $\mathcal{L}^*$, as compared to the one listed in \cite{KnotInfo}. Their positivity/negativity is again denoted by a superscript $\pm$ and it is reflected in the signs of the HZ exponents.} 
of \cite{KnotInfo}) that admit a factorised HZ transform in the form
\be\label{eq:factorisableZZ}
Z(\mathcal{L};\lambda,q)=\frac{\lambda(1+\lambda q^{\alpha_0})(1-\lambda q^{\alpha_1})}{\prod_{i=0}^{3}(1-\lambda q^{e+2i})}.
\ee

\noindent{\bf $\bullet$ } $L7n1\{0\}^+
=T(3,3,2,1)=P(3,\bar2,2)$; $\alpha_0=12$, $\alpha_1=16$; $e=4$\\
 {\bf $\bullet$ } $L7n2\{0\} \& L7n2\{1\}=P(3,\bar2,\bar2)$;
 $\alpha_0=0$; $\alpha_1=-12;e=-6$\\
 {\bf $\bullet$ } $L9n12\{1\}$;  $\alpha_0=12$, $\alpha_1=24$; $e=6$\\
 {\bf $\bullet$ } $L9n14\{0\}=P(3,\bar2,\bar4)$; 
 $m=3;\alpha_0=-6,\alpha_1=10;e=-2$\\
 {\bf $\bullet$ } $L9n15\{0\}^+
 =T(3,4,2,1)=P(3,\bar2,4)$;  $\alpha_0=18$, $\alpha_1=18$; $e=6$\\
 {\bf $\bullet$ } $L10n42\{1\}$;  $\alpha_0=-6$, $\alpha_1=10$; $e=-2$\\
 {\bf $\bullet$ } $ L11n132\{0\}^+
 $; $\alpha_0=6$, $\alpha_1=22$; $e=4$\\
{\bf $\bullet$ } $ L11n133\{0\}^+
$; $\alpha_0=18$, $\alpha_1=26$; $e=8$ \\
{\bf $\bullet$ } $L11n204\{0\}^+
=T(3,5,2,1)=P(3,\bar2,6)$; $\alpha_0=24$, $\alpha_1=20$; $e=8$ \\
{\bf $\bullet$ } $L11n208\{0\}=P(3,\bar2,\bar6)$; $\alpha_0=-12$, $\alpha_1=8$; $e=-4$. 
\vskip2mm
\noindent{}For all  the links listed above  the parameter $m$ is equal to $3$, coinciding with their braid index, with the only exception being $L10n42\{1\}$, which has braid index $\iota=4\neq m $. They all consist of  two components. For the links $L9n12$, $L11n132$ and $L11n133$ these are $(\bigcirc,5_1)$,  
while for the remaining ones they are $(\bigcirc,3_1)$. They are all hyperbolic links, with the exception of the twisted torus links $T(3,n,2,1)$, which have vanishing hyperbolic volume (i.e. their complement does not admit a hyperbolic structure). Their connections via full twists and Jucys-Murphy's braids can be seen in Table~\ref{fig:link map}.

\vskip2mm
\begingroup
  \begin{table}[h!]
  \caption{Factorisability map for 2-component links with $m=3$. 
     Links in this table can be obtained from each other in the horizontal direction by introducing a \textcolor{Purple}{Jucys--Murphy twist $E_m$}, 
     or in the diagonal direction by a \textcolor{teal}{full twist $F_m$}, as indicated by the arrows.}
      \centering

    \setlength{\tabcolsep}{20pt}
\renewcommand{\arraystretch}{2}
\setlength\extrarowheight{-4pt}
{\footnotesize
      \begin{NiceTabular}{|c|ccc|}
      \hline    
\rowcolor{gray!15}&&\tikzmarknode{d}{$\mathcal{L}_{j+1,1}$}&\tikzmarknode{e}{$\mathcal{L}_{j+1,2}$}\\
         j  & $\sigma_1\sigma_2^{2j}$\tikzmarknode{a}{} 
         &\tikzmarknode{b}{$\mathcal{L}_{j,1}=P(\bar 2,3,2j)$}\tikzmarknode{f}{}& \tikzmarknode{c}{$\mathcal{L}_{j,2}$} \\[8pt]
          \hline
          \hline
          \midrowspace
          \midrowspace
          \midrowspace
          3& $\sigma_1\sigma_2^{6}$& $L11n_{204}\{0\}$&\\
          \rowcolor{gray!15} 2& $\sigma_1\sigma_2^{4}$& $L9n_{15}\{0\}$&\\
            1& $\sigma_1\sigma_2^{2}$& $L7n_{1}\{0\}$&$L11n_{133}\{0\}$\\
       \rowcolor{gray!15} 0& $\sigma_1\sigma_2^{0}$& $T(2,2)\#T(2,3)$&$L9n_{12}\{0\}$\\
        -1& $\sigma_1\sigma_2^{-2}$& $L7n_{2}$&$L11n_{132}\{0\}$\\
       \rowcolor{gray!15} -2& $\sigma_1\sigma_2^{-4}$& $L9n_{14}\{0\}$&\\
        -3& $\sigma_1\sigma_2^{-6}$& $L11n_{208}\{0\}$&\\
        \hline
          
      \end{NiceTabular}
      \begin{tikzpicture}[overlay, remember picture,
                    shorten >=4pt, shorten <=4pt]
                    ]
\draw [violet,  thick, ->] (a) --node[midway, below, violet] {$E_3$} (b);
\draw [violet,  thick, ->] (f) --node[midway, below, violet] {$E_3$} (c);
\end{tikzpicture}
\begin{tikzpicture}[overlay, remember picture,
                    shorten >=4pt, shorten <=4pt]
                    ]
\draw [teal,  thick, ->] (a) --node[midway, above, teal] {$F_3$} (d);
\draw [teal,  thick, ->] (f) --node[midway, above, teal] {$F_3$} (e);
\end{tikzpicture}}
    \label{fig:link map}
    \vskip2mm
  \end{table}
\endgroup

{\bf $\bullet$} The pretzel family $P(3,\bar2,2j)=(\bigcirc,3_1)$ which contains $L7n1\{0\}=T(3,3,2,1)$, $L9n15\{0\}=T(3,4,2,1)$  and 
 $L11n204\{0\}=T(3,5,2,1)$ at $j=1,2,3$, respectively (but $P(3,\bar2,2j)\neq T(3,j+2,2,1)$ for $j>3$)  and $L7n2$, $L9n14$, $L11n208$ for $j=-1,-2,-3$ has HZ transform 
 \be
 Z(P(3,\bar2,2j)) = \frac{\lambda(1+ q^{6(1+j)}\lambda)(1-q^{14+2j}\lambda)}{(1-q^{2+2j}\lambda)(1-q^{4+2j}\lambda)(1-q^{6+2j}\lambda)(1-q^{8+2j}\lambda)}.
 \ee

{\bf $\bullet$} The family $\mathcal{L}_{j,k}=\sigma_1\sigma_2^{2j}\otimes E_3^k=(\bigcirc,3_1)$, which is obtained by $k$ Jucys--Murphy braids  and contains $P(3,\bar2,2j)$ at $k=1$, has HZ transform 
 \be
 Z(\mathcal{L}_{j,k}) = \frac{\lambda(1+ q^{6(k+j)}\lambda)(1-q^{10k+2j+4}\lambda)}{(1-q^{2j+4k-2}\lambda)(1-q^{2j+4k}\lambda)(1-q^{2j+4k+2}\lambda)(1-q^{2j+4k+4}\lambda)}.
\ee

{\bf $\bullet$} The family of twisted torus links $T(3,n,2,1)=(\bigcirc,3_1)$ with $n\geq3$, 
has HZ transform when $n\mod 3=0$
 \be
 Z(T(3,n,2,1)) = \frac{\lambda(1+ q^{4n}\lambda)(1-q^{4(n+1)}\lambda)}{(1-q^{2n-2}\lambda)(1-q^{2n}\lambda)(1-q^{2n+2}\lambda)(1-q^{2n+4}\lambda)},
 \ee
 while  the numerator exponents become $\alpha_0=\alpha_1=4n+2$ when  $n\mod 3=1$ and $\alpha_0=4(n+1)$, $\alpha_1=4n$, when $n\mod 3=2$. Additional full twists $T(3,n,2,1)\otimes F_3^k$  just yield $T(3,n+3k,2,1)$. On the other hand, by concatenating with $E_3$ we obtain:\\
 $\bullet$ $T(3,3,2,1)\otimes E_3^k:=(\sigma_2\sigma_1)^3\sigma_2\otimes E_3^k=(\sigma_1\sigma_2)^3\sigma_2\otimes E_3^k$; $\alpha_0=6(k+2)$, $\alpha_1=10k+16$, $e=4k+4$. \\
  $\bullet$ $T(3,4,2,1)\otimes E_3^k:=(\sigma_2\sigma_1)^4\sigma_2\otimes E_3^k$; $\alpha_0=6(k+3)$, $\alpha_1=10k+18$, $e=4k+6$. \\
    $\bullet$ $T(3,5,2,1)\otimes E_3^k:=
(\sigma_1\sigma_2)^2\sigma_1^2(\sigma_2\sigma_1)^2\sigma_2
    \otimes E_3^{k}$: $\alpha_0=6(k+4)$, $\alpha_1=10k+20$, $e=4k+8$, while $(\sigma_2\sigma_1)^5\sigma_2\otimes E_3$ is not factorisable.\\

 {\bf $\bullet$} The family\footnote{The orientation in this case does not need to be specified since $H(L7n2\{0\})=H(L7n2\{1\})$. Note moreover that $L7n2^-\otimes F_3^k$ are just the torus links $T(2,4k)$.} $L7n2^+\otimes F_3^k=(T(2,2k+3),\bigcirc)$ and has the HZ function
 \be
 Z(L7n2^+\otimes F_3^k) = \frac{\lambda(1+ q^{12k}\lambda)(1-q^{12(k+1)}\lambda)}{(1-q^{6k}\lambda)(1-q^{6k+2}\lambda)(1-q^{6k+4}\lambda)(1-q^{6k+6}\lambda)}.
 \ee

 {\bf $\bullet$} The family $L9n14\{0\}\otimes F_3^k$ has HZ  transform
  \be
 Z(L9n14\{0\}\otimes F_3^k) = \frac{\lambda(1+ q^{12k-6}\lambda)(1-q^{12k+10}\lambda)}{(1-q^{6k-2}\lambda)(1-q^{6k}\lambda)(1-q^{6k+2}\lambda)(1-q^{6k+4}\lambda)},
 \ee
 while for $L9n14\{0\}^+\otimes E_3^k:=\sigma_{1}^{-1} \sigma_{2} \sigma_{1} \sigma_{2}^{-2} \sigma_{1} \sigma_{2} \sigma_{1}^{-1} \sigma_{2}\otimes E_3^k$ is becomes
 \be
 Z(L9n14\{0\}\otimes E_3^k) = \frac{\lambda(1+ q^{6(k-1)}\lambda)(1-q^{10(k+1)}\lambda)}{(1-q^{4k-2}\lambda)(1-q^{4k}\lambda)(1-q^{4k+2}\lambda)(1-q^{4k+4}\lambda)}.
 \ee

{\bf $\bullet$} The 2-component twisted torus links $T(4,4k+1,2,1)=(T(3,3k+1),\bigcirc)$ with braid index 4, an example of which is shown in Fig.~\ref{fig:t4521},
 have HZ transform 
\be\label{T4n21}
 Z(T(4,4k+1,2,1)) = \frac{\lambda(1+q^{6(3k+1)}\lambda)(1- q^{6(3k+1)}\lambda)(1-q^{8(3k+1)}\lambda)}{\prod_{i=0}^4(1-q^{12k+2i}\lambda)}.
 \ee
  \begin{figure}[h!]
     \centering
\includegraphics[width=0.14\linewidth,angle=91]{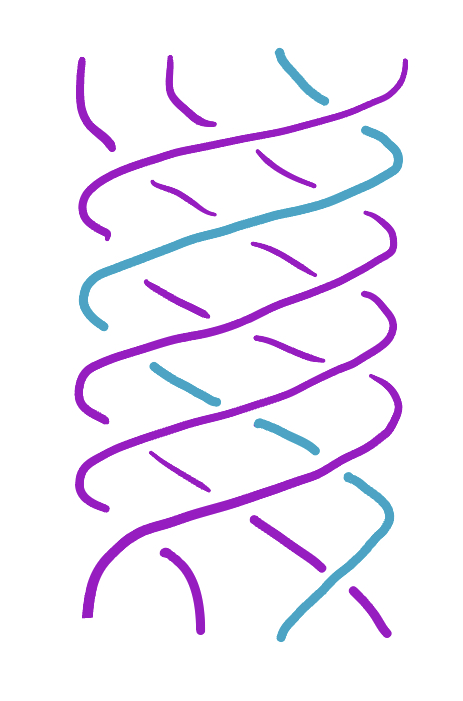}
     \caption{The  twisted torus link $T(4,5,2,1)$ is the closure of the braid $(\sigma_3\sigma_2\sigma_1)^5\sigma_2$. The two components $(\textcolor{Fuchsia}{T(3,4)},\textcolor{TealBlue}{\bigcirc})$ of the link are depicted with different colors.}
     \label{fig:t4521}
 \end{figure}
 
The HZ transform for $T(4,5,2,3)$ and $T(4,n,2,1)$ for  $ n=4,6,7$  are only quasi-factorisable,  while for $n=8$ is completely not factorasibale. Moreover, $Z(T(5,7,2,1))$ is again quasi-factorisable, while $Z(T(5,6,2,1))$ contains two overall factors $(1 - \lambda q^{34}) (1 - \lambda q^{36})$ in the numerator.

Note that, at $q=1$, all the above formulas become $\lambda(1+\lambda)/(1-\lambda)^3 = \sum_{N=0}^\infty N^2 \lambda^N.$ 
Moreover, for  the 3-component link  $L6n1\{0,1\}$ whose HZ transform  given in (\ref{L6n1}),  at $q\to 1$ it can be expanded as a power series as
\be\label{6n1}
Z(\mathcal{L};\lambda,q=1) = \frac{\lambda(1+4 \lambda + \lambda^2)}{(1-\lambda)^4} =\sum_{N=0}^{\infty}N^3 \lambda^N.
\ee 
\begin{prop}\label{prop:Zq1}
 The HZ transform at $q=1$ for a link $\mathcal{L}$ with $l$ components ($l\geq 1$), is expressed as
\be
Z(\mathcal{L};\lambda,q=1)= \sum_{N=0}^{\infty} N^l \lambda^N. 
\ee
\end{prop} 
\begin{proof}
    By the definition (\ref{KhovanovC}) of the HZ transform, the coefficient of $\lambda^N$ in the r.h.s is equal to the unnormalised HOMFLY--PT polynomial $\bar{H}$. To see that this becomes equal to $N^l$ in the limit $q\to 1$, we observe that in the skein relation (\ref{skein}) for the HOMFLY--PT polynomial the factor $z=q-q^{-1}$ vanishes, and hence there are no contributions from  smoothings of a crossing. The switching of over- to under-crossing contributes a factor $1=q^N\big\vert_{q=1}$, and since this operation leaves the number of components $l$ invariant the computation will stop at an $l$-component unlink. Since the latter contributes a factor $(\frac{a-a^{-1}}{q-q^{-1}})^{l-1}$ and noting that there is an additional  factor $\frac{a-a^{-1}}{q-q^{-1}}$ due to the normalisation included in $\bar{H}$ and since $\lim_{q\to 1}\frac{a-a^{-1}}{q-q^{-1}}=N$, this concludes the proof.
\end{proof}

\subsection*{Morton-Franks-Williams inequality}\label{rmk:MFW}
Before we conclude this section, it is important to elaborate on the topological significance of the parameter $m$ 
that appears in the HZ transform (\ref{eq:factorisableZZ}). In particular, it  is related to the $a$-span in the HOMFLY--PT polynomial of a knot or link and according to the \emph{Morton-Franks-Williams (MFW) inequality} \cite{Morton, Franks}, it
provides a lower bound of its braid index\footnote{The braid index is equal to the least number of Seifert circles in any of the knot's projections.} $\iota$. Explicitly, $\iota\ge \frac{1}{2}(E_h-E_l)+1$, where $E_l$ and $E_h$ are the lowest and highest exponents of $a$ in the normalised HOMFLY--PT polynomial $H(a,z)$. In terms of  the HZ parameter $m$, the MFW invequality can be simply rephrased as 
\be
\iota\geq m.
\ee

Although  the inequality is sharp for most knots, there exist 5  exceptional cases   with up to $10$ crossings, namely $9_{42},9_{49}, 10_{132},10_{150}$ and $10_{156}$, which  all have braid index $\iota=4$ (see  the braid index table 15.9 of  \cite{Jones} and the comments therein). The reason  this occurs  for the knots $10_{132}$ and $10_{156}$, is that the HOMFLY--PT polynomial can not distinguish them from other knots that have  a lower braid index and for which the inequality is sharp. In particular, $H(10_{132})=H(5_1)$ and $H(10_{156})=H(8_{16})$, which hence determines $m(10_{132})=\iota(5_1)=2$ and $m(10_{156})=\iota(8_{16})=3$, which are smaller than their actual braid index. Similarly, for the two-component link $L10n42$ holds\footnote{These two links   can be distinguished, however,  by  the Kauffman polynomial and the multi-variable Alexander polynomial \cite{KnotInfo}.}

Such coincidences, may be interpreted through  the HZ transform to be 
  a result of a cancellation between several factors of the form  $(1-\lambda q^{\beta_i})$, which occur both in  the numerator and denominator. In the case of $10_{132}$ for which the HZ is factorised, this can be clearly seen from $Z(10_{132}^{4k})$ in (\ref{1324k}) at $k=0$, in which the factors $(1-\lambda q)$ and $(1-\lambda q^{-1})$ cancel. Prior to  this cancellation $m=4$ is equal to the braid index $\iota(10_{132})$. 
  Beyond 10-crossings, the same holds, for example, for the knots  $12n_{749}=\mathcal{K}_{1,2}$ and $8_{20}\otimes10_{139}=\mathcal{K}_{0,3}$,  
  as their HOMFLY--PT polynomials are equal to the ones for\footnote{Classifying hyperbolic knots which have the same HOMFLY--PT polynomial as a torus knot, which hence are also cases of non-exact MFW inequality, seems to be  an interesting task that we aim to address in the future.} $T(2,7)$ and $T(2,11)$, respectively, and  hence they both have $m=2$. Their braid index $\iota=3$ coincides with $m$ in 
  (\ref{containT211}), in which, indeed, a factor cancels at the values $(j=1,k=2)$ and $(j=0,k=3)$. 
   Similarly, for the remaining cases which do not admit HZ factorisability, it can be understood that the HZ transform is quasi-factorisable, but the overall factor cancels out. By considering the corresponding families $\mathcal{K}\otimes F_{\iota}^ k$ obtained by full twists, we can deduce what those "missing factors" are. In particular, we find $(1-\lambda q^5)$ for $9_{42}$, $(1-\lambda q^{11})$ for $9_{49}$,  $(1-\lambda q^{9})$ for $10_{150}$, $(1-\lambda q^7)$  for $10_{156}$ and $(1-\lambda q^6)$ for $L10n42$\footnote{Note that $Z(L10n42\otimes F_4^k)$ is only quasi-factorisable at $k=1$ and non-factorisable for $k>1$.}. Notwithstanding, the HOMFLY--PT polynomial itself is, in fact, oblivious to this information, and recovering it in this way through the HZ transform requires prior knowledge of the braid index $\iota$ in each case.

\section{Closed formulas for knot polynomials via inverse HZ transform\label{sec:InverseHZ}}
The HOMFLY--PT is a lengthy and complicated polynomial  for knots and links with a high number of crossings, and closed formulas for it exist only for special classes of knots and links, such as the torus knots \cite{RossoJones}. However, in the previous section we have shown that in specific cases the HZ transform of the HOMFLY--PT polynomial takes a remarkably simple, factorised form, that can be easily extrapolated to include members of infinite families. Hence, it is interesting to write down closed formulas for the HOMFLY--PT polynomial of such families, given their HZ function $Z(\mathcal{K})$. This can be achieved by employing the inverse HZ transform, according to the following theorem.

\begin{thm}\label{thmInverseHZ}
The HOMFLY--PT polynomial $H({\mathcal{K}};q^N,q)$  can be obtained from $Z({\mathcal{K}};\lambda,q)$ via a contour integral around the pole at $\lambda=0$ as
\be\label{InverseHZ}
H({\mathcal{K}};q^N,q)= \frac{q-q^{-1}}{q^N-q^{-N}} \frac{1}{2 \pi i} \oint_{\{\lambda=0\}} d\lambda \frac{1}{\lambda^{N+1}} Z({\mathcal{K}};\lambda,q)
\ee
Alternatively, choosing a contour enclosing all the simple poles at $\lambda = q^{-\beta_i}$, it can be expressed as
\be\label{AlternativeJonescontour}
H({\mathcal{K}};q^N,q)= -\frac{q-q^{-1}}{q^N-q^{-N}} \frac{1}{2\pi i}\ \oint_
{\{all \hskip 1mm poles\hskip 1mm \lambda=q^{-\beta_i}\}} d \lambda \frac{1}{\lambda^{N+1}} Z(\mathcal{K};\lambda,q),
\ee
\end{thm}
\begin{proof}
The way to recover the HOMFLY--PT polynomial from the HZ function is by applying the  inverse HZ transform, i.e. by  inverse Laplace transform. First, notice that the HZ function in (\ref{KhovanovC}) is a Maclaurin series with coefficients the unnormalised HOMFLY--PT polynomial $
\bar{H}(q^N,z)=\sum_{i}N_iq^{N\beta_i}z^{c_i}$, ($z=q-q^{-1}$), with $N_i,\beta_i, c_i\in\mathbb{Z}$. The radius of convergence $\rho$ of this series is given by $1/\rho=\lim_{N\to\infty}|\bar{H}(N+1)|/|\bar{H}(N)|$, which gives $\rho=|q^{-\beta_k}|$ for a fixed $k$ s.t. $\beta_k>\beta_i\;\forall i$. Hence $Z$ is analytic in the disk 
$\mathcal{C}=\{\lambda\in\mathbb{C}:|\lambda|<|q^{-\beta_k}|\}$. 
To understand the equivalence of $Z$ with a discrete Laplace transform we set $\lambda=e^{-s}$ for some complex parameter $s=\gamma+i\sigma$ ($\gamma,\sigma\in \mathbb{R}$), with which it becomes periodic $Z(s)=Z(s+2\pi i)$ and it  is analytic for $\gamma>\ln 1/\rho=\ln |q^{\beta_i}|$. Then, from the inverse Laplace (or Mellin) transform we get $\bar{H}(N)=(2\pi i)^{-1} \int_{\gamma-i\infty}^{\gamma+i\infty} Z(s)e^{sN} ds$. After the change of variables $\lambda=e^{-s}$ this gives the contour integral in (\ref{InverseHZ}) evaluated on a counterclockwise contour including only the  pole of order $N$ at $\lambda=0$, which is the only one lying in the domain of analyticity $\mathcal{C}$. Finally, dividing by the normalisation factor $(q^N-q^{-N})/z$ gives the normalised HOMFLY--PT polynomial.  The alternative contour integral in (\ref{AlternativeJonescontour}) utilises the fact that the residue of a pole inside a given contour is equal to minus the sum of residues of the poles lying outside of it.
\end{proof}

As an application of the above theorem, the HOMFLY--PT polynomial  of knots that admit HZ factorisability, i.e. satisfy (\ref{eq:factorisableZ}), can be written in terms of the HZ parameters $m$, $\{\alpha_i\}$ and $\{\beta_i\}$    (substituting $a=q^N$) as
\be\label{FactorisedHOMFLY}
\boxed{H({\mathcal{K}};a,q)= \frac{q-q^{-1}}{a-a^{-1}}\sum_{j=0}^m \frac{\prod_{i=0}^{m-2}(1-q^{\alpha_i-\beta_j})}{\prod_{i=0,i\neq j}^{m}(1-q^{2(i-j)})}q^{-\beta_j}a^{\beta_j}.}
\ee

\noindent{}Some explicit examples follow.

$\bullet$
 The two-strand torus knots $T(2,n)$ have HZ parameters $\alpha_0=3n$, $\beta_i=n-2+2i$. 
 From the above theorem, their HOMFLY--PT polynomial can be derived as \be
 H(T(2,n); a,q) = \frac{a^{n-1}(q^{n+1}-q^{-(n+1)}) - a^{n+1}(q^{n-1}-q^{-(n-1)})}{q^2-q^{-2}},\ee
 in agreement with the Rosso--Jones formula \cite{RossoJones}. This can be written as a polynomial in $a$ and $z=q-q^{-1}$ after the substitutions $q=(z+\sqrt{z^2+4})/2$ and $ q^{-1}= (z-\sqrt{z^2+4})/2$.
 
$\bullet$ The Pretzel knots $P(\bar2, 3,2j+1)$ have HZ transform (\ref{eq:Pmin23plus2j}) with parameters $\alpha_0=n+9$, $\alpha_1=3n-9$, $\beta_i=n-3+2i$, where $n= 6 + 2j$. 
Using (\ref{FactorisedHOMFLY}), we obtain
\ba
H(P(\bar2, 3,2j+1);a,q)&=& \frac{a^{n-2}}{q^{n+6}(q^4-1)}\bigg(q^2(1+q^6)(q^{2n}-q^6)\nonumber\\
&&\hspace{-30mm}+ a^2 (q^{10}(1+q^2+q^6)-q^{2n}(1+q^4 + q^6))+a^4 q^2 (q^{2n}-q^{12})\bigg).
\ea 
$\bullet$ For the family $T(3,s)\otimes E_3^k$ with $s\mod3=1$, the HOMFLY--PT polynomial can be derived from (\ref{t3sE3k}) via the above formula 
to be 
\ba
H(T(3,s)\otimes E_3^k)\hspace{-2mm}&=&\hspace{-2mm}\frac{(a/q)^{4 k + 2 s -2}}{(1-q^2)^2 (1+q^2)(1+q^2+q^4)} \bigg( 1- a^2 q^2 - a^2 q^4 + a^4 q^6\nonumber\\
&&\hspace{-20mm}+ q^{2 k + 2 s}(a^2 - q^2 + a^2 q^4 - a^4q^2)+ q^{6 k + 2 s + 2} ( a^2 - q^2 + a^2 q^4 -a^4q^2)\nonumber\\
&&\hspace{-20mm}+ q^{8 k + 4 s} (a^4 -a^2 q^2 + q^6 - a^2 q^4)\bigg).
\ea  
\vskip2mm
The HOMFLY--PT reduces to the Jones polynomial $J(q^2)=H(q^2,q-q^{-1})$ when $N=2$.
Therefore we deduce the following
\begin{cor}\label{corInverseJones}
\noindent{}The Jones polynomial $J(\mathcal{K};q^2)$ can be expressed as the following contour integral of $\lambda^{-3}Z$ around the double pole at $\lambda=0$,
\be\label{JonesContour}
J(q^2) =\frac{1}{q+ q^{-1}} \frac{1}{2\pi i}\ \oint_
{\{\lambda=0\}} d \lambda \frac{1}{\lambda^3} Z(\mathcal{K};\lambda,q).
\ee
\end{cor}

\vskip2mm \noindent{Example.} The Jones polynomial for the family $8_{20}^{3k}$ can be computed from (\ref{P(-2,3,-2j-1),3k}) at $j=1$ to be 
\be
J(8_{20}^{3k};q)=q^{3k}\left(1+q^2-q^{3k-1}(1-q+q^2-q^3+q^4-q^5+q^6)\right).
\ee

\vskip 2mm

 Theorem~\ref{thmInverseHZ} holds for all $N\ge 1$, but it requires a specification for $N=0$ ($a=1$), corresponding to the Alexander polynomial. This is because dividing by $(a-a^{-1})$ becomes problematic in the limit $a\to 1$.

\vskip 2mm
\begin{thm}\label{thmInverseAlex}
The Alexander polynomial can be obtained from $Z(\mathcal{K};\lambda,q)$ by
\be\label{AlexanderContour} 
\Delta(q^2) = \frac{q-q^{-1}}{2} \frac{1}{2i\pi} \sum_{i=0}^m (-\beta_i)
\oint_{\{
\lambda= q^{-\beta_i}\}}d\lambda  \frac{Z(\mathcal{K};\lambda,q)}{\lambda},
\ee
in which 
the integration contour at each $i$ should include the pole 
 at $\lambda= q^{-\beta_i}$. Via the residue theorem this becomes
 \be\label{InverseAlexResidue}
\Delta(q^2) =\frac{q-q^{-1}}{2}\sum_{\beta_i}(-\beta_i) Res_{\lambda=q^{-\beta_i}}(\lambda^{-1}Z).
 \ee
 \end{thm}
 \begin{proof}
     Under partial fraction decomposition $Z(\mathcal{K})$ can be decomposed as 
 \be\label{left}
Z(\mathcal{K}) = \frac{\lambda g(\lambda,q)}{\prod_i (1- q^{\beta_i }\lambda)} = \sum_{i=0}^m \frac{g_i(q)}{(1-q^{\beta_i}\lambda)},
 \ee
 in which the functions $g_i(q)$ are polynomial in $q^{\pm1}$ and are independent of $\lambda$. 
 By residue calculus, $g_i(q)$ can be obtained from $Z(\mathcal{K})$ via contour integrals around the pole at $\lambda = q^{-\beta_i}$ as
 \be
 g_i (q) =\frac{1}{2\pi i} \oint_{\{\lambda= q^{-\beta_i}\}}   \frac{-g(\lambda,q)}{\prod_j (1- q^{\beta_j }\lambda)}d\lambda.
 \ee 
 These clearly are the coefficients of $a^{\beta_i}$ (which is replaced by $(1-q^{\beta_i}\lambda)^{-1}$ under HZ transform)
 in the unnormalised HOMFLY--PT polynomial and hence the latter can be expressed as $\bar{H}(a,q) =\sum_i  g_i(q)a^{\beta_i}$. The Alexander polynomial $\Delta(q^2)$ can be obtained as the limit $\lim_{a\to1}H(a,q)=\lim_{a\to1}\frac{q-q^{-1}}{a-a^{-1}}\bar H(a,q)$, which evaluates to $0/0$. Using de l'Hôpital's rule 
 \be\label{Alex}
 \Delta(q^2) =\frac{1}{2} (q-q^{-1}) \frac{\partial}{\partial a} \sum_i  g_i(q)a^{\beta_i}\bigg\vert_{a=1}
 = \frac{1}{2} (q-q^{-1}) \sum_i \beta_i g_i(q).
 \ee
  \end{proof}

Note that the above theorems apply to knots and links with both factorised and non-factorised HZ transform.
\vskip2mm \noindent{Example.}  The figure-8 knot with (non-factorisable) HZ function \be
Z(4_1;\lambda,q)= \frac{\lambda(1 + \lambda (q^{-5}-q^{-3}-q^{-1}-q-q^3+q^{5}) + \lambda^2 )}{(1-q^{3}\lambda)(1-q \lambda)(1-q^{-1}\lambda)(1-q^{-3}\lambda)},\ee 
 via (\ref{InverseAlexResidue}) 
 yields $\Delta(4_1;q^2) 
= 3 - q^2-q^{-2}$, which becomes the standard Alexander polynomial $\Delta(4_1;t)$ \cite{Bar-Natan} after replacing $q^2=t$.

For torus knots there is a well known  formula for their Alexander polynomial
 \be
 \Delta(T(m,n);t) = t^{-\frac{(m-1)(n-1)}{2}} \frac{(t^{m n}-1)(t-1)}{(t^m-1)(t^n-1)},
 \ee 
which has alternating   coefficients $\pm 1$. This  can be seen, for example, in $\Delta(T(3,4))$ $= t^3-t^2 + 1- t^{-2}+ t^{-3}$ and $\Delta(T(3,5))= t^4 - t^3 + t - 1 + t^{-1} - t^{-3} + t^{-4}.$ 
Applying Theorem~\ref{thmInverseAlex} to $Z(\mathcal{K})$  given by (\ref{eq:factorisableZ}), with the help of the residue formula we obtain the following closed expression for the Alexander polynomial for knots with factorised HZ transform\footnote{   Alternatively, the Alexander polynomial for factorised cases could be determined by first computing the normalised HOMFLY--PT polynomial using (\ref{FactorisedHOMFLY}), in which the factor $(a-a^{-1})^{-1}$ always cancels, and then set $a=1$.}
\be\label{eq:factInverseAlex}
\boxed{\Delta(\mathcal{K};q^2) = \frac{q-q^{-1}}{2}  \sum_{j=0}^m 
\frac{\prod_{i=0}^{m-2}(1-q^{\alpha_i-\beta_j})}{\prod_{i=0,i\neq j}^m(1-q^{2(i-j)})}\beta_jq^{-\beta_j}.}
\ee
This formula is applied to several examples below.

$\bullet$ The Alexander polynomial for the pretzel family $P(\overline{2}, 3,\overline{ 2j+1})$  
  can be determined from its HZ parameters $m=3$, $\alpha_0=13-2j$, $\alpha_1=3-6j$ and $e=1-2j$ via (\ref{eq:factInverseAlex}) to be
  \be\label{AlexPretzel}
\Delta(P(\overline{2},3,\overline{2j+1});q^2)=\frac{q^{-2 (j+1)}}{1+q^2}
   \left(1-q^2+q^6+q^{4
   j}\left(1-q^4+q^6\right) \right).
   \ee
   For the family $P(\bar2,3,\overline{2j+1})^{3k}$  (including $k$ full twists) from (\ref{P(-2,3,-2j-1),3k}) we get
   \be
\Delta(P(\bar2,3,\overline{2j+1})^{3k};q^2)= \frac{(1+q^2)(q^{4j}+q^{8+12k})+\left(1+ q^{10+4j}\right) q^{6 k}}{q^{2 (1+j+3k)}\left(1+q^2\right) \left(1+q^2+q^4\right)} ,
   \ee
   which reduces to (\ref{AlexPretzel}) at $k=0$, and at $k=j=3$ gives $ \Delta(P(\bar 2,3,\bar 7)^9;t=q^2)= 
 t^7-t^6 + 2 t^4 - 3 t^3 + 2 t^2 -1 + 2 t^{-2} -3 t^{-3} + 2 t^{-4} - t^{-6} + t^{-7}$. 
 At $j=3$ it can also be recursively determined 
 (for $k\ge 3$) by
 $
 \Delta(P(\bar 2,3,\bar7)^{3k}) =
  \Delta(P(\bar 2,3,\bar7)^{3(k-1)})+ t^{3k-2}- t^{3k-3} - t^{-3k +3}+ t^{-3k+2}.
 $

$\bullet$  Another example is $P(\bar 2, 3, 2j+1)^{3k}$, 
 for which  (\ref{3kII}) yields
   \be
 \Delta(P(\bar 2, 3, 2j+1)^{3k};t)= \frac{1+t+t^{3
k+5}+\left(1+\left(1+t\right) t^{3
   k+4}\right) t^{2 j+3 k+2}}{t^{j+3 k+2}\left(1+t\right) \left(1+t+t^2\right)}.
   \ee

$\bullet$ For $10_{128}^{4k}$ 
the Alexander polynomial can be derived explicitly  from  (\ref{1284k}) as
\be
\Delta(10_{128}^{4k})=\frac{1 +  t^{9 +12 k} +  t^{4 k} (1 -t  + t^3 - t^4 + t^{5}) + t^{8 k+4} (1 - t + t^2 - t^4 + t^{5})}{t^{3+6k}(1 + t) (1 + t^2)}.
   \ee

$\bullet$ For the family $5_2^{3k}$ (shown in Fig.~\ref{fig:5_2_3k}), from  (\ref{5_2^{3k}}) we obtain
\be 
 \Delta(5_2^{3k};t)=\frac{1-t+t^2-t^3+t^4+t^{-3 k}+t^{3 k+4}}{t
   \left(1+t+t^2\right)},
\ee
or recursively by  
$ \Delta(5_2^{3k})= \Delta(5_2^{3(k-1)}) + t^{3k+1} - t^{3k} - t^{-3k} + t^{-(3k+1)}$ ($k\ge 1$).

    It may be interesting to note that the Alexander-Conway polynomial $\Delta(\mathcal{K})+ \rho$,  with $\rho=\{-1,0,1\}$, for several  knots $\mathcal{K}$ for which HZ factorisability holds, 
   can be written in terms of $\Delta(3_1)$ and $\Delta(5_1)$ as follows
\ba\label{threecases}
&&\Delta(5_2)+1 = 2 \Delta(3_1),\;\;\Delta(8_{20}) = (\Delta(3_1))^2,\nonumber\\
&&\Delta(10_{125})-1 = \Delta(3_1)(\Delta(5_1)-1),\nonumber\\
&&\Delta(10_{128})-1 = \Delta(3_1) (2 (\Delta(5_1))+ \Delta(3_1)-3),\;\;\Delta(10_{132})= \Delta(5_1),\nonumber\\
&&\Delta(10_{161})+ 1 = \Delta(3_1)(2 \Delta(3_1)^2 + 2 \Delta(3_1)+\Delta(5_1)-3),\nonumber\\
&&\Delta(10_{139}) = \Delta(3_1) (\Delta(3_1)\Delta(5_1)+ 2 \Delta(3_1)^2-2).  
\ea
Note also that $\Delta(5_1)+1 = \Delta(3_1)(\Delta(3_1)+1)$.

Furthermore, noting that  $q^2=t=(-1)^{1/3}$  is a root of $z^2=(q-q^{-1})^2= q^2 +q^{-2}- 2=-1$, at which the the Alexander polynomials for 
  $3_1$, $8_{20}$, $10_{139}$ are vanishing, it may be interesting to consider evaluation of the Alexander polynomials of the remaining factorised cases at this value.
  We find \\
  $\bullet$  For the pretzel family $ \Delta(P(\bar2,3,\overline{2j+1});t= (-1)^{1/3}) = 0$ for $j\mod3=1$, while it becomes equal to $\pm1$ for $j=2,3,8,9,...$ and $j=5,6,11,12,...$, respectively.\\ 
  $\bullet$ For the $5_2^{3k}$ family  
 $ \Delta(5_2^{3k};t = (-1)^{1/3}) = -1$ when $k$ is even, otherwise it is $0$.\\
 $\bullet$ The Alexander polynomial  $\Delta(10_{128}^{4k};t=(-1)^{1/3})$ 
 takes the values $\{0,-1,1\}$ when $k\mod 3=\{1,2,3\}$, respectively.\\
 $\bullet$ The Alexander polynomial for $10_{132}^{4k}$ (which can be derived from
 (\ref{1324k})) shows the same mod 3 behavior, since
$\Delta(10_{132}^{4k};t=(-1)^{1/3})= \{0,1,-1\}$, at $k\mod 3=\{1,2,3\}$, respectively. \\
These observations lead to the following remark.
\begin{remark}
\textup{  The Alexader polynomial of knots with factorised HZ
  take the values  $0$ or $\pm1$,
  when evaluated at  $t=q^2=(-1)^{1/3}$. }
  \end{remark}

 \section{Kauffman   polynomial and HZ factorisability}
The  Kauffman polynomial $KF(\mathcal{K};a,z)$ \cite{Kauffman} of a knot or link $\mathcal{K}$  
 can be defined as\footnote{The definition here differs from the original in \cite{Kauffman} by $a\mapsto a^{-1}$. Our convention is such that for positive knots $\mathcal{K}^+$, the Kauffman polynomial  has positive powers of $a$,  as it is the case for the HOMFLY--PT polynomial defined in (\ref{skein}) and hence it will allow their comparison later in this section.} $KF(\mathcal{K};a,z)=a^{w(\mathcal{K})}L(\mathcal{K};a,z)$, where $w(\mathcal{K})$ is the writhe of the diagram, while $L(\mathcal{K})$ satisfies the skein relation
\ba\label{KauffSkein}
&&L(\img{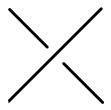})+L(\img{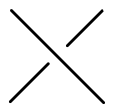})=zL(\img{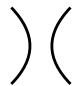})+zL(\img{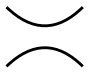})\\\nonumber
&&L(\img{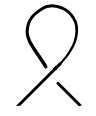})=aL(\img{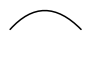}),\;\;(\img{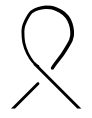})=a^{-1}L(\img{NonLoop.png})
\ea
and the normalisation condition $L(\img{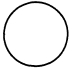})=1$.  For two disconnected knots  $KF(\mathcal{K}_1\sqcup\mathcal{K}_2)=(-1+\frac{a+a^{-1}}{z})KF(\mathcal{K}_1)KF(\mathcal{K}_2)$. 
Note that although the skein relation for $L(\mathcal{K})$ applies to unoriented links, a choice of orientation is necessary to compute the writhe $w(\mathcal{K})$. Hence, the Kauffman polynomial is orientation depended, but  only  through the overall factor $a^{w(\mathcal{K})}$.

As it is the case for the HOMFLY--PT polynomial, it also admits a definition via  Chern--Simons theory  with the gauge group\footnote{Note that the Euler characteristics for the moduli space
 of Riemann surfaces \cite{HZ} can be extended to Lie algebras of type $D$ (i.e. $\mathfrak{so}_{2N}$) containing symmetric and anti-symmetric terms \cite{Hikami2}. However, the Harer-Zagier transform of the Kauffman polynomial does not exhibit factorisation even for the simplest torus knots \cite{Morozov}. } $SO(N+1)$. 
Such a definition is in agreement with (\ref{KauffSkein}) after an overall multiplication with $(-1+\frac{a+a^{-1}}{z})$ and the substitutions  $a= q^{N}$ and $z=q-q^{-1}$. The resulting polynomial is referred to as the unnormalised (or unreduced) Kauffman polynomial. 

 \subsection{Relation between the HOMFLY--PT and Kauffman polynomials of HZ-factorisable knots and links \label{sec:HandKF}}
 There is a well known relation between the Kauffman polynomial $KF(\mathcal{K};a,q)$ and the Jones polynomial $J(\mathcal{K};t)$   (with $t=q^2$) given by \cite{Lickorish} 
 \be
 KF(\mathcal{K}; - t^{\frac{3}{4}}, t^{\frac{1}{4}}+ t^{-\frac{1}{4}}) = J(\mathcal{K};t).
 \ee

  Although, in general, the Kauffman polynomial contains many more terms as compared to the HOMFLY--PT polynomial, in the case of torus knots $T(m,n)$, with $(m,n)$ co-prime, there exist a relation between them. Denoting the  Kauffman 
 and  HOMFLY--PT polynomials for torus knots by $KF_{m,n}$ and $H_{m,n}$, respectively, their relation reads \cite{Jones,Labastida}
 \ba\label{KF}
 H_{m,n}(a,z) &=& \frac{1}{2}(\widetilde {KF}_{m,n}(a,z)+\widetilde {KF}_{m,n}(a,-z))\nonumber\\
 &-& \frac{z}{2(a-a^{-1})}(\widetilde{ KF}_{m,n}(a,z)-\widetilde {KF}_{m,n}(a,-z)).
 \ea
 Here $\widetilde {KF}$ refers to the Dubrovnik version of the the Kauffman polynomial, 
 which is obtained by substituting  $a\mapsto i a$ and $z\mapsto i z$, i.e. $\widetilde {KF}(a,z)=KF(ia,iz)$ \cite{Labastida}.  
 The first and second parts, which shall be labeled as \ba\label{tildeKF}
 &&\widetilde{KF}_{\rm even}:=\frac{1}{2}(\widetilde {KF}(a,z)+\widetilde {KF}(a,-z)),\nonumber\\
 &&\widetilde{KF}_{\rm odd}:=\frac{1}{2}(\widetilde {KF}(a,z)-\widetilde {KF}(a,-z)),\ea 
 contain the terms with even and odd powers of $a$ and $z$, respectively. 
 It can be  easily verified for simple torus knots, such as $3_1=T(2,3)$, for instance, for which $KF(3_1)=a^5z+a^4z^2-a^4+a^3z+a^2z^2-2a^2$. The odd  part is $\widetilde {KF}_{\rm odd}=a^3 z( 1-a^2)$, which is divisible by $(a-a^{-1})$. With the even part being $\widetilde {KF}_{\rm even}=- a^4 z^2- a^4 + a^2 (z^2 + 2)$, the r.h.s of (\ref{KF}) indeed yields the HOMFLY--PT polynomial of the trefoil  $-a^4 + a^2( z^2+2)=H(3_1)$.
Note, however that (\ref{KF}) does not hold for torus links, i.e. when $(m,n)$ are not co-prime.

The above relation also does not hold for general, non-torus knots. 
For instance,  for $4_1$ the r.h.s. of (\ref{KF}) becomes $(a^2+ a^{-2})(z^2+1) -z^4- 3 z^2 -1$, which differs from its HOMFLY--PT polynomial $H(4_1;a,z)=
 a^2+a^{-2}-z^2-1$.  
 
 It is remarkable, however, that for $5_2$, which is a hyperbolic knot but has a factorised HZ transform, the formula  (\ref{KF}) is true. Explicitly,  $\widetilde{KF}_{\rm odd}(5_2)=-a^7 (z^3 + 2  z) + 2 a^5 (z^3+ z)-a^3 z^3$, while  $\widetilde{KF}_{\rm even}(5_2)= a^6(z^4-2 z^2 +1) +a^4(z^4-z^2+1) +a^2(z^2-1)$, hence  yielding the correct HOMFLY--PT polynomial $H(5_2,a,z)$.

 \vskip 2mm
 \begin{thm}\label{thm:}

      The HOMFLY--PT and Kauffman polynomials for the following families of knots 
     \textup{ (i)}     $P(2,-3,\pm(2j+1))$, 
\textup{(ii)} $\mathcal{K}_{j,k}:=\sigma_1\sigma_2^{-2j-1}\otimes E_3^k$ 
      for $k=2,3$ and 
      \textup{(iii)} $5_2^{3k}=5_2\otimes F_3^k$, which admit a factorisable HZ transform, 
      satisfy
      \be\label{KFandH1}
     \boxed{H(\mathcal{K};a,z) =\widehat{KF}(\mathcal{K};a,z):=\widetilde{KF}_{\rm even}(\mathcal{K};a,z)-\frac{z}{a-a^{-1}}\widetilde{KF}_{\rm odd}(\mathcal{K};a,z).}
      \ee
 \end{thm}
 \begin{proof}

   \noindent{}
   (i) For the first pretzel family $P(2,\bar{3},2j+1)$ (with $j\geq0$), using (\ref{KauffSkein}) we find  the following recursive relation for (the Dubrovnik version of) its Kauffman polynomial 
\ba
\nonumber
&&\hspace{-15mm}\widetilde{KF}(P(2,\overline{3},2j+1))=a^{2}\widetilde{KF}(P(2,\overline{3},2j-1))\\\nonumber
&&+z^2a^{2j+2}\sum_{i=1}^{j}a^{-2i}\widetilde{KF}(P(2,\overline{3},2i-1))-z^2a^{2j-4}\sum_{i=1}^{j}a^{2i}\\
&&+za^{2j-1}(1-(a-a^{-1})(z+z^{-1}))\widetilde{KF}(T(2,-3))-za^{4j-3}.
\ea
Here $P(2,\overline{3},-1)=T(2,-5)=5_1^-$ and $T(2,-3)=3_1^-$ with $\widetilde{KF}_{2,-3}=-a^{-4}+2a^{-2}+za^{-3}(-1+a^{-2})+z^2a^{-2}(1-a^{-2})$. 
After some simple algebraic manipulations it can be shown that 
\ba
&&\hspace{-10mm}\widehat{KF}(P(2,\overline{3},2j+1))=a^{2}\widehat{KF}(P(2,\overline{3},2j-1))\nonumber\\
&&\hspace{-4mm}+z^{2}a^{2j+2}\sum_{i=1}^{j}a^{-2i}\widehat{KF}(P(2,\overline{3},2i-1))-a^{2j}(1-a^{-2}+z^{2})\widehat{KF}_{2,-3}.
\ea
This is the same expression as for the recursive formula of its HOMFLY--PT polynomial  found in\footnote{Note that in \cite{Petrou} we considered the mirror family $P(\bar{2},3,\overline{2j+1})$  and hence the HOMFLY--PT  polynomial of $P(2,\bar{3},2j+1)$ is related to that via $(a,z)\mapsto (a^{-1},-z)$.} \cite{Petrou}. That is, all the extra factors arising from the last term in the Kauffman skein relation (\ref{KauffSkein}), which is not present in the HOMFLY--PT skein relation (\ref{skein}), cancel. 
Using the fact that (\ref{KFandH1}) holds for the torus knots $3_1$ and $5_1$, the proof can be  completed by induction.

The proof is very similar for the second pretzel family $P(2,\bar{3},\overline{2j+1})$, which has Kauffman polynomial
\ba
\nonumber
&&\hspace{-15mm}\widetilde{KF}(P(2,\bar{3},\overline{2j+1});a,z)=a^{-2}\widetilde{KF}(P(2,\bar{3},\overline{2j-1});a,z)\\\nonumber
&&+z^2a^{-2j-2}\sum_{i=1}^{j}a^{2i}\widetilde{KF}(P(2,\bar{3},\overline{2i-1});a,z)-z^2a^{-2j-4}\sum_{i=1}^{j}a^{-2i}\\
&&+za^{-2j-3}(1-(a-a^{-1})(z+z^{-1}))\widetilde{KF}_{2,-3}(a,z)+za^{-4j-5},
\ea
 yielding
\ba
&&\hspace{-12mm}\widehat{KF}(P(2,\overline{3},\overline{2j+1}))=a^{-2}\widehat{KF}(P(2,\overline{3},\overline{2j-1}))\nonumber\\
&&\hspace{-8mm}+z^{2}a^{-2j-2}\sum_{i=1}^{j}a^{2i}\widehat{KF}(P(2,\overline{3},\overline{2i-1}))+a^{-2j-2}(1-a^{-2}+z^{2})\widehat{KF}_{2,-3}.
\ea
Again, this agrees with the recursive formula for its HOMFLY--PT polynomial\footnote{Note that this is an alternative but equivalent recursive formula to (\ref{P232j+1}) for $H(P(\overline{2},3,2j+1))$ (after $(a,z)\mapsto (a^{-1},-z)$), which was expressed in terms of torus knots only. }
\ba
&&\hspace{-12mm}H(P(2,\overline{3},\overline{2j+1}))=a^{-2}H(P(2,\overline{3},\overline{2j-1}))\nonumber\\
&&\hspace{-4mm}+z^{2}a^{-2j-2}\sum_{i=1}^{j}a^{2i}H(P(2,\overline{3},\overline{2i-1}))+a^{-2j-2}(1-a^{-2}+z^{2})H_{2,-3}.
\ea

(ii) The  Kauffman polynomial of the family $\mathcal{K}_{j,2}$ with $j\geq 0$, which is the closure of the braid $\sigma_2^{-2j-1}\sigma_1(\sigma_2^2\sigma_1^2)^2$, can be obtained recursively by 
\ba
\nonumber
&&\hspace{-12mm}\widetilde{KF}(\mathcal{K}_{j,2})=a^{-2}\widetilde{KF}(\mathcal{K}_{j-1,2})+z^2a^{-2j-2}\sum_{i=1}^{j}a^{2i}\widetilde{KF}(\mathcal{K}_{i-1,2})\\\nonumber
&&-z^2a^{-2j+8}\sum_{i=1}^{j}a^{-2i}-z^2a^{-2j}\widetilde{KF}(T(3,4))+z^2a^{-2j+6}\widetilde{KF}(5_2^+)\nonumber\\
&&+za^{-2j+3}(-1+(a-a^{-1})(z+z^{-1}))\widetilde{KF}(5_1^+)+za^{-4j+7}.
\ea
This yields
\ba\label{eq:Kauffjk}
&&\hspace{-12mm}\widehat{KF}(\mathcal{K}_{j,2})=a^{-2}\widehat{KF}(\mathcal{K}_{j-1,2})+z^{2}a^{-2j-2}\sum_{i=1}^{j}a^{2i}\widehat{KF}(\mathcal{K}_{i-1,2})\nonumber\\
&&\hspace{-4mm}-a^{-2j+2}(1-a^{2}+z^{2})\widehat{KF}(5_1^+)-z^2a^{-2j}\widehat{KF}(T(3,4))\nonumber\\
&&\hspace{-4mm}+z^2a^{-2j+4}\widehat{KF}(5_1^+)-za^{-2j+3}\widetilde{KF}_{\rm odd}(5_1^+)+z^2a^{-2j+6}\widehat{KF}(5_2^+)\nonumber\\
&&\hspace{-4mm}+z^2/(a-a^{-1})\left(a^{-2j+3}\widetilde{KF}_{\rm even}(5_1^+)-a^{-4j+7}\right)-z^2a^{-2j+8}\sum_{i=1}^ja^{-2i}.
\ea
Using  $\widetilde{KF}_{\rm odd}(5_1^+)=(2a^5-a^7-a^9)z+(a^5-a^7)z^3$, $\widetilde{KF}_{\rm even}(5_1^+)=3a^4-2a^6+(4a^4-3a^6-a^8)z^2+(a^4-a^6)z^4$ and $\widehat{KF}(5_2^+)=H(5_2^+)=a^2+a^4-a^6+(a^2+a^4)z^2$, one can easily show that the terms on the last two lines exactly cancel and hence (\ref{eq:Kauffjk})  has the  same form as the recursive formula for the HOMFLY--PT polynomial in (\ref{eq:recursivejk}). Since the theorem is valid for the knots $5_1$, $5_2$ and $T(3,4)=P(\bar{2},3,3)$, the proof can again be completed by induction.

Similarly, the family $\mathcal{K}_{j,3}$ ($j\geq 0$)  
has the recursive formula for its Kauffman polynomial 
\ba
\nonumber
&&\hspace{-7mm}\widetilde{KF}(\mathcal{K}_{j,3})=a^{-2}\widetilde{KF}(\mathcal{K}_{j-1,3})+z^2a^{-2j-2}\sum_{i=1}^{j}a^{2i}\widetilde{KF}(\mathcal{K}_{i-1,3})\\\nonumber
&&-z^2a^{-2j+12}\sum_{i=1}^{j}a^{-2i}-z^2a^{-2j+2}\widetilde{KF}(P(\bar{2},3,5))-z^2a^{-2j}\widetilde{KF}(P(\bar{2},3,7))\nonumber\\
&&+z^2a^{-2j+10}\widetilde{KF}(7_2^+)+z^2a^{-2j+8}\widetilde{KF}(7_3^+)+za^{-4j+11}\nonumber\\
&&+za^{-2j+5}(-1+(a-a^{-1})(z+z^{-1}))\widetilde{KF}(7_1^+),
\ea
which yields
\ba\label{eq:Kauffjk3}
&&\hspace{-12mm}\widehat{KF}(\mathcal{K}_{j,3})=a^{-2}\widehat{KF}(\mathcal{K}_{j-1,3})+z^{2}a^{-2j-2}\sum_{i=1}^{j}a^{2i}\widehat{KF}(\mathcal{K}_{i-1,3})\nonumber\\
&&-z^2a^{-2j+2}\widehat{KF}(P(\bar{2},3,5))-z^2a^{-2j}\widehat{KF}(P(\bar{2},3,7))\nonumber\\
&&-a^{-2j+4}(1-a^{2}+z^{2})\widehat{KF}(7_1^+)\nonumber\\
&&+z^2a^{-2j+6}\widehat{KF}(7_1^+)-za^{-2j+5}\widetilde{KF}_{\rm odd}(7_1^+)\nonumber\\
&&+z^2/(a-a^{-1})\left(a^{-2j+5}\widetilde{KF}_{\rm even}(7_1^+)-a^{-4j+11}\right)\nonumber\\
&&+z^2a^{-2j+10}\widehat{KF}(7_2^+)+z^2a^{-2j+8}\widehat{KF}(7_3^+)-z^2a^{-2j+12}\sum_{i=1}^ja^{-2i}.
\ea
Using that $\widetilde{KF}_{\rm even}(7_1^+)=4 a^6 - 3 a^8 + 10 a^6 z^2 - 7 a^8 z^2 - 2 a^{10} z^2 - a^{12} z^2 + 
 6 a^6 z^4 - 5 a^8 z^4 - a^{10} z^4 + a^6 z^6 - a^8 z^6$, $\widetilde{KF}_{\rm odd}(7_1^+)=3 a^7 z - a^9 z - a^{11} z - a^{13} z + 4 a^7 z^3 - 3 a^9 z^3 - a^{11} z^3 +
  a^7 z^5 - a^9 z^5$ and $\widehat{KF}(7_2^+)=a^2 + a^6 - a^8 + a^2 z^2 + 3 a^6 z^2 - a^8 z^2 + a^6 z^4$, $\widehat{KF}(7_3^+)=a^4 + 2 a^6 - 2 a^8 + 3 a^4 z^2 + 4 a^6 z^2 - 3 a^8 z^2 + a^{10} z^2 + 
 a^4 z^4 + a^6 z^4 - a^8 z^4$ and after some  algebra, one can show that the terms in the last 3 lines in (\ref{eq:Kauffjk3}) exactly cancel and hence it  has the  same form as the recursive formula for the HOMFLY--PT polynomial in (\ref{eq:recursivej3}), which only involves knots that satisfy Theorem~\ref{thm:}. For the last 2 families the proof can be extended to include the cases with $j<0$, whose recursive formulas will differ slightly from (\ref{eq:Kauffjk}) and (\ref{eq:Kauffjk3}) in a similar way as they do for the two pretzel families.

 (iii) For the family $5_2^{3k}$, obtained by $k $ full twists $F_3$ on $5_2$, we find
  \ba
\nonumber
&&\hspace{-8mm}\widetilde{KF}(5_2^{3k})=a^{2}(1+z^{2})\left(\widetilde{KF}_{T(3,3k+1)}+a^{4}z^{2}\sum_{i=0}^{3(k-1)}a^{2i}\widetilde{KF}_{T(3,3k-1-i)}\right)\nonumber\\
    &&
    +a^{6k+2}(1+z^{2})(1-a^{2}+z^2)+a^{6}z^{2}\widetilde{KF}_{T(3,3k-1)}\nonumber\\
    &&-z^2a^{8k+5}(z+a)+za^{6k+3}(1+z^2)(1-za)-za^{8(k+1)}(z+a)\widetilde{KF}(3_1^-)\nonumber\\
    &&-z^2(1+z^2)a^{8k+4}\left(1+\left(3-\frac{a-a^{-1}}{z}\right)\sum_{i=1}^{k-1}a^{-2i}\right),
\ea
with which we get
\ba\label{KFhat523k}
\hspace{-2mm}\widehat{KF}(5_2^{3k})&=&a^{2}(1+z^{2})\left(\widehat{KF}_{T(3,3k+1)}+a^{4}z^{2}\sum_{i=0}^{3(k-1)}a^{2i}\widehat{KF}_{T(3,3k-1-i)}\right)\nonumber\\
&&+a^{6k+2}(1+z^{2})(1-a^2+z^2)+a^{6}z^{2}\widehat{KF}_{T(3,3k-1)}\nonumber\\
&&-\frac{3z^2}{a-a^{-1}}(1+z^2)a^{3+6k}(a^{2k}-a^2).
\ea
This recursive formula is almost the same as (\ref{523krec}) for the HOMFLY--PT polynomial of $5_2^{3k}$, apart from an extra term in the last line. This is present due to the fact that $\widehat{KF}(T(3,3k))\neq H(T(3,3k)$ for 3-stranded torus links, which are included in the summation of the first line whenever $i \mod 3=2$. 
 To redeem this, we compute
 \be
 H(T(3,3k))=za\sum_{i=0}^{3(k-1)}a^{2i}H(T(3,3k-1-i,2,1))+\frac{a^{6k-2}}{z^2}(1-a^2+z^2)^2
 \ee
 and
 \ba
\nonumber
\widetilde{KF}(T(3,3k))&=&-za\sum_{i=0}^{3(k-1)}a^{2i}\widetilde{KF}(T(3,3k-1-i,2,1))\\\nonumber
&&+a^{6k}(1-(a-a^{-1})(z+z^{-1}))^2+z^2a^{8k}(1-a^{-2k+2}) 
\nonumber\\
&&+za^{6k-1}\left(a^2+3a^{2k}\sum_{i=0}^{k-2}a^{-2i}\right)\left(-1+\frac{a-a^{-1}}{z}\right).
\ea
The latter yields
\ba
&&\hspace{-8mm}\widehat{KF}(T(3,3k))=za\sum_{i=0}^{3(k-1)}a^{2i}\widehat{\widehat{KF}}(T(3,3k-1-i,2,1))\\
&&+a^{6k}(1+(a-a^{-1})^2(z+z^{-1})^2)+(a-a^{-1})a^{6k+1}+\frac{z^2}{a-a^{-1}}a^{6k+1}\nonumber\\
&&+a^{8k}(1-a^{-2k+2})(z^2+3)+3\frac{z^2}{a-a^{-1}}a^{8k-1}\sum_{i=0}^{k-2}a^{-2i}+2a^{6k}(z^2+1),\nonumber
\ea
where $\widehat{\widehat{KF}}(\mathcal{L})=\frac{z}{a-a^{-1}}\widetilde{KF}_{\rm even}(\mathcal{L})-\widetilde{KF}_{\rm odd}(\mathcal{L})$. By further using (\ref{HandKFT3n21}) one can show that 
\ba
&&\hspace{-12mm}\sum_{i=0}^{3(k-1)}a^{2i}\widehat{\widehat{KF}}(T(3,3k-1-i,2,1))=\sum_{i=0}^{3(k-1)}a^{2i}H(T(3,3k-1-i,2,1))\nonumber\\
&&\hspace{20mm}-\frac{za^{8k}}{a-a^{-1}}(\sum_{i=0}^{k-1}a^{-2i}+a^{-2k}\sum_{i=0}^{k-2}a^{2i}+a^{-2}\sum_{i=0}^{k-2}a^{-2i}).
\ea
It is then straight forward to compute $\widehat{KF}(T(3,3k))-H(T(3,3k))=3a^{8k}$.
Taking this into account and using that $\widehat{KF}_{m,n}=H_{m,n}$ for torus knots, the recursive formula (\ref{KFhat523k}) becomes exactly equal to (\ref{523krec}) for the HOMFLY--PT polynomial of $5_2^{3k}$.
\end{proof}

 Using the data in \cite{KnotInfo}, one can further easily confirm that (\ref{KFandH1}) is also satisfied by the knots $10_{128}$, $10_{132}$, $10_{139}$ and $12n_{318}$, which, together with the ones included in the above families, complete the list with up to 13 crossings for which the HZ transform of their HOMFLY--PT polynomial factorises. 
We have also checked that some other factorised cases, such as  twisted torus knots 
also satisfy (\ref{KFandH1}) and, in general, no exception has been found so far. 
 In fact, we have further confirmed that, for knots with up to 12 crossings, (\ref{KFandH1}) is valid if and only if the HZ transform of their HOMFLY--PT polynomial factorises. 
This motivates the following.
\begin{conj}\label{HomflyANDKauff}
 The HOMFLY--PT and Kauffman polynomials for a knot $\mathcal{K}$ satisfy (\ref{KFandH1})   if and only if the HZ transform of its HOMFLY--PT polynomial admits a factorised  form (\ref{eq:factorisableZ}), i.e.
      
      \be\label{KFandH}
    \boxed{ H(\mathcal{K};a,z) =\widehat{KF}(\mathcal{K};a,z)\iff Z(\mathcal{K};\lambda,q)=\frac{\lambda\prod_{i=0}^{m-2}(1-\lambda q^{\alpha_i})}{\prod_{i=0}^{m}(1-\lambda q^{\beta_i})}.}
      \ee
      
      \end{conj}

 \vskip2mm
 For torus knots, there is also a relation between the Kauffman polynomial $\widetilde{KF}(a,z)$ and the Alexander-Conway polynomial $\Delta (z)$, which can be derived from (\ref{KF}) \cite{Labastida}.
As a consequence of Theorem~\ref{HomflyANDKauff}, since the Alexander-Conway polynomial corresponds to the  $a\to1$ limit of the HOMFLY--PT polynomial, we have 
\begin{cor}
The relation   
between the Alexander-Conway and Kauffman polynomials
 \be\label{AKF}
 \Delta(\mathcal{K},z) = 1- \frac{z}{2} \frac{\partial}{\partial a} \widetilde{KF}_{\rm odd}(\mathcal{K};a,z)\bigg\vert_{a=1}.
 \ee
 which is true for torus knots \cite{Labastida}, also holds for the hyperbolic families $\mathcal{K}=P(2,-3,\pm(2j+1))$, $\mathcal{K}_{j,2}$, $\mathcal{K}_{j,3}$ and $5_2^{3k}$.
\end{cor}

We have further tested that this relation holds for all the remaining up-to-12-crossings knots with factorised HZ transform. 
As a natural consequence of the $(\Leftarrow)$ part of conjecture~\ref{HomflyANDKauff},    (\ref{AKF})  should remain valid for all the hyperbolic families  with  factorised HZ transform. However, ($\Rightarrow$) does not hold when restricting to the Alexander case ($a=1$), since the hyperbolic knots  $9_{28}$, $11a_3$ and $12n_{428}$ whose HZ transform does not factorise, do satisfy (\ref{AKF}).

Some consistency checks for the Conjecture~\ref{HomflyANDKauff} can be provided by the following considerations. 
First, note that when the HOMFLY--PT polynomial of a knot has a factorised  HZ transform, it means that it is expressible in terms of the HZ parameters $\{\alpha_i\},m,e$ as in  (\ref{FactorisedHOMFLY}), which is an expansion in powers of $a=q^N$ with $q$-polynomial coefficients. 
Alternatively, the unnormalised HOMFLY--PT polynomial of a knot\footnote{More generally, for a link with $n$ components the lowest power of $z$ in $\bar{H}$ is $-n$.} can be written as an expansion in powers of $z$ as 

\be\label{eq:Hexpansion}
\bar{H}=\sum_{l=-1}^{\nu}p_l(a)z^{2l+1},
\ee
in which the coefficients $p_l$ are polynomial in $a$ and $\nu\in\mathbb{Z}$ is a constant. Similarly, we can expand $z^{-1}(a-a^{-1})\widetilde{KF}_{\rm even}=\sum_{l=-1}^{\nu^{\rm even}}\kappa_l^{\rm even}(a)z^{2l+1}$ and $\widetilde{KF}_{\rm odd}=\sum_{l=0}^{\nu^{\rm odd}}\kappa_l^{\rm odd}(a)z^{2l+1}$. For all knots it holds that $p_{-1}(a)=\kappa_{-1}^{\rm even}(a)$ \cite{MarinoKauffman}, which in the factorised cases can be expressed in terms of the HZ parameters as 
   \be p_{-1}(a)=\sum_{j=0}^ma^{e+2j}(-1)^{j+1}\frac{2}{(2m)!!}\binom{m}{j}\prod_{i=0}^{m-2}(\alpha_i-e-2j).
   \ee 
For  (\ref{KFandH1}) to be valid it is required that $\forall l>-1$ s.t. $2l+1\leq\nu$ $\implies \kappa_l^{\rm even}(a)-\kappa_l^{\rm odd}(a)=p_l(a)$, while $\forall l$ s.t. $2l+1>\nu$ $\implies \kappa_l^{\rm even}(a)=\kappa_l^{\rm odd}(a)$. For the latter to be true, a \emph{necessary (but not sufficient) condition} is  that 
\be\label{eq:KauffNecCond}
\nu^{\rm even}=\nu^{\rm odd},
\ee
which can happen only if \emph{the Kauffman polynomial has even highest power of $z$}.  For alternating knots the highest power of $z$ is equal to $n-1$, where $n$ is the number of crossings \cite{Kauffman}. This implies that no alternating knot with even number of crossings can satisfy (\ref{KFandH1}). In the light of conjecture \ref{HomflyANDKauff}, this also becomes \emph{a necessary  condition for factorisability of the HZ transform}. 
Indeed, the list of up-to-13-crossings hyperbolic knots with factorised HZ transform (see Sec.~\ref{sec:factorisedHZ}) contains only non-alternating knots with even number of crossings\footnote{Note that for odd number of crossings the only knots with factorised HZ are alternating and include $5_2$ and the torus knots $T(2,n)$, which are excluded from that list. Up to $n=9$  the latter are listed as $n_{1}$ (with $n$ odd) in the Rolfsen table, and extend to $11a_{367}$ and $13a_{4878}$ \cite{KnotInfo}.}. The highest power of $z$ in the Kauffman polynomial for non-alternating knots is not known in general, but for up to 10-crossing-knots it is equal to $n-2$ (with the only exception being $10_{161}$ for which it is $6=n-4$) and hence satisfy the necessary condition (\ref{eq:KauffNecCond}).   

 The above discussion applies only to knots. However, it is also of interest to examine the relation between the HOMFLY--PT and Kauffman polynomials in the case of links, with multiple components. 

\begin{remark}
\textup{    The HOMFLY-PT polynomial $H(\mathcal{K};a,z)$ and $\widehat{KF}(\mathcal{K};a,z)$ for knots or links with odd number of components are always even Laurent polynomials in $a$. 
However, for  links with even number of components, $H(\mathcal{K};a,z)$ involves only odd powers of $a$, and hence (\ref{KFandH1})
cannot hold. 
 As a consequence, while for links with odd number of components (including knots) we compare $H$ with 
  \be\label{KFhat}
 \widehat{KF}(\mathcal{L}):=\widetilde{KF}_{\rm even}(\mathcal{L})-\frac{z}{a-a^{-1}}\widetilde{KF}_{\rm odd}(\mathcal{L}),
 \ee
 when the number of components is even, we should instead compare  $H$ with} 
\be\label{KFhathat}
\widehat{\widehat{KF}}(\mathcal{L}):=\frac{z}{a-a^{-1}}\widetilde{KF}_{\rm even}(\mathcal{L})-\widetilde{KF}_{\rm odd}(\mathcal{L}).
\ee
\end{remark}

\paragraph{2-component links.} As discussed in Sec.~\ref{sec:factorisedHZ},  all the links 
 admitting a factorisable HZ transform known to us, have 2 components. 
\textup{The HOMFLY--PT polynomial for the majority of these HZ-factorisable 2-component links $\mathcal{L}= (\mathcal{K}_1,\mathcal{K}_2)$, 
satisfy the following relation to their Kauffman polynomial 
\be\label{linkcriterion}
\boxed{\frac{a-a^{-1}}{z}(H(\mathcal{L})-\widehat{\widehat{KF}}(\mathcal{L})) = a^{4lk(\mathcal{L})}.}
\ee
Here $lk(\mathcal{L})$
is the linking number, which is obtained as a sum of  $\pm\frac{1}{2}$ over the crossings between ${\rm \mathcal{K}_1\; and\; \mathcal{K}_2}$.}  
\noindent{}We have explicitly verified that this holds\footnote{As a reminder,  our conventions differ by the data for the Kauffman polynomial in  \cite{KnotInfo} by $a\mapsto a^{-1}$.} for several such links from Sec.~\ref{sec:factorisedHZ}, such as $L7n2\{0\}=L7n2\{1\}$  (with $lk=0$), $L9n12\{1\}$ ($lk=2$),  $L9n14\{0\}$ ($lk=-1$), $L11n132\{0\}^+$  ($lk=1$), $L11n133\{0\}^+$ ($lk=3$), 
$L11n208\{0\}$ ($lk=2$), $T(4,5,2,1)$ ($lk=4$) and for the ones included in the following families.
 \vskip1mm

$\bullet$ The 2-stranded torus links with parallel orientation $T(2,2k)$ with linking number $lk(T(2,2k))=k$. These include $L2a1 (k=1)$, $L4a1\{1\} (k=2)$, $L6a3\{0\} (k=3)$,
$L8a14\{0\} (k=4)$, $L10a118\{0\} (k=5)$. For instance, in the case of $L2a1\{1\}$  with $lk=1$ the r.h.s. of (\ref{linkcriterion}) becomes $a^4$. 
The opposite choice of orientation $L2a1\{0\}$, results in a HOMFLY-PT polynomial $H(L2a1\{0\})$ that differs from $H(L2a1\{1\})$ only by $a\to -a^{-1}$ and hence it still has a factorised HZ transform. In this case
the r.h.s becomes $a^{-4}$, in agreement with $lk(L2a1\{0\})=-1$.

 $\bullet$ The twisted torus links $T(3,n,2,1)$ satisfy 
\be\label{HandKFT3n21}
H(T(3,n,2,1)) = \widehat{\widehat{KF}}(T(3,n,2,1)) +\frac{z}{a-a^{-1}} a^{4lk(n)}
\ee
where $lk(3k)=2k$, $lk(3k+1)=2(k+1/2)$ and $lk(3k+2)=2(k+1)$. For instance, for $T(3,3,2,1)=L7n1\{0\}^+$: $lk(3)=2$, for $T(3,4,2,1)=L9n15\{0\}^+$: $lk(4)=3$ and for  $T(3,5,2,1)=L11n204\{0\}^+$: $lk(5)=4$.

Some implications from  (\ref{linkcriterion}) are the following. At $N=0$, i.e. at $a=q^N=1$, we obtain  that $\widetilde{KF}_{\rm even}=-1$ since $\widetilde{KF}_{\rm odd}$ and $H$, which are multiplied by $(a-a^{-1})z^{-1}$, both vanish.
For $a=q$ $(N=1)$, the HOMFLY--PT polynomial $H$  and the factor $(a-a^{-1})/z$ become 1, yielding $-1+\widetilde{KF}_{\rm odd}-\widetilde{KF}_{\rm even}-q^{4lk(\mathcal{L})}=0$. 

 However, the link 
$L10n42\{1\}$, despite admitting HZ factorisability, is an exception, which does not satisfy  (\ref{linkcriterion}), since evaluation of the l.h.s. yields 
\ba\label{10n42}
&&a^6+a^4 \left(-z^6-6 z^4-10 z^2-5\right)+\frac{z^6+6 z^4+10 z^2+5}{a^2}\\\nonumber
&&+a^2 \left(z^8+9 z^6+27 z^4+30
   z^2+10\right)-z^8-9 z^6-27 z^4-30 z^2-10.\ea
This can alternatively be written as $H-\widetilde{KF}_{odd}+ \frac{z}{a-a^{-1}}(\widetilde{KF}_{even}+ a^{-6}) = a^{-3}q^{-9} (\frac{1-q^{10}}{1+q^2})((1+a^4)(q^2+q^8)-a^2(1+q^{10}))$ and yields $H=\widetilde{KF}_{odd}$ at $q=1$ ($z=0$). It is noteworthy that, although the HOMFLY--PT polynomial can not distinguish this link from $L9n14\{0\}$, their  Kauffman polynomials $\widehat{\widehat{KF}}$ are very different. 
The reason for this link being exceptional remains mysterious and deserves further investigation. 
The breaking of (\ref{linkcriterion})  in this case, implies that a relation between the HOMFLY--PT and Kauffman polynomials for links is no longer in 1-1 correspondence with HZ factorisability, but it  can still serve as a criterion for it, according to the following.
\begin{conj}\label{HomflyANDKauffLINKS}
 If the HOMFLY--PT and Kauffman polynomials for a 2-component link $\mathcal{L}$ satisfy (\ref{linkcriterion}), then the HZ transform of its HOMFLY--PT polynomial admits a factorised  form,
  i.e.
     \be\label{conjKFlink}
 \boxed{ H(\mathcal{L})-\widehat{\widehat{KF}}(\mathcal{L}) = \frac{za^{4\text{lk}(\mathcal{L})}}{{a-a^{-1}}}\implies \textup{HZ factorisation}}
 \ee

      \end{conj}

\paragraph{3-component links.} The 3-stranded torus links, 
although their HZ transform is not factorisable, satisfy the simple relation
\be\label{L6n1x}
H(T(3,3k))-\widehat{KF}(T(3,3k))=-3a^{8k},
\ee
as proven in  (iii) of Theorem~\ref{thm:}. 
For instance, $T(3,3)=L6n1\{0,1\}$ has a simple, but not factorisable HZ transform of its HOMFLY--PT polynomial, as given in (\ref{L6n1}).
Different choices of orientation for the same link
are labeled as $L6n1\{0,0\},$ $L6n1\{1,0\},$ $ L6n1\{1,1\}$ and have the same HOMFLY--PT polynomial, whose HZ transform becomes $Z = \lambda(1+ (q^{9}+q^{7}+q+q^{-1}) \lambda + q^{8}\lambda^2)/(1-q^{-1}\lambda)(1-q \lambda)(1-q^{3}\lambda)(1-q^{5}\lambda)$. 
They also have the same Kauffman polynomial, which relates to their HOMFLY--PT as
$H - \widehat{KF}= -1 - a^{4}(2+ 4 z^2 + z^4)$.

\paragraph{4-component links.}For the torus link $T(4,4)$, which is the closure of $(\sigma_1\sigma_2\sigma_3)^4$ and its HZ transform does not factorse, we find 
\be
H(T(4,4))-\widehat{\widehat{KF}}(T(4,4))=\frac{a^{12}}{(a-a^{-1})z}(6+4z^2+3a^2(2+z^2)(-2+a^2)),
\ee
while for $T(4,8)$ we get
\ba
&&\hspace{-15mm}H(T(4,8))-\widehat{\widehat{KF}}(T(4,8))=\frac{a^{24}}{\left(a-a^{-1}\right) z} \bigg(z^2(3 a^8-2) \nonumber\\
&&+6 a^6 \left(1+z^2\right)-6 a^4 \left(2+4 z^2+z^4\right)+6 a^2 \left(1+3 z^2+z^4\right)\bigg).
\ea
Thus we observe that, from 4 components, not only the HZ transform for torus links is not factorisable, but also they no longer have a simple relation between their HOMFLY--PT and Kauffman polynomials. For both of the above links it holds that $\widetilde{KF}_{\rm even}(a=1)=-1$.

\subsection{BPS invariants and HZ factorisability\label{sub:BPS}}

There is an intriguing relation between  link polynomials and BPS invariants of topological strings, known as the LMOV conjecture\footnote{The initials BPS and LMOV stand for Bogomol'nyi–Prasad–Sommerfield and Labastida-Mariño-Ooguri-Vafa, respectively.
}~\cite{Labastida,OoguriVafa, LabastidaMarino}. 
Such a relation is based on the gauge/string duality put forward in \cite{GopakumarVafa1998} and suggests that the coefficients of knot invariants count  the number of BPS states of open topological strings. 
In this context, the HOMFLY--PT/Kauffman relation presented in the previous section, corresponds to the vanishing of the two-crosscap BPS invariants $\hat{N}^{c=2}=0$ of  unoriented strings, 
as we shall explain below.

\paragraph{BPS invariants for knots.}
The unnormalised HOMFLY--PT polynomial of knots, corresponding to the $SU(N)$-invariant in CS theory, can be expanded as 
\be\label{orientable1}
 \bar H({\mathcal{K}};a,z)= \sum_{Q}\sum_{g= 0}^{g_{\rm max}}\hat{N}_{g,Q} a^{Q} z^{2g-1},
 \ee 
in which the coefficients $\hat{N}_{g,Q}\in\mathbb{Z}$ are constant integers, while, for an $m$-strand braid,  the index $Q$ can be replaced by $e+2i$ with $i\in\{0,m\}$ (recall that $e$ is the lowest power of $a$ in $\bar H$). Expressed in this way, the fundamental\footnote{Beyond the fundamental case, the polynomials associated to topological strings are an extension of ordinary knot polynomials involving correction terms \cite{LabastidaMarino}. In such cases, the integer coefficients are denoted by $N_{R,g,Q}$, where the subscript $R$ denotes a representation, usually indicated by the corresponding Young diagram (see e.g. \cite{MorozovCheck,Gukov}) and can be interpreted as the magnetic charge. 
 Here, since we restrict only to the fundamental representation $R=[1]=\square$,  we shall omit the additional subscript on the coefficients $\hat{N}_{g,Q}$ for simplicity.} HOMFLY--PT polynomial can be interpreted   as a generating function  which serves to solve enumerative problems in the geometry of topological string theory.  Namely, the integer coefficients $\hat{N}_{g,Q}$ in (\ref{orientable1}) 
 correspond to the counting of BPS states of open  topological strings, which 
 determine the oriented  string amplitudes \cite{OoguriVafa}.  
 The topological string interpretations of the indices 
 $g$ and $Q$ are the spin and the $D2$ brane charge, respectively \cite{LabastidaMarino}. 

 Some basic properties of $\hat{N}_{g,i}$ can be derived from (\ref{orientable1}). First, note that 
 $\bar{H}(a=1,z)$, which is the unnormalised Alexander polynomial, is always vanishing
 due to the overall factor $a-a^{-1}$ that it includes\footnote{Since $a=1$ corresponds to $N=0$, this is also the reason that there is no constant $\lambda$-term in the HZ transform \cite{Morozov}.}, and hence (\ref{orientable1})with $Q=e+2i$ for $i\in\{0,m\}$, yields  
 \be\label{N=0contribution}
 \sum_{i=0}^m\sum_{g\ge 0}^{g_{\rm max}}\hat{N}_{g,i} z^{2g-1}=0.
 \ee
 Since $z$ is apriori an arbitrary parameter, this implies that at each  $g$ 
 \be\label{sum0} \sum_{i=0}^m \hat N_{g,i} = 0.
 \ee 
Noting that the derivative  $\frac{\partial}{\partial a}\bar{H}\vert_{a=1}=\frac{2}{z} H$ 
and since $H(a,z)\vert_{z=0}=1$,
by taking the $a$-derivative of (\ref{orientable1}) and using (\ref{N=0contribution}) we find that in the $z\to 0$ limit
\be\label{sum1}
\sum_{i=0}^m
i \hat N_{0,i}= 1.
\ee
 
Tables~\ref{tab:BPSknots}a-c below  include the values of $\hat N_{g,i}$ and $e$ for the examples of $3_1^+$, $5_2^+$ 
and $8_{20}^+$.
\begin{table}[hbtp] 
   \caption{Examples of BPS states $\hat{N}_{g,i}$ for three HZ-factorisable knots.}
\begin{minipage}{0.25\linewidth}
%
\centering
\subcaption{$[3_1^+]$; $e=1$}
\begin{tabular}{|l|lcc|}
\hline
g \textbackslash i&0&1&2\\
\hline
0&  -2& 3 & -1    \\
1&  -1 &  1& 0 \\
\hline
\end{tabular}
\end{minipage}
\hfill
\hskip2mm
\begin{minipage}{0.35\linewidth}
\centering
\subcaption{$[5_2^+]$; $e=1$}
\begin{tabular}{|l|lccc |  }
\hline
g \textbackslash i&0&1&2&3\\
\hline
0&  -1& 0 & 2& -1    \\
1&  -1 &  0& 1& 0 \\
\hline
\end{tabular}
\end{minipage}
\hfill
\hspace{-4mm}
\begin{minipage}{0.35\linewidth}
\centering
\subcaption{$[8_{20}^+]$; $e=-1$} 
\begin{tabular}{|l|lccc |  }
\hline
g \textbackslash i&0&1&2&3\\
\hline
0&  1& -5 & 6& -2    \\
1&  1 &  -5& 5& -1 \\
2& 0& -1& 1& 0 \\
\hline
\end{tabular}
\end{minipage}
\label{tab:BPSknots}
\end{table} 
It can be easily checked that these examples satisfy
 (\ref{sum0}) and (\ref{sum1}).
 The  BPS invariants $\hat N_{g,i}$ in the case of $8_{20}^-$ were also computed in \cite{MorozovCheck}  
 and their values after $a\to  a^{-1}$, due to mirroring, and up to a sign 
 agree with  Table~\ref{tab:BPSknots}c.

 The Kauffman polynomial, corresponding to the gauge group $SO(N+1)$ or $Sp(N+1)$, may admit a similar 
topological string interpretation, according to the Mariño conjecture \cite{MarinoKauffman}. This  claims that the gauge/string duality for orthogonal or symplectic groups may be obtained from the unitary case 
by introducing an orbifold involution 
in the target space, 
combined with an orientation reversal in the string worldsheet, which  
produces unoriented strings. 
The  correlation functions  of the $SO$-gauge theory will then 
roughly correspond to the correlation functions in the orbifold 
string theory and 
will consist of a sum of two sectors, one involving oriented strings  and one with unoriented strings, which is  twisted by the orientifold.

For the unorientable part (twisted  sector), the string amplitudes  will have contributions from  curves 
  with 1 or 2 crosscaps $c$, which are encoded in the BPS invariants $\hat{N}_{g,Q}^{c=1}$ and $\hat{N}_{g,Q}^{c=2}$ \cite{MarinoKauffman}. 
A  
surface with 
$c=1$ is the real projective plane  $\mathbb{RP}^2$, 
while with two crosscaps $c=2$ it is the Klein bottle, which can be obtained from the connected sum of two projective planes as 
$ \mathbb{RP}^2\# \mathbb{RP}^2$ \cite{BrezinHikamianti}.

The 
generating function for the enumeration of unoriernted BPS invariants  
becomes 
  \cite{MarinoKauffman,Bouchard,Paul} 
  \be\label{hatg2}
  \hat g = \sum_{g,Q} (\hat{N}_{g,Q}^{c=1} z^{2g}a^{Q}+
  \hat{N}_{g,Q}^{c=2}z^{2g+1}a^Q).
  \ee
This can be expressed as 
a combination of the  unnormalised  HOMFLY--PT polynomial $\bar{H}$ and the (Dubrovnik version of the) unnormalised Kauffman polynomial $\widetilde{KF}$ (i.e. multiplied by the overall factor $\frac{a-a^{-1}}{z}-1$), using our conventions, as 
  \be\label{hatg}
  \hat g = \left(\frac{a-a^{-1}}{z}-1\right)(\widetilde {KF}_{\rm even}+\widetilde{KF}_{\rm odd})- \bar H.
  \ee 
Considering the parity of the powers of $a$ and $z$ in the knot polynomials, this can be split into the crosscap $c=1$ and $c=2$ parts, respectively, as 

  \be\label{c=1}
  \sum_{g,Q} \hat{N}_{g,Q}^{c=1}z^{2g} a^Q = - \left(\widetilde{KF}_{\rm even}- \frac{a-a^{-1}}{z}\widetilde{KF}_{\rm odd}\right)=-\widehat{\widehat{\overline{KF}}},
  \ee
where $\widehat{\widehat{KF}}$ is defined in (\ref{KFhathat})  and
\be\label{c=2}
\sum_{g,Q} \hat{N}_{g,Q}^{c=2} z^{2g-1} a^Q=\frac{a-a^{-1}}{z}\widetilde{KF}_{\rm even}-  \widetilde{KF}_{\rm odd}- \overline{ H }=\widehat{\overline{KF}}- \overline{ H }
.
 \ee
These formulas are consistent (up to  overall minus signs)
with the values of the Kauffman LMOV invariants  computed via the methods of \cite{MorozovCheck}, 
for all the knots (up to 8 crossings),  for which the data are available in 
 \cite{MorozovWebsite}. 
For instance,  in the example of $8_{20}^+$, the values of $\hat N_{g,Q}^{c=1}$ obtained via (\ref{c=1}) which are listed in  Table 4, which are listed in  Table~\ref{tab:BPS820}, match the result in 
\cite{MorozovCheck} 
\begin{table}[hbtp]
\caption{$\hat N_{g,Q}^{c=1}[8_{20}^+]$}
\centering

\label{tab:BPS820}
\begin{tabular}{|l|ccccc |}
\hline
\small{g \textbackslash Q}&-2&0&2&4&6\\
\hline
0&  -1 &5  & -12&10 & -3    \\
1&0 &  5&- 16& 15&-4 \\
2& 0& 1& -7& 7&-1 \\
3& 0& 0& -1& 1&0\\
\hline
\end{tabular}
\end{table}
\noindent{}and satisfy the relation $\sum_Q\hat N_{g,Q}^{c=1}=0$ for $g>0$.
From (\ref{c=2}),
the value of ${\hat N}_{g,Q}^{c=2}$ is found to be zero for $8_{20}$,
again in agreement with \cite{MorozovCheck}.
 Similarly, for the trefoil $3_1$, the r.h.s. of (\ref{c=1}) becomes $-(3 + z^2)a^2 + (3 +z^2)a^4 - a^6$, 
 which yields values of $\hat N_{g,Q}^{c=1}$,  consistent with the  result reported in \cite{Paul}.
 \begin{remark}
\textup{According to (\ref{c=2}),  the relation (\ref{KFandH1}) between the Kauffman and HOMFLY--PT polynomial of a knot is equivalent to the vanishing of the two-crosscaps BPS invariants $\hat{N}_{g,Q}^{c=2}$, i.e.
\be\label{BPSandHKF}
\boxed{H=\widehat{KF} \iff \hat{N}_{g,Q}^{c=2}=0.}
\ee
From a 
geometric point of view, this implies that there are no contributions from surfaces with two crosscaps, which are Klein bottles,  in the string amplitudes when the HOMFLY--PT/Kauffman relation is valid.}
 \end{remark} 
In agreement with Conjecture~\ref{HomflyANDKauff},  $\hat{N}^{c=2}_{g,Q}$ was found to be vanishing  for torus knots  \cite{Stevan} and for the 
hyperbolic knots  $5_2$, $8_{20}$, $10_{125}$, $10_{128}$, $10_{132}$, $10_{139}$ and $10_{161}$ \cite{MorozovCheck}, which are (up to 10 crossings) precisely the ones with factorised HZ transform.  

\paragraph{BPS invariants for two-component links.}For a 2-component link $\mathcal{L}=(\mathcal{K}_1,\mathcal{K}_2)$,
assuming the fundamental representation on each component\footnote{ As before, we omit the subscripts $\square,\square$ used in \cite{MorozovCheck} to label representations on each component.}, the BPS invariants corresponding to the orientable part $\hat{h}^{\mathcal{L}}$, given as the coefficients in the expansion $\hat{h}^{\mathcal{L}}= \sum_{g\geq 0,Q} \hat{N}^{c=0}_{g,Q}  z^{2g}a^{Q}$,
can be expressed in terms of the HOMFLY--PT polynomial of the link  and its components as \cite{MorozovCheck} 
 \be\label{hath}
 \hat{h}^{\mathcal{L}}= 2 a^{-2lk(\mathcal{L})}\bar{H}(\mathcal{L})+2 a^{2lk(\mathcal{L})}\bar{H}(\tilde{\mathcal{L}})-
 4 \bar{H}(\mathcal{K}_1) \bar{H}(\mathcal{K}_2),
 \ee
 where $lk(\mathcal{L})$ is the linking number of $\mathcal{L}$, while $\tilde{\mathcal{L}}$ denotes the link obtained from $\mathcal{L}$ by reversing the orientation of one of its components\footnote{The linking numbers of  $\tilde{\mathcal{L}}$ and $\mathcal{L}$ are related by $lk(\tilde{\mathcal{L}})=-lk(\mathcal{L})$. Note that for a fixed orientation, the linking number of a link $\mathcal{L}$ and its mirror image $\mathcal{L}^*$ are also related by  $lk(\mathcal{L}^*)=-lk(\mathcal{L})$. }. The expression (\ref{hath}) 
 includes the additional factors $a^{\pm 2\text{lk}(\mathcal{L})}$ as compared to the one presented in \cite{MorozovCheck}, which are missing there due to their different choice of conventions (see Eq.~2.16 in \cite{MarinoKauffman}). 
 The unorientable part  $ \hat{g}^{\mathcal{L}} =  \sum_{g\geq 0,Q} \left( \hat{N}^{c=1}_{g,Q} z^{2g+1}+\hat{N}^{c=2}_{g,Q} z^{2g+2} \right)a^{Q}$ can be expressed as
 \be\label{unoriented}
 \hat{g}^{\mathcal{L}}=
 a^{-2lk(\mathcal{L})}K^{\mathcal{L}}- K^{\mathcal{K}_1} K^{\mathcal{K}_2}-a^{-2lk(\mathcal{L})}\bar{H}^{\mathcal{L}} -a^{2lk(\mathcal{L})}\bar{H}^{\tilde{\mathcal{L}}}+ 2 \bar{H}^{\mathcal{K}_1} \bar{H}^{\mathcal{K}_2},
 \ee
 in which $K^{\mathcal{L}}$ corresponds to the (Dubrovnik version of the) unnormalised 
 Kauffman polynomial and hence it
 should be replaced by $(\frac{a-a^{-1}}{z}-1)(\widetilde{KF}_{\rm even}(\mathcal{L})+\widetilde{KF}_{\rm odd}(\mathcal{L}))$. By splitting  the terms with odd or even parity,
 writing these in terms of $$\widehat{\overline{KF}}=\frac{a-a^{-1}}{z}\widetilde{KF}_{\rm even}-\widetilde{KF}_{\rm odd}\;\;{\rm and}\;\; \widehat{\widehat{\overline{KF}}}=\widetilde{KF}_{\rm even}-\frac{a-a^{-1}}{z}\widetilde{KF}_{\rm odd},$$ 
 respectively, we find 
 \be\label{c1general}
 \sum_{g,Q} \hat{N}^{c=1}_{g,Q} z^{2g+1}a^Q= a^{-2lk(\mathcal{L})}\widehat{\overline{KF}}(\mathcal{L})+\widehat{\overline{KF}}(\mathcal{K}_1)\widehat{\widehat{\overline{KF}}}(\mathcal{K}_2)+\widehat{\widehat{\overline{KF}}}(\mathcal{K}_1)\widehat{\overline{KF}}(\mathcal{K}_2)
 \ee
   \ba\label{c2general}
 \sum_{g,Q}\hat{N}_{g,Q}^{c=2}
 z^{2g+2} a^Q\hspace{-2mm} &=&\hspace{-2mm}-a^{-2lk(\mathcal{L})}\widehat{\widehat{\overline{KF}}}(\mathcal{L})-\widehat{\widehat{\overline{KF}}}(\mathcal{K}_1)\widehat{\widehat{\overline{KF}}}(\mathcal{K}_2)-\widehat{\overline{KF}}(\mathcal{K}_1)\widehat{\overline{KF}}(\mathcal{K}_2)
 \nonumber\\
 &&\hspace{-2mm}-a^{-2lk(\mathcal{L})}\bar{H}(\mathcal{L})-a^{2lk(\mathcal{L})}\bar{H}(\tilde{\mathcal{L}})+2\bar{H}(\mathcal{K}_1)\bar{H}(\mathcal{K}_2).
 \ea
 \begin{remark}\label{rmk:K0vsK1}
 \textup{The BPS invariants for 2-component links are independent of orientation, since
 \be\label{KFreverseOrient}
 a^{-2lk(\mathcal{L})}  KF((\mathcal{L}))= a^{-2lk(\tilde{\mathcal{L}})} KF(\tilde{\mathcal{L}}).
 \ee}
  \end{remark}
\noindent{}The  identity (\ref{KFreverseOrient}) can be understood by recalling that the only orientation dependence in the Kauffman polynomial is through the overall factor $a^{w(\mathcal{L})}$, involving the writhe. For a link $\mathcal{L}=(\mathcal{K}_1,\mathcal{K}_2)$ the writhe counts twice the linking number $lk(\mathcal{K}_1,\mathcal{K}_2)$ (since the latter involves a factor $\frac{1}{2}$ in its definition) plus any additional self crossings of each of the components $\mathcal{K}_1$ and $\mathcal{K}_2$. Orientation reversing on one\footnote{The same identity can also apply to $l$-component links $\mathcal{L}=(\mathcal{K}_1,...,\mathcal{K}_l)$ upon orientation reversal on several of its components, only when these are chosen in such a way that the overall linking number $lk(\mathcal{K}_1,...,\mathcal{K}_l)$ (which is equal to the sum of the linking number of each pair of components) exactly reverses sign.}   component 
leaves the self crossings  unchanged, while the  linking number reverses sign.
 
 For the link
 $L7a3=(3_1^+,\bigcirc)$, which does not admit a factorisable HZ transform,
 the BPS invariants  can be evaluated from the above definition 
 and the obtained results
 agree with
 \cite{MorozovCheck}. Namely, using that $lk(L7a3)=0$ and hence $H(L7a3\{0\})=H(L7a3\{1\})$ and that $\bar{H}(3_1) \bar{H}(\bigcirc) =z^{-2}(a-a^{-1})^2 (-a^{4} + a^{2} z^2 + 2 a^{2})$, while noting that for 2-component links $\bar{H}$ is given by (\ref{SU(N)links}) (i.e. it is multiplied by an overall minus sign), 
 from (\ref{hath}) we find 
 \ba\label{Lc=0}
 &&\hspace{-6mm}\sum_{g,Q} \hat{N}^{c=0}_{g,Q} z^{2g} a^{Q} = 4 \biggl(\frac{a^{-1}-a}{z} H(L7a3;a,q) -\frac{(a-a^{-1})^2}{z^2} H(3_1;a,q)\biggr)\\\nonumber
 &&\hspace{-3mm}=(8 - 24 a^2 + 24 a^4-8 a^6)+ z^2 ( 4 -20 a^2 + 20 a^4-4 a^6) + z^4 (-4 a^2 + 4 a^4).
 \ea
 Note that  the subtraction by the term $z^{-2}(a-a^{-1})^2 H(3_1)$ completely cancels the term of order $z^{-2}$, which is always present in $\bar{H}$ for two-component links. 
The BPS invariant of  crosscap number two, i.e. $\hat{N}^{c=2}_{g,Q}$,  is given by (\ref{c2general}) using that $\bar{H}(3_1)=\widehat{\overline{KF}}(3_1)$ and $\widetilde{{KF}}(\bigcirc)=1$ 
\ba\label{Lc=1}
 \sum_{g,Q}\hat{N}_{g,Q}^{c=2}
 z^{2g+2} a^Q &=&-\widehat{\widehat{\overline{KF}}}(L7a3)-2\bar{H}(L7a3)\nonumber
\\
 &&\hspace{-20mm}-\widehat{\widehat{\overline{KF}}}(3_1)-\frac{a-a^{-1}}{z}\widehat{\overline{KF}}(3_1)+2\bar{H}(\bigcirc)\bar{H}(3_1)\nonumber\\
&&\hspace{-20mm}= -\widetilde{KF}(L7a3)_{\rm even}+ \frac{a-a^{-1}}{z}\widetilde{KF}(L7a3)_{\rm odd}+\frac{2(a-a^{-1})}{z}H(L7a3)\nonumber\\
 &&\hspace{-20mm}+  \left(\frac{a-a^{-1}}{z}\right)^2 H(3_1)- \widetilde{KF}(3_1)_{\rm even} +\frac{a-a^{-1}}{z}\widetilde{KF}(3_1)_{\rm odd}.
 \ea 
 The r.h.s. becomes $z^2(2 - 5 a^2+ 3 a^4+ a^6 - a^8)+ z^4 (1- 5 a^2 + 4 a^4)+z^6(-a^2+ a^4)$, which is the same as the one in\footnote{Note that in \cite{MorozovCheck} the values for $\hat{N}^{c=1}$ and $\hat{N}^{c=2}$ were accidentally swapped, as we confirmed with the authors in a private communication.} \cite{MorozovCheck}. Note that   the terms of order $z^{-2}$ and $z^{0}$ in the first three and last three terms are canceled with each other.
 The $c=1$ BPS invariant can be computed via (\ref{c1general}) to be
  \ba\label{Lc=2}
\sum_{g,Q} \hat{N}^{c=1}_{g,Q} z^{2g+1}a^Q&=& \widehat{\overline{KF}}(L7a3)+\widehat{\overline{KF}}(3_1)+\frac{a-a^{-1}}{z}\widehat{\widehat{\overline{KF}}}(3_1)\\
&=&-(-5a + 14 a^3 -12a^5 + 2 a^7 + a^9)z\nonumber\\
&&-(-2a + 11 a^3 -10 a^5 +a^7)z^3 - (2a^3-2a^5)z^5\nonumber
 \ea 
agreeing (up to an overall minus sign) with \cite{MorozovCheck}. 

Another example is the link $\mathcal{L}=L4a1\{1\}=T(2,4)=(\bigcirc,\bigcirc)$ with $lk(L4a1\{1\})=2$, for which $H(\mathcal{L})\neq H(\tilde{\mathcal{L}}=L4a1\{0\})$, is HZ-factorisable. Its 
 $c=0$ BPS invariants are the coefficients of the polynomial (\ref{hath}), which becomes  
 \ba\label{t24c0}
 \hat{h}&=&-2\left(\frac{a-a^{-1}}{z}\right)\left( a^{-4} H(L4a1\{1\})+ a^4 H(L4a1\{0\})\right)
 - 4 (\bar{H}(\bigcirc))^2 \nonumber\\
 &=& 6 a^{-2}- 10 + 2 a^2 + 2 a^4 + z^2 (-2 + 2 a^{-2}).
 \ea
  The BPS invariants for one crosscap $(c=1)$  are obtained from (\ref{c1general}) as
 \ba
 \sum_{g,Q} \hat N^{c=1} z^{2g+1} a^{Q}&=&a^{-4}\widehat{\overline{KF}}(L4a1\{1\}) 
 +2\frac{a-a^{-1}}{z}\nonumber\\
 &=& z \left(4a^{-1}- 4 a\right)+ z^3 \left(a^{-1}-a\right), 
 \ea
 in which the last term cancels the term of order $z^{-1}$.
 
For the  links $L6a3=T(2,6)=(\bigcirc,\bigcirc)$ with $\text{lk}(L6a3\{0\})=3$, and  $L7n1^+=T(3,3,2,1)=(3_1^+,\bigcirc)$ with $\text{lk}(L7n1\{0\}^+)=2$  the  BPS invariants 
 are similarly evaluated and are given in Tables~\ref{BPSL6a3} and Table~\ref{BPSL7n1}, respectively. The two crosscap BPS invariants  for $L4a1$, $L6a3$ and $L7n1^+$ are computed via (\ref{c2general}) to be 
  $\hat N^{c=2}=0.$  These evaluations  agree (up to some overall minus signs) with the 
 results provided to us, through a personal communication with the authors of \cite{MorozovCheck}.  Some further evaluations for links with factorised HZ transform are included in the Appendix.  
 
\begin{table}[h!]
\caption{BPS invariants for $L6a3$.}
\centering

\label{BPSL6a3}
\begin{tabular}{c}
\begin{minipage}{0.49\linewidth}
\centering
$\hat{N}^{c=0}_{g,Q}:$ \begin{tabular}{|l|lcccc| c |}
\hline
\small{g \textbackslash Q}&-2&0&2&4&6\\
\hline
0&  12& -20 & 6&0&2    \\
1&  10 & - 12& 2&0&0 \\
2&  2 & - 2& 0&0&0 \\
\hline

\hline
\end{tabular}
\end{minipage}
\hfill
\hspace{-6mm}
\begin{minipage}{0.49\linewidth}
\centering
$\hat{N}^{c=1}_{g,Q}:$ \begin{tabular}{|l|lc | c |}
\hline
\small{g \textbackslash Q}&-1&1\\
\hline
0&  9& -9    \\
1&  6 &  -6 \\
2&  1 &  -1\\
\hline
\end{tabular}
\end{minipage}
\end{tabular}
\end{table}

  \vspace{-1mm}
\begin{table}[h!]
\caption{BPS invariants for $L7n1^+ $.}
\centering

\label{BPSL7n1}
\begin{tabular}{c}
\begin{minipage}{0.49\linewidth}
\centering
$\hat{N}^{c=0}_{g,Q}:$ \begin{tabular}{|l|lccc| c |}
\hline
\small{g \textbackslash Q}&0&2&4&6\\
\hline
0&  10& -22 & 14&-2    \\
1&  10 & - 14& 4&0 \\
2&  2 & - 2& 0&0 \\
\hline
\end{tabular}
\end{minipage}
\hfill
\hspace{-6mm}
\begin{minipage}{0.49\linewidth}
\centering
$\hat{N}^{c=1}_{g,Q}:$ \begin{tabular}{|l|lcc | c |}
\hline
\small{g \textbackslash Q}&1&3&5\\
\hline
0&  8& -12&4    \\
1&  6 &  -7&1 \\
2&  1 &  -1&0\\
\hline
\end{tabular}
\end{minipage}
\end{tabular}
\end{table}

 \FloatBarrier

    We have confirmed that $\hat{N}^{c=2}=0$ is valid for all the up-to-12-crossings links  that admit HZ factorisability. The only exception is again $L10n42$, which does also not satisfy the HOMFLY--PT-Kauffman relation (\ref{linkcriterion}).  Note in fact that, unlike for knots, the condition $\hat{N}^{c=2}=0$ for 2-component links does not seem to be a direct consequence of (\ref{linkcriterion}). Additionally, it is interesting that although the HZ factorisability depends on the choice of orientation (unless the linking number is $0$),  the vanishing of the two-crosscap BPS invariants is orientation independent. This implies that further investigation is needed in order to understand the intricate inter-relations between HZ factorisability, the HOMFLY--PT-Kauffman relation and the vanishing of the two-crosscap BPS invariants in the case of links, and also to extend them to more than two components. Links are actually of interest from a physics viewpoint, due to the fact that    their Khovanov homology is more clearly related to the space of BPS states in M-theory  \cite{Witten2}.

\section{ Homological relation between Kh and HZ}\label{sec:Kh}
 Khovanov homology  (Kh)  is a knot invariant that can be thought of as  the categorification of  the Jones polynomial. 
 It is a doubly graded (co)homology theory, which is constructed in a way such that the (unnormalised) Jones polynomial takes the role of the
 graded Euler characteristic  \cite{Khovanov,Bar-Natan}.  Both the HOMFLY--PT polynomial and Khovanov homology  are strong knot invariants that generalise the Jones polynomial, but they are sensitive to different aspects of knots and are apriori unrelated to each other. 
In this section,  we establish a connection between Khovanov homology and the exponents of $q$ in the factorised HZ function, 
by associating them to the gradings of the non-trivial homology groups, as explained in the sequel.

The content of  Khovanov homology is contained  in the Khovanov polynomial (the notation here 
 follows\footnote{The subscript $\mathbb{Q}$ indicates that this is a rational homology. We further remark that here $t$ refers to an arbitrary parameter that is not related to $q$, as in the previous sections. 
To avoid confusion, within this section, the Jones polynomial will be expressed only in terms of $q^2$ or $q$.}  \cite{Bar-Natan1}) 
\be\label{eq:KhovanovPoly}
Kh_{\mathbb{Q}}(q,t) = \sum_{j,r} q^j t^r \dim{H^{r,j}}.
\ee
At $t=-1$, this polynomial can be thought of as a generating function for the 
Euler characteristic of the doubly-graded homology group $H^{r,j}$, with respect to one of the gradings.  
In particular, the Euler characteristic $\chi$, labeled as 
$\chi_j:=\chi(j)$  at each fixed $j$, is expressed as the alternating sum over $r$
of the dimensions of $H^{r,j}$, i.e.  $\chi_j=\sum_r (-1)^r \dim H^{r,j}$ (the coefficient of $q^j$).  
The $q$-graded Euler characteristic $Kh_{\mathbb{Q}}(q,t=-1)=\sum_j \chi_j q^j$ is, by construction, the unnormalised Jones polynomial $\bar{J}(q^2)$, 
while its normalised version is related to the Khovanov polynomial 
by 
\be\label{JonesKh}
J(q^2) = \frac{1}{q+q^{-1}} Kh_{\mathbb{Q}}(q,t=-1). 
\ee

The correspondence with the HZ exponents can be directly observed via the Khovanov homology tables, 
which list the dimensions of 
$H^{r,j}$, for each value of the gradings $r$ and $j$. For simple knots, such tables  are available in  \cite{Bar-Natan}, while, more generally, they can be derived
from a graphical
analysis using the Khovanov cube, as explained in \cite{Bar-Natan1}.  Direct comparison of the HZ function with the tables in  \cite{Bar-Natan}, reveals that the HZ exponents for a given knot coincide with the values of $j$ (label of the rows) of its Khovanov table, while their appearance in the HZ numerator or denominator corresponds to  odd or even parity of the grading $r$ (label of the columns), respectively. 

 Below we make this correspondence more explicit through Prop.~\ref{propJonesCoeffpm2} and through  several examples of  knots and links with factorised HZ transform, the Khovanov tables of which are obtained from\footnote{As before (c.f. footnote~\ref{ft:mirror-knot}), in some cases we consider the mirror image of knots as compared to the ones presented in \cite{Bar-Natan}. The difference of their Khovanov tables is a change of sign of the values of $r$ and $j$ and hence, also, in the signature $s$.}
\cite{Bar-Natan}. In these tables, we add an extra column labelled as  "HZ exponents" (or just "HZ"), indicating which numbers appear as exponents in the denominator  and which in the numerator.

\paragraph{Knots.} 
As a first, instructive example, we consider the right handed trefoil  knot $3_1=T(2,3)$, for which the HZ function is 
 \be 
 Z(3_1)= \frac{\lambda(1-q^9 \lambda)}{(1-q\lambda)(1-q^3 \lambda)(1-q^5\lambda)}\ee
and the Kh table is shown in Table~\ref{tab:kh31}.
\begin{table}[h!]
\caption{$ (\dim H^{r,j})$ for $3_{1}$ \cite{Bar-Natan}.}
\centering

\label{tab:kh31}
\begin{tabular}{|l|lccc | c |c|}
\hline
j \textbackslash r&3&2&1&0&$\chi $& HZ exponents\\
\hline
1&  &  & & \textcircled{1}&    1& denominator\\
3&   &  & & \textcircled{1}&  1& \\
5&   & \textcircled{1}& &  &   1&\\
\hline
7&  & $\ast$ &  & &  0 &\\
\hline
9& \fbox{1}&  & & &  -1 &numerator\\
\hline
\end{tabular}
\end{table}

\vskip 1mm
\noindent{}The $j$-value  $9$ (corresponding to the boxed entry  in Table 7) and  $(1,3,5)$ (circled entries), 
coincide with the exponents in the numerator and denominator  of $Z(3_1)$, respectively. The entry indicated by $\ast$ corresponds to the appearance of $\mathbb{Z}_2$ in the integral homology table (labelled as ${\rm dim}\mathcal{G}_{2r+i} KH_{\mathbb{Z}}^r$ in \cite{Bar-Natan}), as we shall explain below.
  Whenever $j$ coincides with an exponent in the numerator or denominator of HZ, their contribution to $\chi$ is equal to $-1$ or $+1$, respectively. As mentioned above, this is due to the different parity of $r$: odd for numerator exponents  and even for denominator exponents.

It is noteworthy that  the table can be split into two parts, with only  denominator or numerator exponents, since there are no overlapping values between them. 
We say that, the Kh table in the factorised cases can be \emph{separated}\footnote{  However, separability of the Kh table does not necessarily imply HZ factorisability. An example is the
 knot $11n_{24}$, whose HZ can be expressed as $Z=((1-\lambda q^9)(1-\lambda q^{-9})-(1-\lambda q^7)(1-\lambda q^{-7})+(1-\lambda q^5)(1-\lambda q^{-5}))/(1-\lambda q^{-3})(1-\lambda q^{-1})(1-\lambda q^{1})(1-\lambda q^3)$ which is clearly not factorisable, but its Kh table can be separated.}.   
Moreover, for simple knots with factorisable HZ transform, most entries in the Khovanov homology tables are   $1$. 
 
The Khovanov polynomial of the trefoil
becomes 
\be\label{KhQ}
Kh_{\mathbb{Q}}(3_1;q,t) = q^{1}+ q^{3} + q^{5} t^2 + q^{9} t^3.
\ee
Substituting $t=-1$, we obtain the generating function of the Euler characteristics $\chi$
\be\label{Kh3_1}
Kh_{\mathbb{Q}}(q,t=-1) =q^{1}+ q^{3}+ q^{5} - q^{9},
\ee
which reflects the above observations about the exponents of $q$ and the signs of the coefficients. 
Dividing with $(q+q^{-1})$ and further substituting  $q^2\mapsto q$, this results in the familiar Jones polynomial of the trefoil knot
\be\label{J31}
J(3_1;q)= - q^{4}+q^{3} + q^{1}.
\ee

The \emph{signature} $s=s(\mathcal{K})$ of a knot,  which is an invariant quantity defined as the difference between
 the number of positive and negative eigenvalues of its Seifert matrix, can also be  obtained from the Khovanov table  \cite{Bar-Natan} by 
\be\label{r,j}
j- 2 r = s+1 \hskip 2mm {\rm or} \hskip 2mm  j-2r=s-1.
\ee
In the example of $3_1$, from Table~\ref{tab:kh31} we observe that for $j=9$, $r=3$, the signature $s$ satisfies $j-2r=3=s+1$, while for  $j=1$,  $r=0$,  $j- 2r = 1 = s-1$, and hence $s(3_1)=2$.

\begin{remark}\label{rmk:altJones}
\textup{The Jones polynomial of the trefoil can alternatively be obtained
from its Khovanov table according to}
\be\label{remark5-2}
J(\mathcal{K};q)= q^{s/2}[ -1 + \sum_{j,r} (-1)^r q^{r} H^{r,j}],
\ee
\textup{which gives (\ref{J31}) for $s=2$. As we shall see in several examples below, (\ref{remark5-2}) is valid whenever the Kh table contains only 2 diagonal lines.} 
\end{remark}
The simple Kh table structure of the trefoil is shared for all 2-stranded torus knots $T(2,n)$. In general, it consists of a region with the rows $j$ corresponding to  the three HZ denominator exponents, and a region with $j$ corresponding to the single  numerator exponent,  with Euler characteristics $\chi=+1$ and $-1$, respectively. Between them there is a row with $\chi=0$, at which  $\ast$ is inserted, that
  connects the two regions via a "$\mathbb{Z}_2$-lego piece". 
  In particular, as mentioned above, the insertion of $\ast$ 
  indicates the presence of the $\mathbb{Z}_2$ torsion for the  integral homology at the specific location \cite{Bar-Natan} 
 (this can be evaluated from the cycles in the Kauffman state \cite{Przytyeki}). 
 The asterix $\ast$, positioned at $(j,r)$, 
 is accompanied by the entries $1$ at $(j-2,r)$ and  at $(j+2,r+1)$, which are  related to each other by a \emph{knight move} \cite{Shumakovitch}, that is, they are located one step up  and one right diagonal downward step as compared to $\ast$, respectively. 
 When grouped together with the $\mathbb{Z}_2$ entry, these yield 
 \be
 \begin{pmatrix}  & {1}\\
&{\mathbb{Z}_2}\\
{1}& \\
\end{pmatrix},
\ee
which is what we refer to as the \emph{$\mathbb{Z}_2$-lego piece}. For  factorised knots,
 the torsion $\mathbb{Z}_2$  almost always appears within the  
 knight move.

It is interesting to observe that the Khovanov table of the HZ-factorised knot $5_2^-$ can be obtained from the one for the torus knot $5_1^-$ by shifting all  $j$ to $j+2$ and then adding 
 a lego piece with the $\mathbb{Z}_2$ positioned at $(j,r)=(-5,-1)$.
 Similarly, the table of $8_{20}$ is obtained from the $5_2^-$ table by again shifting all $j\to j+2$ and adding a lego piece at the position $(1,1)$ (c.f Tables~\ref{table:kh5_2} and~\ref{tab:kh820}); and the one of $10_{125}$ by adding a lego piece at 
 $(7,3)$ to the $8_{20}$ table. This implies that there is a  successive sequence for factorised knots  $5_1^-\to 5_2^-\to 8_{20}\to 10_{125}$ obtained by addition of one $\mathbb{Z}_2$-lego piece. 
 Since these knots are, in fact, members of the pretzel family $P(2,-3,2j+1)$ with $j\geq -1$, and hence they are related by  adding a two-strand full twist on the rightmost tangle, the insertion of such  a lego piece may be related to this operation. 
 
 Moreover, the Kh table of $10_{132}$ ($\iota=4$) can be obtained from that of $5_2$ ($\iota=3$) with the shift of $(j,r)\to (j-2,r\to r-2)$ and the  addition of two lego pieces at $(-3,-1)$ and $(-7,-2)$. 
 The Kh table of $10_{161}$ ($\iota=3$) is obtained from the table of the torus knot $7_1$ ($\iota=2$) by inserting, again, two $\mathbb{Z}_2$-lego pieces, one at (-11,-3)  and one at 
 (-15,-5). 
 The integral homology of $10_{128}$ is similarly obtained  from $7_1$ by adding two lego pieces. 
 Therefore, under addition of two $\mathbb{Z}_2$-lego pieces, the sequences  $5_2 \Rightarrow 10_{132}$,  $7_1\Rightarrow 10_{161}$ and $7_1\Rightarrow 10_{128}$ are obtained. This operation seems to involve an increase in the braid index $\iota$. Finally, the Kh table of $10_{139}$ ($\iota=3$) is obtained from $9_1$ ($\iota=2$) by the shift of a $\mathbb{Z}_2$-lego piece.

\vskip1mm
\noindent{}Some further examples follow.
\vskip 2mm
 {\bf $\bullet$} $5_2^-$:
 $Z(5_2^-)= \frac{\lambda(1-q^{-3} \lambda)(1-q^{-13}\lambda)}{(1-q^{-1}\lambda)(1-q^{-3} \lambda)(1-q^{-5}\lambda)(1-q^{-7}\lambda)}$\\
\begin{table}[hbtp]
\caption{$ (\dim H^{r,j})$ for $5_2^-$.}
\centering

\label{table:kh5_2}
\centering
\begin{tabular}{|l|lccccc | c |c|}
\hline
j \textbackslash r&-5&-4&-3&-2&-1&0&$\chi $& HZ exponents\\
\hline
-1 &   &  &&& &\textcircled {1}&  1 &denominator\\
\hline
-3 &   & &  &   & \fbox{1}& \textcircled{1}&  0& denom. \& num.\\
\hline
-5&  &  &  &\textcircled{1}  & $\ast$  &   &  1&denominator\\
-7&  &  &  &\textcircled{1} $\ast$ &  &   &   1& denominator\\
\hline
-9&  &1 &1 &   &  & &   0 &      \\
-11& &$\ast$  &  &  &  & &    0& \\
\hline
-13&\fbox{1}& &   &   &  & & -1& numerator\\
\hline
\end{tabular}
\end{table}

\noindent{}\noindent{}In the Kh table of $5_2$, shown in  Table~\ref{table:kh5_2},  there  is an overlap of exponents at $j=-3$. Note, however, that due to the factorisability, the corresponding factors in the denominator and numerator cancel each other and the Euler characteristic $\chi$ in the row $j=-3$ of the Kh table~\ref{table:kh5_2} vanishes.  If we ignore such cancellations, we can still say that the Khovanov table is separable. 
The Jones polynomial, given by (\ref{JonesKh}),  becomes
$ J(5_2^-;q)= - q^{-6}+ q^{-5}- q^{-4}+ 2 q^{-3}- q^{-2} + q^{-1} $. This can also be obtained by Remark~\ref{rmk:altJones} with signature  $s(5_2^-)=-2$. 

\vskip 2mm
\newpage
 {\bf $\bullet$} $8_{20}$:
 $Z(8_{20})= \frac{\lambda(1-q^{11} \lambda)(1-q^{-3}\lambda)}{(1-q^{-1}\lambda)(1-q \lambda)(1-q^3\lambda)(1-q^5\lambda)}$
\begin{table}[hbtp]
\caption{$ (\dim H^{r,j})$ for $8_{20}$.}
\centering

\label{tab:kh820}
\centering
\begin{tabular}{|l|lcccccc | c |c|}
\hline
j \textbackslash r&-5&-4&-3&-2&-1&0&1&$\chi $& HZ\\
\hline
3&   &  &&& & &\fbox{1}    & -1& numerator\\
\hline
1&   &  &&& &\textcircled{1} & $\ast$  &1 &\\
-1&   & && &1 &1+\textcircled{1} &     &1 &\\
-3&   & &&\textcircled{1}& $\ast$  &  &     & 1  & denominator\\
-5&  &  &&\textcircled{1}  &   &   &  &  1&\\
\hline
-7&  &1  &1  &  &  &   &&   0& \\
-9&  &$\ast$ & &   &  & &&   0 &      \\
\hline
-11& \fbox{1}&  &  &  &  & &&    -1& numerator \\
\hline
\end{tabular}
\vskip2mm
\end{table}

\vskip 2mm
\noindent{}Its Kh table is shown in Table 9. The Jones polynomial $
J(8_{20}; q) = -q^5+q^4-q^3+ 2 q^2 - q + 2 - q^{-1}$ can also be obtained from Remark~\ref{rmk:altJones} with $s(8_{20})=0$. 
\vskip 2mm

 {\bf $\bullet$} Torus knot $T(3,5)=10_{124}$; $
Z(T(3,5))= \frac{\lambda(1-q^{21}\lambda)(1-q^{19}\lambda)}{(1-q^{13}\lambda)(1-q^{11}\lambda)(1-q^{9}\lambda)(1-q^7\lambda)}
$
\setlength{\tabcolsep}{4pt} 
\renewcommand{\arraystretch}{0.8}
\begin{table}[hbtp]
\caption{$ (\dim H^{r,j})$ for $T(3,5)=10_{124}$.}
\centering

\label{tab:kht35}
\centering
\begin{tabular}{|l|lccccccc | c |c|}
\hline
j \textbackslash r&0& 1&2&3&4&5&6&7&$\chi $& HZ\\
\hline
21&   &  &&& & & &\fbox{1}    & -1& numerator\\
19&   &  &&& &\fbox{1}& & $\ast$  &-1 & \\
\hline
17&   &  &&  &  &1 &$\underline{1}$ &     &0 &\\
15&   &  &&1& 1&  &   &  &0& \\
\hline
13&  &  && $\ast$ &  \textcircled{\underline{1}}&   &  &  & 1&\\
11&  &  &\textcircled{1}&  &  &   &&   &1& \\
9& \textcircled{1} & & &   &  & & &&   1 &  denominator    \\
7& \textcircled{1} & & &   &  & & && 1& \\
\hline
\end{tabular}
\end{table}
\vskip2mm

\noindent{}The Khovanov homology is shown in Table~\ref{tab:kht35}. The Jones polynomial $J(10_{124};q)=-q^{10}+q^6+q^4$, can be obtained from $\sum q^j \chi_j= q^7+q^9+q^{11}-q^{19}-q^{21}$ after dividing by $(q+q^{-1})$ and replacing $q^2\to q$. The signature is $s(10_{124})=8$. However, Remark~\ref{rmk:altJones} does not hold in this case, since there more than $2$ diagonal lines 
in the Kh table. The entries belonging to the third diagonal are underlined (these correspond to the entries colored by red in \cite{Bar-Natan}).

It is remarkable that the entries in the homological tables of torus knots $T(3,n)$ have a (roughly) reflection symmetric structure about the main diagonal. 
This excludes the entries of the form $\begin{pmatrix}1\\1\end{pmatrix}$ appearing in the $r=4 k$ or $r=0$ column, when $n= 3 k + 1$ or  $n= 3 k  + 2 $, respectively. In the latter case, this can be seen in Table~\ref{tab:kht35} for $n=5$.

\vskip 2mm
 {\bf $\bullet$} Torus knot $T(4,7)$:
$
Z(T(7,4))= \frac{\lambda(1-q^{37}\lambda)(1-q^{35}\lambda)(1-q^{33}\lambda)}{(1-q^{25}\lambda)(1-q^{23}\lambda)(1-q^{21}\lambda)(1-q^{19}\lambda)(1-q^{17}\lambda)}
$\\
\setlength{\tabcolsep}{4pt} 
\renewcommand{\arraystretch}{0.8}
\begin{table}[hbtp]
\caption{$ (\dim H^{r,j})$ for $T(4,7)$.}
\centering

\label{tab:kht47}
\begin{tabular}{|l|lccccccccccccc | c |c|}
\hline
j \textbackslash r&0& 1&2&3&4&5&6&7&8&9&10&11&12&13&$\chi $& HZ\\
\hline
39&   &  &  &  & & & &  & & & & & 1& 1 & 0&\\
\hline
37&   &  &  & & & & &  & &  & & \fbox{1}&$\ast$  &$\ast \ast$   & -1&\\
35&   &  & &  & & & &  &  & &$\ast$ & \fbox{1}+1$\ast$&$\underline{1}$ &  &-1& num.\\
33&   &  &  &  & &  & &  & &\fbox{1}+1 &1$\ast$ &$\ast$ &   & & -1 &\\
\hline
31&  &   &  &  &  &   & & $\underline{1}$& &$\ast\oplus\ast\ast$ & $\underline{1}$&  &  &  &0     &\\
29&  &  &   &  &  & $\underline{1}$ & &   1$\ast$& 2&  &  &  & & &0 &\\
27&  &  &   &  &  &   $\underline{1}$& 1& $\ast$  &   &   & &  & &  & 0 &    \\
\hline
25&  & & &  $\underline{1}$ & $\underline{1}$ & & \textcircled{1}& & & & & &  && 1& \\
23&  & & & $\ast$    & \textcircled{1} & & & & & & & & && 1& \\
21& & & \textcircled{\underline{1}}&     &  & & & & & & & & && 1& denom.\\
19&\textcircled{\underline{1}} & & &     &  & & & & & & & & && 1& \\
17&\textcircled{\underline{1}} & & &     &  & & & & & & & & && 1& \\
\hline
\end{tabular}
\vskip2mm
\end{table}

\noindent{}In Table~\ref{tab:kht47} the notation $\ast$ and $\ast\ast$ is used to denote the presence of $\mathbb{Z}_2$ and $\mathbb{Z}_4$, respectively, in the integral homology. The Jones polynomial is $J(q) = q^9+q^{11}+q^{13} - q^{14}+q^{15} -q^{16}-q^{18}$. The signature of $T(4,7)$ is $s=14$, but  again Remark~\ref{rmk:altJones} does not hold since there are five diagonals in the Kh table (there are 10 off-diagonal entries, which are underlined  in Table~\ref{tab:kht47}).
The configuration of entries in the last 5 rows in the Kh table of $T(4,n)$, for all  odd $n$, contains a fixed pattern,  
which includes 
one $\mathbb{Z}_2 $, indicated as $(\ast)$, and the entries \textcircled{1} at the values of $j$ corresponding to denominator exponents and at even values of $r$.

   \begin{remark}\label{rmk5-1}
\textup{The sum of Euler characteristics $\chi$ over all $j$ is always equal to $2$. This is due to the fact that the Jones polynomial for all knots satisfies
$J(q=1)=1$ (c.f. Prop.~\ref{prop:Zq1}), while  the value $2$ comes form  the normalisation factor  
$\left (q+q^{-1})\right\vert_{q=1}=2$.}
\end{remark}

\begin{prop}\label{propJonesCoeffpm2}
     When a knot $\mathcal{K}$ is HZ-factorisable,
     i.e. $Z(\mathcal{K})$ is of the form (\ref{eq:factorisableZ}), then its graded Euler characteristic, or unnormalised Jones polynomial, can be written as 
     \be\label{JonesfromHZ}
     Kh_{\mathbb{Q}}(\mathcal{K};q,-1)=\bar{J}(\mathcal{K};q^2)=-\sum_{i=0}^{m-2}q^{\alpha_i}+\sum_{i=0}^m q^{\beta_i}. 
     \ee
    
\end{prop}

\begin{proof} 
    This can be easily proven by applying Cor.~\ref{corInverseJones} (excluding the normalisation factor $(q+q^{-1})^{-1}$ in (\ref{JonesContour})) to the general expression of the factorised HZ transform (\ref{eq:factorisableZ}). Explicitly,  $\bar{J}(q^2)=Res_{\lambda=0}(\lambda^{-3}Z)=\lim_{\lambda\to0}\frac{d}{d\lambda}\frac{\prod_{i=0}^{m-2}(1-\lambda q^{\alpha_i})}{\prod_{i=0}^{m}(1-\lambda q^{\beta_i})}=-\sum_{i=0}^{m-2}q^{\alpha_i}+\sum_{i=0}^m q^{\beta_i}$, where $\beta_i=e+2i$.
\end{proof}

Since $Kh_{\mathbb{Q}}(\mathcal{K};q,-1)=\sum_j\chi_j q^j$,
     (\ref{JonesfromHZ}) indeed implies that the Euler characteristic $\chi_j$ in the Khovanov homology of $\mathcal{K}$ is equal to  $ +1$ or $-1$ when $j$ coincides with an HZ  exponent in the denominator or numerator, respectively. Note that after dividing with $q+q^{-1}$ to obtain the normalised Jones polynomial, 
 the coefficients are often again just $\pm1$ (at least for simple knots), but there are exceptions such as $J(10_{128};q)=-q^{10}+q^9-2 q^8+2 q^7-q^6+2 q^5-q^4+q^3$  and 
$
J(8_{20}\otimes 10_{161};q) =
q^4-q^5+ 2 q^6 - q^7 + q^8 - q^9 + q^{10} -q^{11}+q^{12}-q^{13}+q^{14}-q^{15}$). Moreover, it should be pointed out that when the Jones polynomial is of the form (\ref{JonesfromHZ}), it does not necessarily imply HZ factorisability.

\begin{remark}
    \textup{
    Proposition~\ref{propJonesCoeffpm2} implies that in the factorised cases, $\bar{J}(\mathcal{K})$   essentially contains the same information as $Z(\mathcal{K})$ and hence it is equivalent to the HOMFLY--PT polynomial. Furthermore, if the Khovanov homology is restricted to the 2 main diagonals, the Khovanov table can be reconstructed by $\bar{J}$ plus the signature $s$. 
    Since the Alexander polynomial can be obtained from $Z(\mathcal{K})$ via Theorem~\ref{thmInverseAlex}, this implies that, for such knots, the Alexander polynomial can be related to the Jones polynomial and its categorification.}
    \end{remark}

It is noteworthy that neither Khovanov homology nor the HOMFLY--PT polynomial are stronger to each other. Examples that show this are  the pairs of knots $8_9$ and $12n_{462}$ or $8_8$ and $10_{129}$, which have the same Khovanov homology but are distinguished by the HOMFLY--PT polynomial. On the other hand, $Kh(5_1)\neq Kh(10_{132})$, while $H(5_1)= H(10_{132})$, but they have different signatures as $s(5_1)=4$ and $s(10_{132})=0$.

\paragraph{Links.}
In Sec.~\ref{sec:factorisedHZ}, we found that several two-component links also admit a factorisable HZ and hence it is natural to also consider their relations to Khovanov homology. 
An explicit example belonging to the family of 2-stranded torus links  $T(2,2k)$, is  $L6a3\{0\}=T(2,6)$, for which
\be\label{T(2,6)}
Z(L6a3\{0\}) = \frac{ \lambda(1+q^{18}\lambda)}{(1-q^4 \lambda)(1-q^6 \lambda)(1-q^8 \lambda)}
\ee
and the Kh table is shown in Table~\ref{tab:L5a3}.
\begin{table}[hbtp]
\caption{$ (\dim H^{r,j})$ for $L6a3\{0\}$.}
\centering

\label{tab:L5a3}
\begin{tabular}{|l|lcccccc|c|c|}
\hline
$j$ \textbackslash $r$&6&5&4&3&2&1&0&$\chi $&HZ\\
\hline
4&    & & & & &     &\textcircled{1}  &1&\\
6 &     & & & & &   & \textcircled{1}&1& denom.\\
8&     & & &  &\textcircled{1}&   &  &1&\\
\hline
10&     & &  &&$\ast$&    & & 0&\\
12&     &  &1&1& &    & & 0&\\
14&    &  &$\ast$&&  &    & & 0&\\
16&  1  &1&& &  &    & &  0& \\
\hline
18&  \fbox{1} && & &  &   & & 1&num.\\ 
\hline
\end{tabular}
\end{table}

\noindent{}The Euler characteristics  $\chi$ at each $j$ in the Kh homology table can again be determined by the exponents of the HZ transform, but now it is equal to $+1$ for $j$'s corresponding to both denominator and numerator exponents. This is related to the positive sign in the numerator of (\ref{T(2,6)}).

\begin{prop}\label{propJonesLinkCoeff1}
 When a 2-component link $\mathcal{L}$  has a factorised HZ transform  of its HOMFLY--PT polynomial with $m\leq 4$, i.e. $Z(\mathcal{L})$ is of the form (\ref{T(2,2j)}), (\ref{eq:factorisableZZ}), or (\ref{T4n21}), then the graded Euler characteristic, or unnormalised Jones polynomial, can be written as 
     \be
     Kh_{\mathbb{Q}}(\mathcal{L};q,-1)=\bar{J}(\mathcal{L};q^2)=q^{\al_0}-\sum_{i=1}^{m-2}q^{\alpha_i}+\sum_{i=0}^m q^{\beta_i}. \ee
\end{prop}

\begin{proof}
    The proof is very similar to that of Proposition~\ref{propJonesCoeffpm2}.
\end{proof}

The Jones polynomial of $L6a3\{0\}$ can also be obtained by the alternating sum of 
${\rm dim}H^{r,j}$ (with signature $s=5$)
\ba
J(L6a3)& =& q^{s/2} \sum_j \sum_r (-1)^r {\rm dim} H^{r,j}\nonumber\\
&=& -q^{9/2}-q^{5/2} -q^{17/2}+q^{15/2}-q^{13/2}+q^{11/2}.
\ea

\begin{remark}\label{rmk:sumofEuler}
   \textup{ The sum of all Euler characteristics $\chi$ for a link 
with $l$ components satisfies $\sum_j \chi_j=2^l$. This is the same as the value  of their Jones polynomial at $q=1$, which is $2^{l-1}$ (c.f. Prop.~\ref{prop:Zq1}), multiplied by a factor $2=q+q^{-1}|_{q=1}$.
This is consistent with the sum of Euler characteristics for knots ($l=1$), which, as mentioned in Remark~\ref{rmk5-1}, it is always equal to $2$.}
\end{remark}

As a generalisation of the $sl(2)$ case in the above remark, which corresponds to Khovanov homology, we suggest that the sum of  Euler characteristics corresponding to an $sl(N)$ homology is equal to $N^l$, which is the value of the unnormalised HOMFLY--PT polynomial at $q\to 1$, as given in Prop.~\ref{prop:Zq1}.

\section{ Summary and discussions\label{Sec:summary} }

Extending the previous study \cite{Petrou}, we found  new infinite families of hyperbolic knots and links whose HOMFLY--PT polynomial has a factorised HZ transform. We prescribed methods for generating them involving full twists $F_m$ and  Jucys-Murphy twists $E_m$. Moreover, we have derived a closed formula (\ref{FactorisedHOMFLY}) for the HOMFLY--PT polynomial for such families, obtained via inverse HZ transform, that can be expressed in terms 
 of the number of strands $m$, and the HZ numerator and denominator exponents, which depend on the number of additional twists $F_m$ and $E_m$.

   In Theorem~\ref{thm:}  we have proven that   relation (\ref{KFandH1}),  which was known to  hold only between the Kauffman and HOMFLY--PT polynomials of torus knots, is also valid for the newly discovered hyperbolic families. This motivates the Conjecture~\ref{HomflyANDKauff}, which suggests that such a HOMFLY--PT/Kauffman relation is in 1-1 correspondence with HZ factorisability.   We further showed in (\ref{BPSandHKF}) that 
   relation (\ref{KFandH1})  is equivalent to the vanishing of the number $\hat{N}_{g,Q}^{c=2}$ of  BPS states  with two-crosscaps, in topological string theory. In the case of two-component links, a similar relation between the HOMFLY--PT and  Kauffman polynomials was proposed in (\ref{linkcriterion}) which, again, seems to correspond to the property $\hat{N}_{g,Q}^{c=2}=0$. 
   A more thorough investigation is required, however, to fully understand how the latter is implied by the 
   HOMFLY--PT/Kauffman relation 
   in the case of links.

The  vanishing of the 2-crosscap BPS degeneracy 
for these HZ-factorisable 
 knots and links 
 is a peculiar property whose physical interpretation remains largely  mysterious. 
It shows a relation between knots, links and supersymmetric gauge theory, similar to the relation between the HOMPLY-PT polynomial and Wilson loops in CS theory. 
Speculatively, it may be thought of as a supersymmetry manifesting in the HZ transform, 
viewed  as a Hilbert or Poincar\'e 
 series. Further investigation may provide the connection to supermatrices, which are realised  in the intersection numbers of a moduli space \cite{ BrezinHikamianti,Hikami,HikamiBrezinII,Hikami6}.

The relation between BPS states and link homology is also attractive. In \cite{Witten2},
  a new interpretation of Khovanov homology via a time-dependent gauge theory  is suggested; 
  while, in \cite{Gukov}, it is conjectured that the $\mathfrak{sl}(N)$ homology groups, which form the Khovanov-Rozansky homology, 
can be interpreted in terms of the BPS invariants of oriented topological strings. 
The latter refer to a refinement, with respect to an additional grading, 
of the integers $\hat{N}_{g,Q}$ introduced in Sec.~\ref{sub:BPS}.  Such a conjecture is based on the fact that the BPS spectrum $\hat{N}_{g,Q}$  can be obtained as the Euler characteristics of the cohomology of the moduli space of holomorphic curves of genus g, that have a  boundary ending on the Lagrangian submanifold associated to the knot.

The HZ transform, which was originally employed to compute the Euler characteristics of the moduli space of Riemann surfaces, maps the information contained in the integers $\hat{N}_{g,Q}$, into the sets of numerator and denominator exponents. In Sec.~\ref{sec:Kh} we have shown that when the HZ transform is factorisable, these exponents can be extracted from just the unnormalised Jones polynomial (c.f. Prop~\ref{propJonesCoeffpm2})  and their values are associated to the double-grading of the $\mathfrak{sl}(2)$ homology. 
Although the construction of a Kauffman link homology is a hard problem that remains open, it would be an interesting future task to investigate a homological  interpretation of the unoriented BPS invariants  $\hat{N}_{g,Q}^{c}$ ($c\neq 0$) in this simplified scenario, when a relation  between the $SU(N)$ 
and the $SO(N+1)$ invariants exists.

 \paragraph{Declaration of competing interest.} The authors have no competing interest to disclose.

\paragraph{Acknowledgements.}
We are grateful to Andrew Lobb for several useful discussions
and to Michael Willis for pointing out the Jucys--Murphy braid  to us. We also thank Alexei and Andrey Morozov for providing us with their knot data tables. We are grateful to Louis Kauffman for his insightful comments about knot polynomials.
We acknowledge   Dror Bar-Natan for his several suggestive lectures, including the one  at the  University of Tokyo-OIST Symposium in September 2023. We also thank Stavros Garoufalidis for  suggesting several computer softwares for knot invariants. We are grateful to Reiko Toriumi and her unit members for constructive feedback and especially Cihan Pazarbaşı for assisting with Mathematica coding. 
This work is supported by OIST funding  and by JSPS Kakenhi-19H01813.  

\newpage
\section*{Appendix. BPS invariants for links admitting HZ factorisation }

The BPS invariants were introduced in Sec.~\ref{sub:BPS}. Here we present the results for $\hat{N}^{c=0}_{g,Q}$, $\hat{N}^{c=1}_{g,Q}$ and $\hat{N}^{c=2}_{g,Q}$, obtained via (\ref{hath}), (\ref{c1general}) and (\ref{c2general}), respectively, for the links that admit HZ factorisability. In all cases, with the only exception being $L10n42$, we find $\hat{N}^{c=2}_{g,Q}=0$.
\vskip2mm
 \noindent{}$\bullet$ $T(2,8)=L8a14=(\bigcirc,\bigcirc)$; $lk(L8a14\{0\})=4$

 \vspace{-2mm}
\begin{table}[h!]
\begin{tabular}{c}
\begin{minipage}{0.49\linewidth}
\centering
$\hat{N}^{c=0}_{g,Q}:$
\begin{tabular}{|l|lccccc | c |}
\hline
\small{g \textbackslash Q}&-2&0&2&4&6&8\\
\hline
0&  20& -34 & 12&0&0&2    \\
1&  30 & - 40& 10&0&0&0 \\
2&  14 & - 16& 2&0&0&0 \\
3&  2 & - 2& 0&0&0&0 \\
\hline

\hline
\end{tabular}
\end{minipage}
\hfill
\hspace{-6mm}
\begin{minipage}{0.49\linewidth}
$\hat{N}^{c=1}_{g,Q}:$
\centering
\begin{tabular}{|l|lc | c |}
\hline
\small{g \textbackslash Q}&-1&1\\
\hline
0&  16& -16    \\
1&  20 &  -20 \\
2&  8 &  -8\\
3&  1 &  -1\\
\hline
\end{tabular}
\end{minipage}
\end{tabular}
\end{table}

\noindent{}$\bullet$ $T(2,10)=L10a118=(\bigcirc,\bigcirc)$; $lk(L10a118\{0\})=5$ 
 \vspace{-2mm}
\begin{table}[h!]
\begin{tabular}{c}
\begin{minipage}{0.54\linewidth}
\centering
$\hat{N}^{c=0}_{g,Q}:$
\begin{tabular}{|l|lcccccc | c |}
\hline
\small{g \textbackslash Q}&-2&0&2&4&6&8&10\\
\hline
0&  30& -52 & 20&0&0&0&2    \\
1&  70 & - 100& 30&0&0&0&0 \\
2&  56 & - 70& 14&0&0&0&0 \\
3&  18 & - 20& 2&0&0&0 &0\\
4&  2 & - 2& 0&0&0&0 &0\\
\hline

\hline
\end{tabular}
\end{minipage}
\hfill
\hspace{-6mm}
\begin{minipage}{0.45\linewidth}
$\hat{N}^{c=1}_{g,Q}:$
\centering
\begin{tabular}{|l|lc | c |}
\hline
\small{g \textbackslash Q}&-1&1\\
\hline
0&  25& -25   \\
1&  50 &  -50 \\
2&  35 &  -35\\
3&  10 &  -10\\
4&  1 &  -1\\
\hline
\end{tabular}
\end{minipage}
\end{tabular}
\end{table}

From the above two cases and also from the results for $T(2,4)$ and $T(2,6)$ given in (\ref{t24c0}) and   in Table~\ref{BPSL6a3}, we can observe that the torus links $T(2,2k)$ follow the pattern shown below, in which the sign of non-zero entries is indicated. 
 \vspace{-2mm}
\begin{table}[h!]
\begin{tabular}{c}
\begin{minipage}{0.49\linewidth}
\centering
$\hat{N}^{c=0}_{g,Q}:$
\begin{tabular}{|l|lccccc | c |}
\hline
\small{g \textbackslash Q}&-2&0&2&4&...&2k\\
\hline
0&  +& - & +&0&...&2    \\
1&  +& - & +&0&...&0 \\
\vdots&   &&\vdots &&& \\
k-2&  +& - & +&0&...&0 \\
k-1&  2 & - 2& 0&0&...&0 \\
\hline

\hline
\end{tabular}
\end{minipage}
\hfill
\hspace{-6mm}
\begin{minipage}{0.49\linewidth}
$\hat{N}^{c=1}_{g,Q}:$
\centering
\begin{tabular}{|l|lc | c |}
\hline
\small{g \textbackslash Q}&-1&1\\
\hline
0&  +& - \\
\vdots&  \vdots & \vdots  \\
k-3&  +&  -\\
k-2&  2k &  -2k\\
k-1&  1 &  -1\\
\hline
\end{tabular}
\end{minipage}
\end{tabular}
\end{table} 
 \noindent{}$\bullet$ $L7n2=(3_1^-,\bigcirc)$; $lk(L7n2)=0$ 
 \vspace{-1mm}
\begin{table}[h!]
\begin{tabular}{c}
\begin{minipage}{0.49\linewidth}
\centering
$\hat{N}^{c=0}_{g,Q}:$
\begin{tabular}{|l|lccc | c |}
\hline
\small{g \textbackslash Q}&-6&-4&-2&0\\
\hline
0&  -4& 12 & -12& 4    \\
1&  0 &  4& -4&0 \\
\hline

\hline
\end{tabular}
\end{minipage}
\hfill
\hspace{-6mm}
\begin{minipage}{0.49\linewidth}
$\hat{N}^{c=1}_{g,Q}:$
\centering
\begin{tabular}{|l|lccc | c |}
\hline
\small{g \textbackslash Q}&-7&-5&-3&-1\\
\hline
0&  3& -9 & 9& -3    \\
1&  1 &  -6& 6& -1 \\
2&  0 &  -1& 1&0 \\
\hline
\end{tabular}
\end{minipage}
\end{tabular}
\end{table}

 \noindent{}$\bullet$ $L9n12=(5_1^+,\bigcirc)$;  $lk(L9n12\{1\})=2$ 
 \vspace{-1mm}
\begin{table}[h!]
\begin{tabular}{c}
\begin{minipage}{0.49\linewidth}
\centering
$\hat{N}^{c=0}_{g,Q}:$
\begin{tabular}{|l|lcccc | c |}
\hline
\small{g \textbackslash Q}&2&4&6&8&10\\
\hline
0&  16& -38 & 26& -2&-2    \\
1&  28 &  -50& 22& 0&0 \\
2&  14 & - 18& 4& 0&0 \\
3&  2 &  -2& 0& 0&0 \\
\hline

\hline
\end{tabular}
\end{minipage}
\hfill
\hspace{-4mm}
\begin{minipage}{0.49\linewidth}
$\hat{N}^{c=1}_{g,Q}:$
\centering
\begin{tabular}{|l|lcccc | c |}
\hline
\small{g \textbackslash Q}&3&5&7&9&11\\
\hline
0&  15& -26 & 8& 6&-3    \\
1&  20 &  -30& 6& 5&-1 \\
2&  8 &  -10& 1& 1&0 \\
3&  1 &  -1& 0& 0&0 \\
\hline
\end{tabular}
\end{minipage}
\end{tabular}
\end{table}

\noindent{}$\bullet$ $L9n14=(3_1^+,\bigcirc)$; $lk(L9n14\{0\})=-1$ 
 \vspace{-2mm}
\begin{table}[h!]
\begin{tabular}{c}
\begin{minipage}{0.49\linewidth}
\centering
$\hat{N}^{c=0}_{g,Q}:$
\begin{tabular}{|l|lcccc| c |}
\hline
\small{g \textbackslash Q}&-2&0&2&4&6\\
\hline
0&  2& 4 & -22& 24&-8    \\
1& 0&2& -14 &  14& -2 \\
2&  0&0 & - 2& 2& 0 \\
\hline

\hline
\end{tabular}
\end{minipage}
\hfill
\hspace{-4mm}
\begin{minipage}{0.49\linewidth}
$\hat{N}^{c=1}_{g,Q}:$
\centering
\begin{tabular}{|l|lccc | c |}
\hline
\small{g \textbackslash Q}&1&3&5&7\\
\hline
0&  8& -21 & 19&-6    \\
1&  6&-22 &  21& -5 \\
2& 1&- 8 &  8&- 1 \\
3& 0& -1 &  1& 0 \\
\hline
\end{tabular}
\end{minipage}
\end{tabular}
\end{table}
\newpage

\noindent{}$\bullet$ $L9n15^+=T(3,4,2,1)=(3_1^+,\bigcirc)$; $lk(L9n15\{0\})=3$ 
 \vspace{-2mm}
\begin{table}[h!]
\begin{tabular}{c}
\begin{minipage}{0.49\linewidth}
\centering
$\hat{N}^{c=0}_{g,Q}:$
\begin{tabular}{|l|lccc| c |}
\hline
\small{g \textbackslash Q}&0&2&4&6\\
\hline
0&  20& -42 & 24& -2    \\
1&  30 &  -44& 14& 0 \\
2&  14 & - 16& 2& 0 \\
3&  2 &  -2& 0& 0 \\
\hline

\hline
\end{tabular}
\end{minipage}
\hfill
\hspace{-8mm}
\begin{minipage}{0.49\linewidth}
$\hat{N}^{c=1}_{g,Q}:$
\centering
\begin{tabular}{|l|lcc | c |}
\hline
\small{g \textbackslash Q}&1&3&5\\
\hline
0&  18& -27 & 9    \\
1&  21 &  -27& 6 \\
2&  8 &  -9& 1 \\
3&  1 &  -1& 0 \\
\hline
\end{tabular}
\end{minipage}
\end{tabular}
\end{table}

\noindent{}$\bullet$ $L10n42=(3_1^+,\bigcirc)$; $lk(L10n42\{1\})=-1$ 

 \vspace{-2mm}
\begin{table}[h!]
\begin{tabular}{c}
\begin{minipage}{0.49\linewidth}
\centering
$\hat{N}^{c=0}_{g,Q}:$
\begin{tabular}{|l|lcccc| c |}
\hline
\small{g \textbackslash Q}&0&2&4&6&8\\
\hline
0&  14& -42 & 44& -18&2    \\
1&  12 &  -44& 44& -12&0 \\
2&  2 & - 16& 16& -2&0 \\
3& 0& -2 &  2& 0& 0 \\
\hline

\hline
\end{tabular}
\end{minipage}
\hfill
\hspace{-8mm}
\begin{minipage}{0.49\linewidth}
$\hat{N}^{c=1}_{g,Q}:$
\centering
\begin{tabular}{|l|lccc | c |}
\hline
\small{g \textbackslash Q}&1&3&5&7\\
\hline
0&  8& -21 & 19&-6    \\
1&  6 &  -22& 21&-5 \\
2& 1&- 8 &  8&- 1 \\
3& 0& -1 &  1& 0 \\
\hline
\end{tabular}
\end{minipage}
\end{tabular}
\end{table}

\begin{table}[h!]
$\hat{N}^{c=2}_{g,Q}:$
\centering
\begin{tabular}{|l|lccc | c |}
\hline
\small{g \textbackslash Q}&0&2&4&6\\
\hline
0&  5& -15 & 15&-5    \\
1&  5 &  -20& 20&-5 \\
2& 1&- 8 &  8&- 1 \\
3& 0& -1 &  1& 0 \\
\hline
\end{tabular}
\end{table}

\vskip2mm
\noindent{}$\bullet$  $L11n132^+=(5_1^+,\bigcirc)$; $lk(L11n132\{0\})=1$ 
  \vspace{-2mm}
\begin{table}[h!]
\begin{tabular}{c}
\begin{minipage}{0.49\linewidth}
\centering
$\hat{N}^{c=0}_{g,Q}:$
\begin{tabular}{|l|lcccc | c |}
\hline
\small{g \textbackslash Q}&2&4&6&8&10\\
\hline
0&  12& -32 & 26& -4&-2    \\
1&  10 &  -30& 22& -2&0 \\
2&  2& - 6& 4& 0&0 \\
\hline

\hline
\end{tabular}
\end{minipage}
\hfill
\hspace{-4mm}
\begin{minipage}{0.49\linewidth}
$\hat{N}^{c=1}_{g,Q}:$
\centering
\begin{tabular}{|l|lcccc | c |}
\hline
\small{g \textbackslash Q}&3&5&7&9&11\\
\hline
0&  11& -23 & 8& 10&-6    \\
1&  10 &  -26& 6& 15&-5 \\
2&  2 &  -9& 1& 7&-1 \\
3&  0 &  -1& 0& 1&0 \\
\hline
\end{tabular}
\end{minipage}
\end{tabular}
\end{table}


\noindent{}$\bullet$ $L11n133^+=(5_1^+,\bigcirc)$; $lk(L11n133\{0\})=3$ 
 \vspace{-2mm}
\begin{table}[h!]
\begin{tabular}{c}
\begin{minipage}{0.49\linewidth}
\centering
$\hat{N}^{c=0}_{g,Q}:$
\begin{tabular}{|l|lcccc| c |}
\hline
\small{g \textbackslash Q}&2&4&6&8&10\\
\hline
0&  28&  -64& 42& -4&-2   \\
1&  70& -118 & 50& -2&0  \\
2&  56 & - 74& 18& 0&0 \\
3&  18 &  -20& 2&0& 0 \\
4&  2 &  -2& 0&0& 0 \\
\hline

\hline
\end{tabular}
\end{minipage}
\hfill
\begin{minipage}{0.49\linewidth}
$\hat{N}^{c=1}_{g,Q}:$
\centering
\begin{tabular}{|l|lcccc| c |}
\hline
\small{g \textbackslash Q}&3&5&7&9&11\\
\hline
0&   27&-44 & 15&3&-1    \\
1&  54 &  -75& 20&1&0 \\
2&  36 &  -44& 8&0&0 \\
3&  10 &  -11& 1&0&0 \\
4&  1 &  -5& 0&0&0 \\
\hline
\end{tabular}
\end{minipage}
\end{tabular}
\end{table}

\FloatBarrier

 \noindent{}$\bullet$ $L11n204^+=T(3,5,2,1)=(3_1^+,\bigcirc)$; $lk(L11n204\{0\})=4$
    \vspace{-1mm}
\begin{table}[h!]
\begin{tabular}{c}
\begin{minipage}{0.51\linewidth}
\centering
$\hat{N}^{c=0}_{g,Q}:$
\begin{tabular}{|l|lcccc | c |}
\hline
\small{g \textbackslash Q}&0&2&4&6&8\\
\hline
0&  34& -72 & 44& -8&2    \\
1&  72 &  -11& 44&-1&0 \\
2&  56 & -72& 16& 0&0 \\
3&  18 &  -20& 2&0&0 \\
4&  2 &  -2&0&0&0 \\
\hline

\hline
\end{tabular}
\end{minipage}
\hfill
\begin{minipage}{0.47\linewidth}
$\hat{N}^{c=1}_{g,Q}:$
\centering
\begin{tabular}{|l|lccc | c |}
\hline
\small{g \textbackslash Q}&1&3&5&7\\
\hline
0&  33& -51 & 19& -1    \\
1&  56 &  -77& 21& 0 \\
2&  36 &  -44& 8& 0 \\
3&  10 &  -11& 1& 0 \\
4&  1 &  -1& 0& 0 \\
\hline
\end{tabular}
\end{minipage}
\end{tabular}
\end{table}
\FloatBarrier
\noindent{}Observe that the BPS invariants for the family $T(3,n,2,1)$ (c.f. $L9n15$ above and $L7n1$ in Table~\ref{BPSL7n1}) again follow a certain pattern. This involves an alternating sign in each column and the fixed entries of $\pm2$ or $\pm1$ at the bottom left corner for $c=0$ and $c=1$, respectively.
\vskip1mm
 \noindent{}$\bullet$ $L11n208=(3_1^+,\bigcirc)$; $lk(L11n208\{0\})=-2$ 
     \vspace{-1mm}
\begin{table}[h!]
\begin{tabular}{c}
\begin{minipage}{0.51\linewidth}
\centering
$\hat{N}^{c=0}_{g,Q}:$
\begin{tabular}{|l|lccccc | c |}
\hline
\small{g \textbackslash Q}&-4&-2&0&2&4&6\\
\hline
0&  2& 0 & 10& -42&44&-14    \\
1&  0 &  0& 10& -44&44&-10 \\
2&  0 & 0& 2& -16&16&-2 \\
3&  0 &  0& 0& -2&2&0 \\
\hline

\hline
\end{tabular}
\end{minipage}
\hfill
\begin{minipage}{0.47\linewidth}
$\hat{N}^{c=1}_{g,Q}:$
\centering
\begin{tabular}{|l|lccc | c |}
\hline
\small{g \textbackslash Q}&1&3&5&7\\
\hline
0&  18& -42 & 34& -10    \\
1&  21 &  -62& 56& -15 \\
2&  8 &  -37& 36& -7 \\
3&  1 &  -10& 10& -1 \\
4&  0 &  -1& 1& 0 \\
\hline
\end{tabular}
\end{minipage}
\end{tabular}
\end{table}

\end{document}